\pdfoutput=1

\documentclass[sigconf, nonacm]{acmart}

\usepackage{preamble}

\begin{document}

\newtheorem{remark}{Remark}
\newtheorem{todo}{Todo}
\newtheorem{question}{Question}

\newcommand{\tx}{\textsf{tx}}
\newcommand{\Ledger}{\mathbf{L}}
\newcommand{\LState}{\mathsf{L}}
\newcommand{\negl}{\textsf{negl}(\kappa)}
\newcommand{\view}{\textsc{view}}
\newcommand{\fullview}{\textsc{view}^{t,n}_{\Pi,\mathcal{A},\mathcal{Z}}}
\newcommand{\validityattack}{\textsc{persistence}\textsc{-}\textsc{game}}
\newcommand{\livenessattack}{\textsc{liveness}\textsc{-}\textsc{game}}
\newcommand{\true}{\mathsf{true}}
\newcommand{\false}{\mathsf{false}}
\newcommand{\approxn}{\widetilde{n}}
\newcommand{\dexpand}[1]{#1_{\pm\Delta}}
\newcommand{\validator}{\ensuremath{\mathsf{vl}}}
\newcommand{\valid}{valid\xspace}
\newcommand{\invalid}{invalid\xspace}

\newcommand{\ckpt}{\ensuremath{\mathsf{ts}}}
\newcommand{\powchain}{\ensuremath{\mathcal{B}}}
\newcommand{\TT}{\ensuremath{\mathcal{T}}}

\newcommand{\pk}{\mathsf{pk}}
\newcommand{\sk}{\mathsf{sk}}
\newcommand{\pr}{\textit{R}}
\newcommand{\sr}{\textit{k}}

\newcommand{\concat}{\,\|\,}
\newcommand{\eg}{\emph{e.g.}\xspace}
\newcommand{\ie}[0]{\emph{i.e.}\xspace}
\newcommand{\cf}[0]{\emph{cf.}\xspace}

\hyphenation{log-a-rith-mic}
\hyphenation{unique-ly}
\hyphenation{ex-tractabil-i-ty}

\algnewcommand\algorithmicon{\textbf{upon}}
\algblockdefx[ON]{On}{EndOn}[1]
{\algorithmicon\ #1\ \algorithmicdo}
{\algorithmicend\ \algorithmicon}

\algnewcommand\waitcommand{\textbf{wait until }}
\def\Wait#1{\State\waitcommand #1}

\def\chain{\mathcal{C}}

\newcommand{\getsrandomly}{\overset{\$}{\gets}}

\algnewcommand{\IfThen}[2]{%
  \State \algorithmicif\ #1\ \algorithmicthen\ #2}

\algdef{SE}[UPON]{Upon}{EndUpon}[1]{\textbf{upon}\ #1}{\textbf{end upon}}

\newcommand{\Tconfirm}[0]{\ensuremath{T_{\mathrm{cf}}}}
\newcommand{\Tconfirmbtc}[0]{\ensuremath{T^{\mathsf{btc}}_{\mathrm{cf}}}}
\newcommand{\Adv}[0]{\ensuremath{\mathcal A}}
\newcommand{\Env}[0]{\ensuremath{\mathcal Z}}
\newcommand{\fS}[0]{\ensuremath{f_{\mathrm{s}}}}
\newcommand{\fA}[0]{\ensuremath{f_{\mathrm{a}}}}
\newcommand{\fL}[0]{\ensuremath{f_{\mathrm{l}}}}
\newcommand{\client}{\ensuremath{\mathsf{c}}}
\newcommand{\poly}[0]{\ensuremath{\operatorname{poly}}}

\newcommand{\secret}{\ensuremath{\mathsf{sc}}}
\newcommand{\Secret}{\ensuremath{\mathcal{S}}}
\newcommand{\ct}{\ensuremath{\mathsf{ct}}}
\newcommand{\Ct}{\ensuremath{\mathcal{C}}}

\newcommand{\EOTS}{EOTS\xspace}
\newcommand{\hctx}{\ensuremath{H_\mathsf{ctx}}}
\newcommand{\hsig}{\ensuremath{H_\mathsf{sig}}}
\newcommand{\deq}{\ensuremath{\coloneqq}}

\newcommand{\sct}{\mathsf{sct}}
\newcommand{\pct}{\mathsf{pct}}
\newcommand{\Z}{\mathbb{Z}}
\newcommand{\G}{G}
\newcommand{\Tctx}{T_{\mathsf{ctx}}}
\newcommand{\OSch}{\mathcal{O}_{\mathsf{Sch}}}
\newcommand{\OEOTSctx}{\mathcal{O}_{\mathsf{EOTS\text{-}Sign\text{-}Ctx}}}
\newcommand{\OEOTS}{\mathcal{O}_{\mathsf{EOTS\text{-}Sign}}}
\newcommand{\Qctx}{Q_{\mathsf{ctx}}}
\newcommand{\Bad}{\mathsf{Bad}}
\newcommand{\Advantage}{\mathsf{Adv}}
\newcommand{\sample}{\mathrel{{\leftarrow}_\$}}
\newcommand{\hashctx}{\mathsf{hct}}

\def\Q{\mathcal{Q}}
\def\T{\mathcal{T}}
\title{
   Bitcoin Staking
}

\author{Xinshu Dong}
\affiliation{
\institution{Babylon Labs}
\city{}
\state{}
\country{}
}
\email{xinshu@babylonlabs.io}

\author{Orfeas Stefanos Thyfronitis Litos}
\affiliation{
\institution{Common Prefix}
\city{}
\state{}
\country{}
}
\email{orfeas.litos@proton.me}

\author{Ertem Nusret Tas}
\affiliation{
\institution{a16z Crypto Research}
\city{}
\state{}
\country{}
}
\email{nusret@stanford.edu}

\author{David Tse}
\affiliation{
\institution{Byzantine Research}
\city{}
\state{}
\country{}
}
\email{dntse@babylonchain.io}

\author{Robin Linus Woll}
\affiliation{
\institution{Stanford University and ZeroSync}
\city{}
\state{}
\country{}
}
\email{robin@zerosync.org}

\author{Lei Yang}
\affiliation{
\institution{MegaETH}
\city{}
\state{}
\country{}
}
\email{lei@yangl1996.com}

\author{Mingchao Yu}
\affiliation{
\institution{Babylon Labs}
\city{}
\state{}
\country{}
}
\email{fishermanymc@gmail.com}

\begin{abstract}

The idea of security sharing goes back to Nakamoto’s introduction of merge mining, a technique that enables Bitcoin miners to reuse their hash power to bootstrap and secure other Proof-of-Work (PoW) blockchains. However, with the rise of Proof-of-Stake (PoS) chains, there is a need for new methods of Bitcoin security sharing. We introduce {\em Bitcoin staking}, a protocol that allows Bitcoin holders to trustlessly use their idle asset to secure a PoS chain. The key challenge is to enable automatic slashing of bitcoins on the Bitcoin chain upon safety violations on the PoS chain. We achieve this using double-authentication-preventing signatures, finality gadgets and bi-directional timestamping between Bitcoin and the PoS chain. Our design is entirely modular and can be integrated with any PoS chain. A version of this protocol was deployed to secure the Babylon mainnet in April 2025 and currently has over 58,000 bitcoins staked (about 4 billion USD at current prices) while paying only $0.05\%$ APR reward to the stakers. This is $2$ orders of magnitude cheaper security cost than in PoS chains secured by their native token.

\end{abstract}

\maketitle

\section{Introduction}

\subsection{Bitcoin Security Sharing}
The concept of \emph{security sharing} is nearly as old as Bitcoin itself. In 2010, Satoshi Nakamoto proposed \emph{merge mining}, which enables Bitcoin miners to reuse their mining power to secure other blockchains \cite{nakamoto_merge}. The goal of merge mining is to achieve \emph{scalability} of Bitcoin: rather than hosting multiple applications on Bitcoin, the protocol supports Bitcoin as a dedicated payment system while using its mining power to secure separate blockchains for other use cases. Namecoin~\cite{namecoin}, Dogecoin~\cite{dogecoin}, and Rootstock (RSK)~\cite{rsk} are examples of chains using merge mining.
However, merge mining faces two major limitations:
\begin{enumerate}
    \item It is only suitable for sharing security with Proof-of-Work (PoW) chains. Since most modern blockchains are Proof-of-Stake (PoS), merge mining has limited applicability.
    \item It allows \emph{costless attacks}: a Bitcoin miner can attack the merge-mined chain without consequences on Bitcoin itself. Because miners are primarily invested in Bitcoin, they can attack the merge-mined chain with little economic risk. As such, merge mining shares \emph{hash power}, but not \emph{security}.
\end{enumerate}

\subsection{Remote Staking}

{\em Remote staking} is a recently emerged approach to security sharing in PoS systems. PoS blockchains typically rely on their native tokens for security. However, this inherently limits the economic security of the chain to the market capitalization of that token. By allowing remote staking—staking of assets from a different blockchain—PoS chains can increase their security through a higher total staked value. In such a protocol, crypto assets are locked in a smart bond contract on the \emph{security provider chain}, designating a preferred validator of the \emph{security consumer chain}. This bond contract enables slashing of the staked asset if and only if the validator commits a provable offense.

This concept underpins \emph{mesh security} in the Cosmos ecosystem~\cite{sunny-mesh, mesh-sec-github}, where assets from one Cosmos chain help secure another. It is also inspired by Ethereum's Eigenlayer \emph{restaking}~\cite{eigenlayer}, which uses ETH collateral to secure middleware components like bridges, data availability layers, and oracle networks.

\subsection{Bitcoin Staking}

Bitcoin is a PoW chain with immense hash power. It is also an asset with a market cap of over  \$1 trillion USD, comprising over 60\% of the total crypto market. Enabling remote staking of bitcoins could unlock this immense asset pool to secure PoS chains. Additionally, slashing bitcoins creates true economic cost for malicious behavior, thus overcoming the costless attack limitation of merge mining. However, Bitcoin lacks Turing-complete smart contracts, posing a challenge for automatic slashing.

The main contribution of this paper is {\em Bitcoin staking}, a remote staking protocol that uses Bitcoin as the provider chain while achieving \emph{optimal economic safety} for the PoS consumer chain. 
Bitcoins are locked in a bond contract on Bitcoin itself -- no bridging is required.  Slashing of the staked assets without smart contracts is enabled by a novel combination of cryptography, consensus design and the limited opcodes of the Bitcoin script. 

A version of the Bitcoin staking protocol secures the Babylon PoS chain \cite{btc-staking}, launched in April 2025. As of this writing, over $58,000$ BTC (about 4 billion USD at current prices) has been staked on the Babylon chain, making it the $8$th largest PoS chain by staking market capitalization (Figure \ref{fig:staking_cap}). Moreover, this significant economic security is obtained at a much lower cost than any of the PoS chains that obtain security from their native token (Figure \ref{fig:security_cost}).

\begin{figure}[h]
    \centering
    \includegraphics[width=\linewidth]{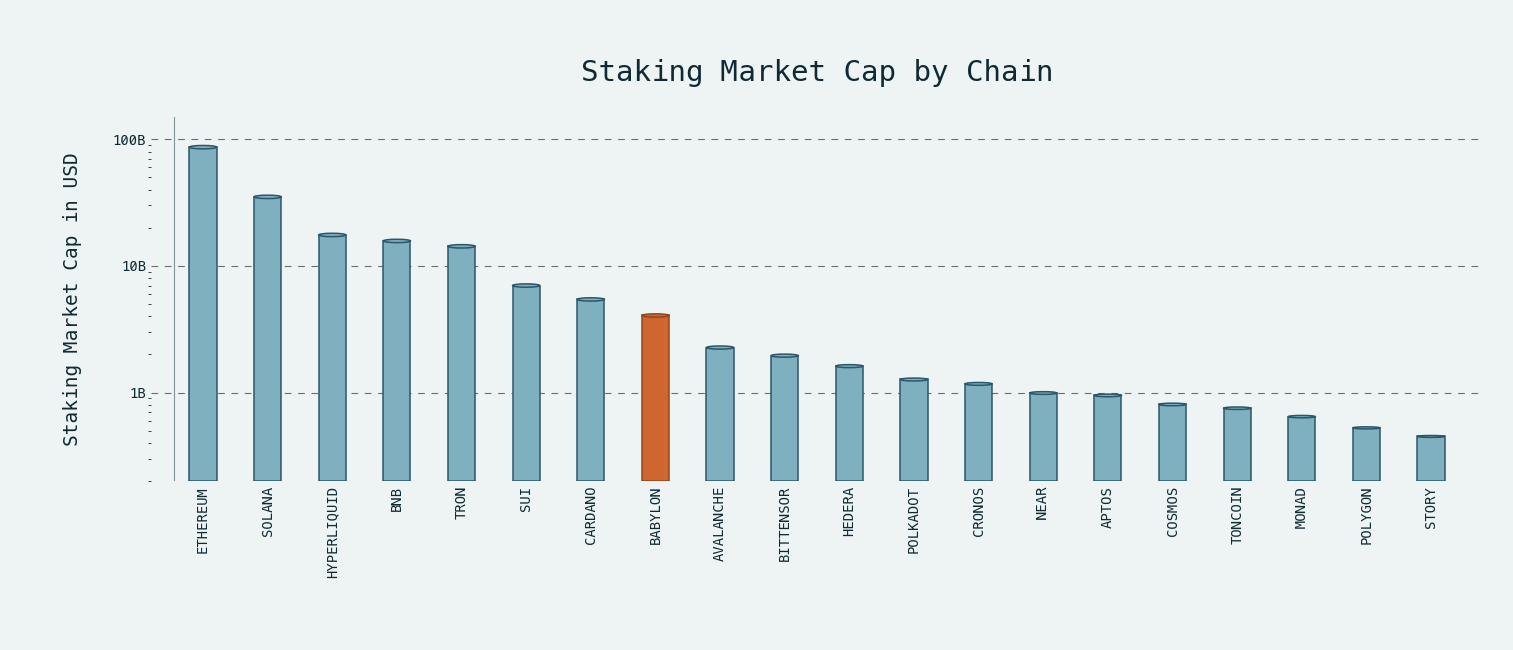}
    \caption{Staking market capitalizations of the top $20$ PoS chains in comparison to Babylon. Note that the $y$-axis is in log scale. Data is from~\cite{stakingreward}, accessed on April 27, 2026.}
\label{fig:staking_cap}
\end{figure}

\begin{figure}[h]
    \centering
    \includegraphics[width=\linewidth]{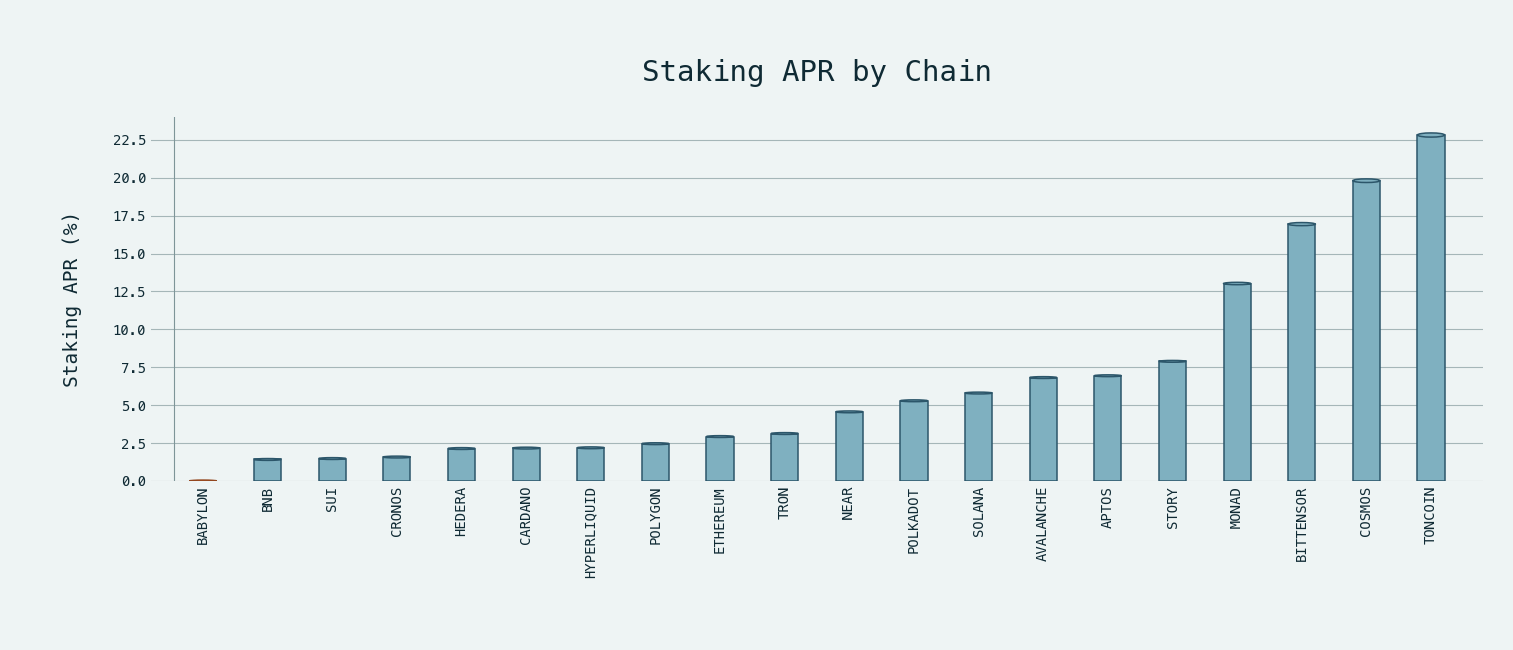}
    \caption{Current cost of security for the top $20$ PoS chains. PoS chains pay rewards programmatically to stakers for locking their capital, and the APR here refers to the annualized reward a staker earns as a percentage of the amount staked. The current APR earned by Bitcoin stakers on the Babylon chain is $0.05\%$, by far the smallest among all chains. Data is from~\cite{stakingreward}, accessed on April 27, 2026.}
\label{fig:security_cost}
\end{figure}

\subsection{Technical Innovations}

There are two main technical innovations that enable the security of the Bitcoin staking protocol:

\subsubsection{Cryptographic slashing without smart contracts}

Because of the lack of smart contracts on Bitcoin, one cannot just send any safety violation evidence and rely on Bitcoin to process such evidence. Our protocol instead allows the sending of an evidence that can directly lead to slashing: the staker's private key. To ensure that the private key of an adversarial staker is leaked whenever a safety violation occurs, we combine two ideas: (a) 
Extractable One-Time Signatures (EOTS) from cryptography, which are related to double-authentication-preventing signatures (DAPS)~\cite{PS14} and accountable assertions~\cite{BKS15}, and (b) finality gadgets, from blockchain consensus \cite{casper, NTT21}.

In EOTS (as in DAPS), each signature is associated with a public \emph{context}, and when a signer signs two messages using the same private key and context, then the private key can be extracted from the signatures on the two messages. 
Unlike DAPS, our EOTS variant splits the context into a private and a public component. The private component is used for signing, whereas the public component is used for verifying.

Accountable assertions and DAPS have been proposed as a general method for penalizing equivocation~\cite{BKS15}, such as double-spending the same bitcoin. However, slashing conditions for consensus protocols are more complex than equivocation on two specific messages. For example, in Casper FFG~\cite{casper}, the slashability module of the PoS Ethereum protocol, there are two sets of slashing conditions. The first set of conditions, that of signing two blocks at the same height, is an equivocation. The second set is more complex and cannot be expressed as an equivocation. Similarly, in Tendermint, there are again two sets of slashing conditions. One set of conditions refers to signing two blocks in the same round at the same height, but the other set of conditions comes from so-called amnesia attacks~\cite{tendermintacc}, which again are not directly expressible as equivocations.

We bypass this issue by not changing the PoS base consensus protocol itself, but instead add an extra signing round after the base consensus protocol has finalized a block, signed using EOTS. 
A block is considered truly finalized if it is both finalized by the base protocol {\em and} receives EOTS signed by more than  $2/3$ of the stake. 
One can interpret this extra round of signing as a type of {\em finality gadget}~\cite{NTT21}, an EOTS finality gadget. It is shown that if there is a safety violation in this modified protocol, then more than $1/3$ of the stake has signed two blocks at the same height using EOTS. This leads to the extraction of the private keys of those stakers. Moreover, the EOTS signature scheme can be implemented based on Schnorr signatures, the signature scheme used in Bitcoin (See~\Cref{sec:daps-constr}). Hence, these extracted private keys can be used to spend the slashing transactions.

One very important advantage of this finality-gadget-based solution is its {\em modular nature}: it can be used on top of all BFT consensus protocols with no change to the base consensus protocols themselves. This enables the technology to be PoS-chain-agnostic.

\subsubsection{Secure unbonding via bi-directional timestamping}

In a PoS chain secured by its native asset, the stake distribution, which tracks the evolution of who has how much stake as tokens change hands, is maintained by the consensus of the blockchain itself. This means that the voting powers of the validators on any block are well-defined.  However, for a remote staking protocol like Bitcoin staking, the stake distribution is maintained on the Bitcoin  chain while the voting is happening on the PoS chain, so there is no {\em a priori} connection between the two. An attacker can slow down the PoS chain, resulting in an out-of-date staker set used to validate a PoS block. This would mean that an attacker can unstake on the Bitcoin chain but still has the voting power to fork on the PoS chain. Even though the private key of the staker is leaked, it is too late to slash it because the staker has already unstaked on the Bitcoin chain. 

To avoid this attack, the PoS chain should be tightly synchronized with the Bitcoin chain. This can be achieved by {\em bi-directional timestamping}, where the PoS block hashes and the staker set voting for the blocks are recorded on the Bitcoin chain, and the Bitcoin block hashes are similarly recorded in the PoS blocks. Interestingly, (single-directional) timestamping of PoS blocks on Bitcoin is also very useful for solving long-range posterior corruption attacks~\cite{DPP19,BGK+18,DPS19} in PoS chains with native staking, by using Bitcoin as an external trust \cite{btc-pos}. Our bi-directional timestamping protocol simultaneously solves the long range attack problem and the stake distribution synchronization problem. In fact, this contribution applies more generally to remote staking protocols for any provider chain, not necessarily Bitcoin.

\begin{figure}[t]
\centering
\resizebox{\columnwidth}{!}{%
\begin{tikzpicture}[x=1cm,y=1cm,line cap=round,line join=round,font=\sffamily]

\draw[bond,line width=0.9pt] (0.55,4.60) rectangle (1.08,5.36);
\draw[bond,line width=0.8pt] (0.68,5.16) -- (0.94,5.16);
\draw[bond,line width=0.8pt] (0.68,5.04) -- (0.94,5.04);
\draw[bond,line width=0.8pt] (0.68,4.92) -- (0.90,4.92);
\draw[bond,line width=0.8pt] (0.68,4.80) -- (0.83,4.80);
\draw[bond,line width=1.0pt] (0.68,4.68) -- (0.77,4.77);
\draw[bond,line width=1.0pt] (0.77,4.68) -- (0.68,4.77);
\draw[bond,line width=1.0pt] (0.90,4.72) -- (1.00,4.72);

\node[align=center,text=bond] at (3.85,4.98)
  {{\bfseries Bond Contract:}\\Validators lock bitcoin for $k$\\blocks in the bond contract.};

\node[align=center,text=stamp] at (9.15,4.93)
  {{\bfseries Timestamping Protocol:}\\PoS block hashes and\\the EOTS signatures};

\node[anchor=west,text=bond] at (0.10,3.30)
  {{\bfseries Provider Chain: Bitcoin}};

\node[align=center,text=finality] at (1.85,1.75)
  {{\bfseries Finality Gadget:}\\Validators sign the finalized\\PoS blocks with EOTS.};

\node[anchor=west,text=consumer] at (1.15,0.45)
  {{\bfseries Consumer Chain}};

\def\xb{4.55}
\def\xc{6.10}
\def\xd{8.05}
\def\xe{10.90}

\draw[bond,line width=0.5pt,->] (4.30,4.28) -- (\xb,3.72);

\draw[bond,fill=bondfill,line width=0.7pt,rounded corners=0.08cm]
  (\xb-0.37,3.14) rectangle (\xb+0.37,3.48);
\fill[stamp] (\xb-0.13,3.14) rectangle (\xb+0.13,3.48);

\draw[bond,fill=bondfill,line width=0.7pt,rounded corners=0.08cm]
  (\xc-0.37,3.14) rectangle (\xc+0.37,3.48);

\draw[bond,fill=bondfill,line width=0.7pt,rounded corners=0.08cm]
  (\xd-0.37,3.14) rectangle (\xd+0.37,3.48);
\fill[stamp] (\xd-0.13,3.14) rectangle (\xd+0.13,3.48);

\draw[bond,fill=bondfill,line width=0.7pt,rounded corners=0.08cm]
  (\xe-0.37,3.14) rectangle (\xe+0.37,3.48);
\fill[stamp] (\xe-0.13,3.14) rectangle (\xe+0.13,3.48);

\draw[bond,line width=0.5pt,->] (\xc-0.42,3.31) -- (\xb+0.42,3.31);
\draw[bond,line width=0.5pt,->] (\xd-0.42,3.31) -- (\xc+0.42,3.31);
\node[text=bond] at (9.00,3.31) {{\bfseries\ldots}};
\draw[bond,line width=0.5pt,->] (\xe-0.42,3.31) -- (9.35,3.31);

\draw[stamp,line width=0.55pt,->] (\xd,3.49) -- (9.00,4.20);
\draw[stamp,line width=0.55pt,->] (\xe,3.49) -- (9.18,4.20);

\foreach \x/\y in {
  4.55/2.10, 4.55/1.55, 4.55/1.00,
  6.05/2.10, 6.05/1.55, 6.05/1.00,
  7.95/2.10, 7.95/1.55, 7.95/1.00,
  10.85/2.10,10.85/1.55,10.85/1.00
}{
  \begin{scope}[shift={(\x,\y)}]
    \draw[finality,line width=0.9pt] (-0.48,-0.19) rectangle (0.48,0.19);
    \draw[finality,line width=0.8pt] (-0.34,0.08) -- (-0.06,0.08);
    \draw[finality,line width=0.8pt] (-0.34,-0.02) -- (0.20,-0.02);
    \draw[finality,line width=0.8pt] (0.24,0.04) -- (0.42,0.04);
    \draw[finality,line width=0.8pt]
      (0.07,-0.09) .. controls (0.13,-0.03) and (0.19,-0.15) .. (0.25,-0.09)
                   .. controls (0.31,-0.03) and (0.35,-0.13) .. (0.41,-0.08);
  \end{scope}
}
\node[text=finality] at (9.45,1.55) {{\bfseries\ldots}};

\draw[stamp,line width=0.55pt] (3.90,2.47) -- (3.90,2.60) -- (5.20,2.60) -- (5.20,2.47);
\draw[stamp,line width=0.55pt,->] (4.55,2.60) -- (4.55,3.10);

\draw[stamp,line width=0.55pt] (7.30,2.47) -- (7.30,2.60) -- (8.60,2.60) -- (8.60,2.47);
\draw[stamp,line width=0.55pt,->] (7.95,2.60) -- (7.95,3.10);

\draw[stamp,line width=0.55pt] (10.20,2.47) -- (10.20,2.60) -- (11.50,2.60) -- (11.50,2.47);
\draw[stamp,line width=0.55pt,->] (10.85,2.60) -- (10.85,3.10);

\foreach \x in {4.55,6.05,7.95,10.85}{
  \draw[consumer,fill=consumerfill,line width=0.7pt,rounded corners=0.08cm]
    (\x-0.40,0.25) rectangle (\x+0.40,0.59);
}
\draw[consumer,line width=0.5pt,->] (5.57,0.42) -- (4.98,0.42);
\draw[consumer,line width=0.5pt,->] (7.45,0.42) -- (6.48,0.42);
\node[text=consumer] at (9.00,0.42) {{\bfseries\ldots}};
\draw[consumer,line width=0.5pt,->] (10.35,0.42) -- (9.55,0.42);

\end{tikzpicture}%
}
\caption{Bitcoin staking protocol. Validators lock their stake in a \emph{bond contract}. They then become eligible to run the consensus protocol of the PoS chain. During this time, they sign the PoS blocks confirmed by the underlying consensus protocol with EOTS as part of the \emph{finality gadget}.  Hashes of the PoS blocks are periodically timestamped on Bitcoin along with the EOTS signatures on them as part of the \emph{bi-directional timestamping protocol}. The stake is either slashed and sent to a burnt address if the EOTS finality gadget extracts a private key, or eventually unbonded.}
\label{fig:dumb-contracts}
\end{figure}

\vspace{0.3in}
\noindent
The full Bitcoin staking protocol is illustrated in Figure \ref{fig:dumb-contracts}. A bond contract on the Bitcoin chain, written in Bitcoin script,  controls the bonding, unbonding and slashing of a Bitcoin staker using information from the EOTS finality gadget and timestamping protocol. The bond contract can be instantiated with timelocks and \emph{covenants}, a new primitive that enables restricting the spending addresses of Bitcoin contracts, or in lieu of covenants, with a \emph{covenant committee} that emulates the functionality of covenants with an external committee of signers. The covenant functionality is needed to enforce that the slashed stake goes to an unspendable burn address.

An important property of the Bitcoin staking protocol is that it does not require changing the rules of the underlying consumer chain's consensus protocol (it is an \emph{add-on} rather than a change of the rules); it can be combined with any consumer chain, including existing ones already secured by native stake.

\subsection{Security Properties}
To capture slashing after a safety violation, we extend the notion of \emph{accountable safety}~\cite{casper,DBLP:conf/ccs/ShengWNKV21} to \emph{economic safety}, which includes both identification and slashing of adversarial validators. Our protocol guarantees $1/3$-economic safety: after any safety violation, the bitcoin stake of at least $f+1$ adversarial validators is slashed, where $n = 3f+1$ is the validator-set size. With homogeneous unit stake, this is at least $1/3$ of the total stake on the consumer chain.

\begin{theorem}[Security, Informal]
\label{thm:economic-safety-informal}
Suppose Bitcoin is secure.
Then, the Bitcoin staking protocol with a covenant committee satisfies $1/3$-economic safety, if one of the committee members is honest.
The protocol satisfies liveness, provided that the fraction of adversarial validators remains below $1/3$. 
The protocol also satisfies \emph{trustless staker safety}: an honest validator's bitcoin cannot be slashed under any circumstances.
\end{theorem}
To ensure the slashing of the adversary's stake, we require one of the covenant committee members to be honest (existential honesty).
If covenants are enabled on Bitcoin in the future, the protocol would not need existential honesty for slasing adversarial stake.

A remote staking protocol satisfying $1/3$-economic safety means that no matter how many adversarial validators there are, at least $1/3$ of them are guaranteed to be slashed after a safety violation.
On the other hand, no PoS blockchain secured only by its native stake can slash the validators if the fraction of adversarial validators exceeds $2/3$~\cite{roughgarden_talk}.
Indeed, by borrowing a sufficient amount of stake, the adversary can temporarily control over $2/3$ of the validator set, gaining complete power over the PoS chain for some time.
It can then cause a safety violation, and subsequently withdraw its stake to pay back its loan, thus violating safety without any financial cost.
This attack highlights the circularity in the security argument for native staking: the chain where the native stake is locked is the same as the chain secured by this stake.
The Bitcoin staking protocol overcomes the limitations of native staking by breaking the circularity above, \ie, by separating the consumer chain secured by the Bitcoin stake, and the Bitcoin chain maintaining this stake.

As a strengthening of accountable safety, $1/3$-economic safety implies $1/3$-accountable safety.
Moreover, no consensus protocol can simultaneously satisfy $1/3$-accountable safety and liveness with resilience greater than $1/3$~\cite[Theorem B.1]{DBLP:conf/ccs/ShengWNKV21}.
Therefore, Theorem~\ref{thm:economic-safety-informal} implies the optimality of the Bitcoin staking protocol in terms of economic safety and liveness resiliences.

In this paper, we focus on the case when the consumer chain is entirely secured by the remote Bitcoin assets.
Variation of the design incorporating both the remote and the native asset in securing the consumer chain (dual staking) is possible but is beyond the scope of this paper.

\subsection{Paper Outline}

\Cref{sec:related-work} discusses related work, and~\Cref{sec:model} provides the preliminaries.
\Cref{sec:daps-constr} describes our EOTS construction.
\Cref{sec:finality-gadget} explains the finality gadget, and \Cref{sec:bond-contract} focuses on the bond contract.
\Cref{sec:timestamping-protocol} describes the bi-directional timestamping protocol.
\Cref{sec:implementation} reports on our proof-of-concept implementation, and \Cref{sec:cost} analyzes the security cost of Bitcoin staking via a case study on the Babylon chain.

\section{Related Work}
\label{sec:related-work}

\paragraph{Accountability and Slashing.}
\emph{Accountable safety}, 
\ie, the ability to identify a fraction of the adversarial validators in the event of a safety violation, is central to the design of PoS protocols such as Ethereum~\cite{casper,gasper} and Cosmos zones based on Tendermint~\cite{buchman2016tendermint,tendermintacc}.
These protocols require the validators to stake the native tokens as collateral; so that these tokens can be \emph{slashed}, \ie, taken away, if the validator is found responsible for a safety violation.
However, identifying adversarial validators is not sufficient to ensure their slashing. 
For instance, adversarial validators can \emph{unbond} their stake and then build a conflicting chain as if they were part of the validator set (posterior corruption or long range attack~\cite{DPP19,BGK+18,DPS19}).
Although these validators can be identified as adversarial, they cannot be slashed.
Indeed, \cite{btc-pos} proved that without external trust assumptions, blockchains cannot guarantee slashing of the adversarial validators even when they can eventually be identified.

As a solution, \cite{btc-pos} also proposed using a separate chain as a secure timestamping server for checkpointing the confirmed PoS blocks.
Through these timestamps, the PoS chain obtains \emph{slashable safety}, \ie, the ability to identify the adversarial validators not only after a safety violation, but also \emph{before they unbond their stake}.
However, slashable safety does not imply the act of slashing either.
Indeed, when over $1/3$ of the validators are adversarial, they can censor the evidence of safety violation and their culpability from being included on the chain and enforcing their slashing. 
In such cases, a complex social consensus process has to happen off-chain so that the violators can be slashed and kicked out of the validator set.
In contrast, our Bitcoin staking protocol does not suffer from this issue as the remote stake resides on Bitcoin, not on the PoS consumer chain, and it is automatically slashed if the safety of the consumer chain breaks down.
It also builds on ideas from \cite{btc-pos} to mitigate long range attacks.

\paragraph{Finality Gadgets.}
The term ``finality gadget'' denotes a broad class of protocols deployed on top of existing consensus protocols to provide extra guarantees such as safety under network partitions and accountable safety; our finality gadget is an example.
An early example is Casper FFG used in Ethereum on top of a dynamically available consensus protocol (LMD GHOST) to checkpoint blocks.
These checkpoints then constitute an accountably-safe prefix of the Ethereum ledger~\cite{casper,gasper}.
Other examples include~\cite{SWK+21,NTT22}.
Our finality gadget also provides accountable safety to the underlying protocol; however unlike Casper FFG, it is instantiated on top of a partially-synchronous protocol with a fixed-sized validator set (without dynamic availability).
A similar gadget was used in~\cite{oflex} with the purpose of enabling clients to opt for higher safety resilience at the expense of reduced liveness resilience.

Stakechain~\cite{stakechain} explored remote staking for consensus protocols without using a finality gadget.
For accountable safety, it relies directly on a quorum intersection argument over the validators' signatures on the consumer blocks.
However, without a view change mechanism, the construction loses liveness when there is an adversarial block proposer.

\paragraph{DAPS and Accountable Assertions.}
DAPS~\cite{PS14} provide an efficient algorithm for extracting a signer's secret key if it signs two different messages using the same \emph{context}, a public attribute of the signatures.
DAPS were later generalized to lattice-based \emph{predicate authentication preventing signatures} that support general predicates for exposing the keys~\cite{BKN17}.
However, the signing keys in~\cite{PS14,BKN17} do not match those of ECDSA or Schnorr signatures used to authorize Bitcoin transactions, thus these schemes cannot be used in the Bitcoin staking protocol.

A related primitive to DAPS is accountable assertions~\cite{BKS15}, which allow extracting a user's Bitcoin secret key if it makes multiple assertions with the same context.
Unlike early works, accountable assertions enable using an ECDSA or Schnorr key pair to make assertions, thus enabling extraction of a user's Bitcoin secret key if it equivocates.
However, the assertion size is $4$ kB, which is much larger than the state-of-the-art signatures used in blockchains (\eg, $32$ Bytes for Schnorr). 

In~\Cref{sec:daps-constr}, we design EOTS, a simple variant of DAPS based on Schnorr signatures~\cite{DBLP:conf/eurocrypt/Seurin12}.
EOTS enables signing using the Bitcoin secret key and verifying signatures with the corresponding public key.
It also has the same size as Schnorr signatures.
In EOTS, the context consists of a private component required for signing, and a public component used in verification.
As such, the EOTS scheme can be viewed as a weakening of DAPS --- it is nevertheless sufficient for the purpose of our Bitcoin staking protocol.

\paragraph{Covenants.}
\label{sec:covenants}
Covenants are powerful primitives to express Bitcoin contracts. 
Bitcoin is based on the `UTXO model', which is inherently stateless: when a transaction is executed, its input coins are destroyed, and new output coins are created. 
In a regular transaction, the owner of the input coin chooses which output coins are created. 
Covenants limit this freedom and restrict a coin such that the owner can send it only to a certain recipient or contract. 
Covenants have been discussed in the Bitcoin community since at least 2013~\cite{gmax_covenants} and in academic literature since 2016~\cite{moser2016bitcoin}. 
There are Bitcoin improvement proposals~\cite{bip118,op_cat} currently in discussion that can be used to emulate covenants~\cite{cat_csfs}.

\section{Preliminaries}
\label{sec:model}

Let $\kappa \in \mathbb{N}$ denote the security parameter. An event has overwhelming probability (w.o.p.) if its probability is at least $1-\negl$.

\paragraph{State Machine Replication.}
In state machine replication, validators (also known as replicas) receive transactions and execute a protocol to impose a total order on these transactions.
Then, clients collect consensus messages (\eg, blocks, votes) from the validators, and 
invoke a \emph{confirmation rule} to output a sequence of \emph{confirmed} transactions called the \emph{ledger}.
The set of clients includes honest validators
and external observers (\eg, wallets) that might query for the ledger at arbitrary times.

The validator set of an SMR protocol can be \emph{static} or \emph{dynamic}.
In the \emph{static} case, we assume a public-key infrastructure (PKI) that assigns unique and publicly known identities to a fixed set of validators.
In the \emph{dynamic} case, we assume a proof-of-stake Sybil resistance mechanism, wherein a node becomes a validator
upon bonding some minimum amount of stake in the protocol.
A validator can also leave the validator set by unbonding its stake.
Although validators can bond different amounts in practice,
we consider validators with homogeneous unit-stake in this work, as an entity with a large stake can be represented as multiple unit-stake entities controlled by the same validator.

\paragraph{Blocks and Chains.}
Transactions are often batched into \emph{blocks} to ensure higher throughput, and the SMR protocol outputs a chain of blocks denoted by $\chain$.
There is a publicly-known \emph{genesis} block $B_0$.
Each block points to a \emph{parent} block via a collision-resistant hash function.
A block $B$ is an ancestor of $B'$, denoted by $B' \preceq B$, if $B' = B$, or $B$ can be reached from $B'$ via a path of parent pointers.
Thus, each block $B$ identifies a unique chain that starts at the genesis block and ends at $B$.
Two blocks $B$ and $B'$ (and their chains) are said to \emph{conflict} if neither $B \preceq B'$ nor $B' \preceq B$.

\paragraph{Adversary.}
The adversary $\Adv$ is an efficient algorithm that corrupts a subset of the validators, hereafter called \emph{adversarial}.
It gains access to the internal states of the corrupted validators and can cause them to deviate from the protocol in an arbitrary and coordinated fashion (Byzantine faults).
The remaining validators are called \emph{honest} and execute the prescribed protocol.
$f$ denotes the maximum number of adversarial validators.

\paragraph{Networking.}
Messages are sent over 
authenticated point-to-point channels~\cite{DBLP:journals/toplas/LamportSP82}.
The adversary controls the timing of message delivery and can inspect the messages before they are delivered.
Upon becoming online, clients receive all messages delivered to them while asleep.
We say that a validator \emph{broadcasts} a message if its intended recipients include all other validators and clients.

The network is \emph{partially synchronous}: 
the adversary has total control over message delays until an adversarially determined, finite global stabilization time (GST).
After GST, the adversary must deliver the messages sent by an honest validator within a known delay bound $\Delta$.

\paragraph{Security.}
Let $\chain^{\client}_t$ denote the confirmed chain output by a client $\client$ at time $t$.
\begin{definition}
\label{def:security}
We say that an SMR protocol is secure with latency $\Tconfirm = \poly(\kappa)$ if:

\noindent    
\textbf{Safety:} For any time slots $t,t'$ and clients $\client,\client'$, either $\chain_{t}^{\client} \preceq \chain_{t'}^{\client'}$ or vice versa. 
For any client $\client$, $\chain^{\client}_{t} \preceq \chain^{\client}_{t'}$ for all times $t$ and $t' \geq t$.

\noindent
\textbf{Liveness:} If a transaction $\tx$ is input to an honest validator\footnote{In the case of dynamic validator set, the transaction is input to an honest validator eligible to participate in the protocol in its local view.} at some slot $t$, then $\tx \in \chain_{t'}^{\client}$ for all $t' \geq \max(t,\mathsf{GST})+\Tconfirm$ and all $\client$.
\end{definition}
\begin{definition}[\cite{casper,snapchat,DBLP:conf/ccs/ShengWNKV21}]
\label{def:accountable-safety}
A protocol provides accountable safety with resilience $\fA$, if when there is a safety violation, clients can generate a proof that identifies (i) at least $\fA$ adversarial validators as protocol violators, and (ii) no honest validator.
Such a protocol provides \emph{$\fA$-accountable-safety}.
\end{definition}
A protocol provides $\fS$-safety if it satisfies safety w.o.p. for all efficient $\Adv$ with $f \leq \fS$.
When the safety of a protocol with $\fA$-accountable-safety is violated, at least $\fA$ adversarial validators must be identified, which cannot happen if fewer than $\fA$ validators are adversarial.
Thus, $\fA$-accountable safety implies $(\fA-1)$-safety.

\paragraph{Bitcoin.}
\label{sec:btc}

Let $\powchain^\client_t$ denote the confirmed Bitcoin chain in the view of a client $\client$ at time $t$, \ie, the $k$-deep block and its prefix within the longest Bitcoin chain held by $\client$ at time $t$, where $k$ is the Bitcoin confirmation-depth parameter.
We denote the confirmed consumer chain by $\chain^\client_t$ and use capital $B$ and small $b$ for the consumer and Bitcoin blocks respectively.
The following proposition is used in the description and analysis of the timestamping protocol (Section~\ref{sec:timestamping-protocol}).
\begin{proposition}
\label{prop:bitcoin-secure}
Suppose Bitcoin is safe and live with latency $\Tconfirmbtc$.
Then, there exists a parameter $k_f$ such that if a transaction is broadcast to the Bitcoin network at the time some honest client's confirmed Bitcoin chain has height $h$, then for any honest client $\client_2$, the transaction appears in $\powchain^{\client_2}_t$ by the first time $t$ at which $|\powchain^{\client_2}_t| \geq h+k_f$.
Moreover, there exists a function $k_c(\Delta t)$, linear in $\Delta t$, such that for any $\Delta t > 0$, any honest client $\client$, and any time $t$, $\powchain^\client$ grows by at most $k_c(\Delta t)$ blocks during the interval $[t,\, t+\Delta t]$ in $\client$'s view.
\end{proposition}
Proposition~\ref{prop:bitcoin-secure} follows from the security analysis of Bitcoin provided by~\cite{backbone}.
Going forward, we assume that Bitcoin is secure and \Cref{prop:bitcoin-secure} holds.

Note that we assume a synchronous network for the communication among Bitcoin miners and clients (these clients include consumer chain validators), while assuming a partially-synchronous network for the communication among consumer validators and clients.
Although the consumer validators and clients could have uses Bitcoin as a synchronous bullet-in board for every block (which essentially implies consensus), this is infeasible in practice due to Bitcoin's limited throughput and long latency.

\section{EOTS Construction}
\label{sec:daps-constr}

\begin{definition}[Extractable One-Time Signatures]%
\label{EOTS:syntax}
EOTS consist of six algorithms:

\smallskip
\noindent
$\sk \overset{\$}{\gets} \mathsf{EOTS\textnormal{-}KeyGen}(1^{\kappa}) \colon$ The key generation algorithm 
outputs a secret signing key.

\smallskip
\noindent
$\pk \gets \mathsf{EOTS\textnormal{-}PK}(\sk) \colon$ The public key derivation algorithm takes $\sk$ and outputs a public verification key.

\smallskip
\noindent
$(\sct, \pct) \overset{\$}{\gets} \mathsf{EOTS\textnormal{-}ContextGen}(1^{\kappa}) \colon$ The context generation algorithm outputs the secret and public context components.

\smallskip
\noindent
$\sigma \overset{\$}{\gets} \mathsf{EOTS\textnormal{-}Sign}(\sk, \sct, m) \colon$ The signing algorithm is a probabilistic algorithm that outputs a signature $\sigma \in \Sigma$ given a secret signing key $\sk$, a message $m \in \mathcal{M}$ and a secret context $\sct$.

\smallskip
\noindent
$\{0, 1\} \gets \mathsf{EOTS\textnormal{-}Ver}(\pk, m, \pct, \sigma) \colon$ The verification algorithm is a deterministic algorithm that outputs $1$ if a given signature $\sigma$ is verified against a public verification key $\pk$, a message $m$ and a public context $\pct$ ($0$ otherwise).

\smallskip
\noindent
$\sk \gets \mathsf{EOTS\textnormal{-}Ext}(\pk, m_1, \sigma_1, m_2, \sigma_2, \pct) \colon$ The extraction algorithm is a probabilistic algorithm that outputs the secret signing key $\sk$ of a validator given two distinct message-signature pairs $(m_1,\sigma_1)$ and $(m_2,\sigma_2)$, where the signatures are valid under the same public context $\pct$.
\end{definition}

We separate the secret key generation $\mathsf{EOTS\textnormal{-}KeyGen}()$ from public key derivation $\mathsf{EOTS\textnormal{-}PK}(.)$ to facilitate the $\mathsf{EXT\textnormal{-}SCMA}$ security property that is analogous to extractability~\cite{PS14,BKS15}.
This is necessary for an EOTS scheme without a trusted setup, when there may be many different (hard to find) secret signing keys corresponding to a given public verification key.
For the same reason, \cite{PS14} assumes the existence of an efficient algorithm akin to our $\mathsf{EOTS\textnormal{-}PK}(.)$ (without explicitly defining the algorithm), which verifies that a given secret key $\sk$ is the key corresponding to a public key $\pk$.
In turn, \cite{BKS15} assumes that for each $\pk$, there is a unique $\sk$.

\begin{algorithm}
\caption{(Informal) Game for Strong Existential Unforgeability under 
Chosen Message Attacks with Distinct Contexts (\textsf{sEUF-CMADC}). The full game is presented in~\Cref{daps:game:seuf-cmadr} with the difference that the adversary is also allowed to query an $\OEOTS$ oracle that samples a new context tuple and returns the public component along with the signature.}
\label{daps:game:seuf-cmadr-main-body}
\begin{algorithmic}[1]\small
\State $\sk \overset{\$}{\gets} \mathsf{EOTS\textnormal{-}KeyGen}(1^{\kappa})$; $\pk
\gets \mathsf{EOTS\textnormal{-}PK}(\sk)$
\For{$i = 1$ to $N$} \Comment{Here, $N$ is $\poly(\kappa)$.}
    \State $(\sct_i, \pct_i) \overset{\$}{\gets} \mathsf{EOTS\textnormal{-}CtxGen}(1^{\kappa})$
\EndFor
\State $Q \gets \emptyset$
\State $\Qctx \gets \emptyset$
\State $(m^*, \pct^*, \sigma^*) \gets \mathcal{A}^{\OEOTSctx(\cdot)}(\pk, \{\pct_i\}_{i=1,\ldots,N})$ \Comment{$\mathcal{A}$ can call the oracle multiple times}
\State \Return $(m^*, \pct^*, \sigma^*) \notin Q \wedge \mathsf{EOTS\textnormal{-}Verify}(\pk, \pct^*, m^*,
\sigma^*)$
\Statex

\Function{$\OEOTSctx$}{$m, \pct$} \Comment{Signing oracle with context}
  \If{$\pct \in \{\pct_i\}_{i = 1, \ldots, N}$ and $\pct \notin \Qctx$}
    \State $\sigma \overset{\$}{\gets} \mathsf{EOTS\textnormal{-}Sign}(\sk, \sct_{i'}, m)$ where $\pct = \pct_{i'}$ \Comment{$\sct_{i'}$ is the secret context component corresponding to $\pct_{i'}$.}
    \State $Q \gets Q \cup \{(m, \pct, \sigma)\}$
    \State $\Qctx \gets \Qctx \cup \{\pct\}$
    \State \Return $(\pct, \sigma)$
  \Else
    \State \Return $\bot$
  \EndIf
\EndFunction
\end{algorithmic}
\end{algorithm}

\paragraph{Security Definitions.}
Security of our EOTS scheme is characterized by three properties: correctness, $\mathsf{sEUF\textnormal{-}CMADC}$ security (strong existential unforgeability under adaptive chosen message attacks), and $\mathsf{EXT\textnormal{-}SCMA}$ security (extractability under single chosen message attacks).
Intuitively, correctness guarantees that a correctly generated signature always passes verification.
$\mathsf{EXT\textnormal{-}SCMA}$ security guarantees that two valid signatures on distinct messages under the same key and \emph{same context} can be used to extract the secret key.
This is a crucial requirement to slash the adversarial parties' stake which requires obtaining their secret keys.
Formal definitions are stated in~\Cref{sec:eots-proof}.

$\mathsf{sEUF\textnormal{-}CMADC}$ security, captured by the game in~\Cref{daps:game:seuf-cmadr-main-body}, is analogous to existential unforgeability and ensures that signatures are unforgeable when the secret key is unknown, even after querying for multiple signatures in different contexts.
The main difference is that the adversary receives a set of public context components at the setup of the game and is allowed to query for the signatures tied to these components.
This mimics how EOTS are used within the Bitcoin staking protocol, where each validator creates a sequence of context tuples for different heights and publishes the public context components on Bitcoin for consistency.

\paragraph{Construction.} We next provide an EOTS construction based on Schnorr signatures~\cite{DBLP:conf/eurocrypt/Seurin12}, specifically the flavor used by Bitcoin\footnote{\url{https://github.com/bitcoin/bips/blob/master/bip-0340.mediawiki}}, which is assumed to be \textsf{sEUF-CMADC} secure.
The system parameters are a group $G$ of prime order $q$ (the size of which is
determined by the security parameter) with a generator $g$, as well as two hash functions $\hctx \colon G \mapsto \{0, 1\}^\kappa$ and $\hsig \colon \{0, 1\}^* \mapsto \mathbb{Z}_q$, where $\kappa$ is a security parameter.
In practice, $G$ and $\hctx$ must match the group and the hash function used by Bitcoin's Schnorr implementation.

Intutively, each Schnorr signature consists of two components: a public random nonce $r=kG$ and $s=k+sk\cdot e$, with $e$ depending on the signed message. Our key observation is that Schnorr signatures require a fresh $r$ to sign each message, and reusing randomness across signatures leaks the secret key. To turn Schnorr into EOTS, we use this hazard as a feature: for each consumer chain block height $H$, each validator fixes a unique $k$ and publishes the \emph{hash of its public nonce} $r \deq g^{k}$ as the public context component for height $H$: $\pct_H \deq \hctx(r)$.
Consumer chain validators and clients require the signature $\sigma \deq (r, s)$ on a block at height $H$ to also match $\pct_H$: $\pct_H =_? \hctx(r)$.
When $\hctx$ is collision-resistant, equivocation implies two signatures with the same value $r$, which allows extraction of $\sk$, a standard Schnorr private key. This can in turn be used to slash the stake on Bitcoin.
Our scheme is formally presented by~\Cref{daps:keygen,daps:pk-der,daps:randgen,daps:sign,daps:alg:verify,daps:ext}.

\begin{algorithm}[t]
\caption{$\mathsf{EOTS\textnormal{-}KeyGen}(1^{\kappa})$:}
\label{daps:keygen}
\begin{algorithmic}[1]\small
    \State $x \overset{\$}{\gets} U(\mathbb{Z}^{\times}_q)$ \Comment{$\mathbb{Z}^{\times}_q$: Multiplicative group of size $q$}
    \State \Return $x$
\end{algorithmic}
\end{algorithm}

\begin{algorithm}[t]
\caption{$\mathsf{EOTS\textnormal{-}PK}(x)$:}
\label{daps:pk-der}
\begin{algorithmic}[1]\small
    \State \Return $g^x$ \Comment{$g$ is a generator of group $G$ of prime order $q$.}
\end{algorithmic}
\end{algorithm}

\begin{algorithm}[t]
\caption{$\mathsf{EOTS\textnormal{-}ContextGen}(1^{\kappa})$:}
\label{daps:randgen}
\begin{algorithmic}[1]\small
    \State $k \overset{\$}{\gets} U(\mathbb{Z}^{\times}_q)$
    \State $r \gets g^k$
    \State $\hashctx_r \gets \hctx(r)$ \Comment{$\hctx$ is a hash function.}
    \State \Return $(\sct \deq k,\ \pct \deq \hashctx_r)$
\end{algorithmic}
\end{algorithm}

\begin{algorithm}[t]
\caption{$\mathsf{EOTS\textnormal{-}Sign}(x, k, m), m \in \{0, 1\}^*$:}
\label{daps:sign}
\begin{algorithmic}[1]\small
    \State $r \gets g^k$ 
    \State $e \gets \hsig(r || m)$ \Comment{$\hsig$ is a hash function.}
    \State $s \gets k + xe$ 
    \State \Return $\sigma \deq (r, s)$ 
\end{algorithmic}
\end{algorithm}

\begin{algorithm}[t]
\caption{$\mathsf{EOTS\textnormal{-}Verify}(y, m, \hashctx_r, \sigma)$:}
\label{daps:alg:verify}
\begin{algorithmic}[1]\small
    \State $(r, s) \gets \sigma$
    \State $r_v \gets g^s (y^{\hsig(r || m)})^{-1}$ 
    \State \Return $\hctx(r) \overset{?}{=} \hashctx_r$ and $r_v \overset{?}{=} r$
    \Comment{The first check is unique to EOTS.}
\end{algorithmic}
\end{algorithm}

\begin{algorithm}[t]
\caption{$\mathsf{EOTS\textnormal{-}Extract}(y, \hashctx_r, m_1, \sigma_1, m_2, \sigma_2)$:}
\label{daps:ext}
\begin{algorithmic}[1]\small
    \State $(r_1, s_1) \gets \sigma_1$; $(r_2, s_2) \gets \sigma_2$
    \State Abort if $r_1 \neq r_2$; else let $r \deq r_1 = r_2$
    \State $e_1 \gets \hsig(r || m_1)$; $e_2 \gets \hsig(r || m_2)$ \Comment{When $s_1, s_2$ made with} 
    \State \Return $\frac{s_1 - s_2}{e_1 - e_2}$\Comment{same nonce, output is equal to secret key $x$}
\end{algorithmic}
\end{algorithm}

\paragraph{Security.}
We prove the security of our scheme in~\Cref{sec:eots-proof}.
We show that it satisfies correctness, and when $\hctx$ and $\hsig$ are collision-resistant, $\mathsf{EXT\textnormal{-}SCMA}$ security.
We also prove that it is $\mathsf{sEUF\textnormal{-}CMADC}$ secure in the random oracle model via a reduction to the security of Schnorr signatures.
In the proof, $\hctx$ is modeled as a random oracle (so is $\hsig$ for the original security proof of Schnorr signatures).
This enables the adversary trying to forge a Schnorr signature to simulate the $\OEOTSctx$ oracle in~\Cref{daps:game:seuf-cmadr-main-body}: upon receiving an input $(\pct, m)$, the adversary queries the Schnorr signing oracle, obtains a new signature $\sigma = (r, s)$, programs $\hctx(r) \deq \pct$, and returns $(\pct, \sigma)$.

\section{Finality Gadget}
\label{sec:finality-gadget}

We next describe how to use the EOTS in~\Cref{sec:daps-constr} to expose the secret signing keys of the adversarial validators after a safety violation.
We first consider a static validator set 
and extend our construction to a dynamic set in~\Cref{sec:timestamping-protocol}.

\paragraph{Tendermint.} Tendermint is a PBFT-style~\cite{DBLP:conf/osdi/CastroL99} partially-synchronous SMR protocol.
It proceeds in \emph{rounds}, each with a known leader and $n = 3f+1$ validators.
Each round in turn consists of three steps: $\mathsf{Proposal}$, $\mathsf{Prevote}$ and $\mathsf{Precommit}$.
At the beginning of the $\mathsf{Proposal}$ step, the leader sends a signed $\mathsf{Proposal}$ message, $\langle \mathsf{Proposal}, H, r, B, vr \rangle$ for some block $B$.
Here, $H$ and $r$ denote the leader's current height and round number respectively.
Upon observing a proposal, each validator enters the $\mathsf{Prevote}$ step and sends a $\mathsf{Prevote}$ message $\langle \mathsf{Prevote}, H, r, s \rangle$ for either the proposed block ($s = \mathsf{hash}(B)$), or a special \emph{nil} value ($s = \bot$), depending on the proposal and its internal state.
If the validator observes $2f+1$ prevotes for a block $B$ (or the \emph{nil} value), it subsequently enters the $\mathsf{Precommit}$ step and sends a $\mathsf{Precommit}$ message $\langle \mathsf{Precommit}, H, r, \mathsf{hash}(B) \rangle$ for $B$ or the \emph{nil} value.
Finally, a client \emph{confirms} $B$ for height $H$ upon observing $2f+1$ precommits with $H$ and $B$.

A validator locks on a block $B \neq \bot$ at some round $r$ upon sending a round $r$ precommit for $B$.
In future rounds of the same height, the validator does not send prevotes for other blocks $B' \neq B$, unless it observes $2f+1$ or more prevotes for $B'$ from a round $r' > r$.

We assume that the length of a Tendermint round is lower bounded by some $\epsilon > 0$, matching practical implementations.

\paragraph{Challenges with Exposing the Adversarial Keys.}
We next describe a strawman scheme to extract the adversarial keys in the event of a safety violation.
Suppose the Tendermint validators use EOTS, with context equal to the height-round tuple $(H,r)$, to sign prevotes and precommits.
Consider an adversarial validator that sends two prevotes or precommits for different blocks $B$ and $B' \neq B$ with the same context $(H, r)$ to cause a safety violation.
Then, by the extractability property of EOTS (Def.~\ref{EOTS:sec:ext-scma}), its secret key can be extracted (w.o.p.).
However, 
in the case of Tendermint and other PBFT-style protocols (\eg, HotStuff~\cite{hotstuff}), there can be safety violations that are not due to the adversary double-signing messages with the same context.
For instance, in Tendermint, $f+1$ adversarial validators can first send precommits $\langle \mathsf{Precommit}, H, r, \mathsf{hash}(B) \rangle$ for a block $B$ at round $r$, and then prevotes $\langle \mathsf{Prevote}, H, r+1, \mathsf{hash}(B') \rangle$ for a conflicting block $B' \neq B$ at the next round $r+1$ to release their locks on $B$.
They do this despite not having observed $2f+1$ prevotes for $B'$ from rounds $>r$.
Similarly, in HotStuff, 
$f+1$ adversarial validators can send commit messages in a view $v$ for a block $B$ and then prepare messages in view $v+1$ for a conflicting block $B' \neq B$ with a $\mathsf{highQC}$ (quorum certificate of $2f+1$ prepare messages) from some view $v''<v$, again releasing their lock on $B$ without justification.
In both cases, both blocks $B$ and $B'$ eventually become confirmed (for height $H$), and these $f+1$ validators can be provably identified as protocol violators.
However, as the messages for the conflicting blocks were signed with \emph{different} contexts (view/height and round numbers), their secret signing keys cannot be extracted.

\begin{algorithm}
\caption{A validator $\validator$'s execution of the finality gadget. The function $\textsc{Broadcast}$ broadcasts the provided messages. A message $m$ and a signature on it by the validator is denoted by $\langle m \rangle_\validator$. 
Each validator keeps track of the latest $\mathsf{height}$ for which a \EOTS signature was broadcast.}
\label{alg:validator-finality}
\begin{algorithmic}[1]\small
    \Let{\mathsf{height}}{0}
    \Upon{$\text{decision}[H]$}~\Comment{A block is confirmed at height $H$.}
        \If{$H = \mathsf{height} + 1$}
            \State $\textsc{Broadcast} \ \langle \mathsf{Final}, H, \mathsf{hash}(\text{decision}[H]) \rangle_\validator$
            \Let{\mathsf{height}}{\mathsf{height}+1}
        \EndIf
    \EndUpon
\end{algorithmic}
\end{algorithm}

\begin{algorithm}
\caption{The finalization algorithm run by a client $\client$ on Tendermint with the finality gadget. The inputs $\T$ and $\mathsf{sigs}$ denote the blocks and \EOTS signatures received by $\client$. The input $\chain$ denotes the chain of blocks previously finalized by $\client$ ($\chain = B_0$ if no block has been finalized yet). The function $\textsc{GetBlocks}(\T)$ returns the sequence of blocks within $\T$ in increasing order of heights, ties broken arbitrarily.
Height of a block $B$ is denoted by $|B|$. The algorithm returns a chain of finalized (valid) blocks. 
}
\label{alg:client-finality}
\begin{algorithmic}[1]\small
\Function{\sc OutputChain}{$\T, \mathsf{sigs}, \chain$}
    \For{$B = B_1, \ldots, B_H \xleftarrow{} \textsc{GetBlocks}(\T)$}
        \If{$|B| = |\chain|+1\ \land\ \chain \preceq B\ \land\ \exists (2f+1)\ \langle \mathsf{Final}, |B|, \mathsf{hash}(B) \rangle \in \mathsf{sigs}$}
            \Let{\chain}{\chain \concat B}
        \EndIf
    \EndFor
\State \Return $\chain$
\EndFunction
\end{algorithmic}
\end{algorithm}

\subsection{Tendermint with the Finality Gadget}
\label{sec:tendermint-finality-gadget}

The attacks above imply that extracting the adversary's secret keys must somehow involve double-signing of messages with the same context in any attack scenario.
To achieve this functionality, we add a layer of \emph{\EOTS signatures} for the confirmed blocks (\Cref{alg:validator-finality}).
More specifically, the finality gadget replaces the original confirmation rule of Tendermint with a \emph{finalization} rule based on \EOTS signatures.

\paragraph{Generating the Context Tuples.}
Each validator $\validator$ periodically generates context tuples $(\sct_{H,\validator}, \pct_{H,\validator}) \gets \mathsf{EOTS\textnormal{-}ContextGen}(1^{\kappa})$ for a window of $w_0$ Tendermint chain heights $H \in \{w+1, \ldots, w+w_0\}$ and publishes the public components $\pct_{H,\validator}$ on Bitcoin.
The initial set of components for heights $\{1, \ldots, w_0\}$ is published before the protocol starts.
Here, Bitcoin serves as a coordination layer which ensures that all validators and clients agree on the public context components for each height and validator.
The number $w_0$ is determined so that there is sufficient time for $\pct_{H,\validator}$, $H \in \{w+1, \ldots, w+w_0\}$ to appear on Bitcoin, if they are sent when the Tendermint chain is at height $w-w_0+1$.

Given $n$ validators, up to $n w_0$ public context components might be posted to Bitcoin at each iteration.
To reduce this load, in practice, validators can commit to the $w_0$ public context components via a Merkle tree and periodically post a Merkle root covering the Tendermint heights $\{w+1, \ldots, w+w_0\}$.
In this case, each EOTS signature must come with both the component $\pct_{H,\validator}$ and a Merkle proof with respect to the root covering height $H$.
For simplicity, we describe a version of the protocol, where the values $\pct_{H,\validator}$ are published on Bitcoin in the clear.

\paragraph{Generating the EOTS Signatures.}
Upon confirming a block $B$ for height $H$ within the Tendermint protocol, \ie, given $\text{decision}[H] = B$~\cite[Algorithm 1, line 49]{2018tendermint}, each validator $\validator$ sends a height $H$ \EOTS signature $\sigma_{H,B,\validator}$ for block $B$, if it had not already sent a height $H$ \EOTS signature.
Each \EOTS signature $\sigma_{H,B,\validator}$ by $\validator$ is an EOTS created with the secret signing key $\sk_\validator$ of the validator on the message $\mathsf{hash}(B)$ with private context $\sct_{H,\validator}$:
$$\sigma_{H,B,\validator} = \mathsf{EOTS\textnormal{-}Sign}(\sk_\validator, \sct_{H,\validator}, \mathsf{hash}(B))$$
(it can be verified with the corresponding public verification key $\pk_\validator = \mathsf{EOTS\textnormal{-}PK}(\sk_\validator)$ and public context $\pct_{H,\validator}$).
In other words, \EOTS signatures are EOTS with a message space of block hash values and a context space of heights $H \in \{0,1,\ldots\}$.
Other signatures used by Tendermint need not be EOTS and can be of any type.
For consistency with the Tendermint notation, we denote a height $H$ \EOTS signature for a block $B$ by $\langle \mathsf{Final}, H, \mathsf{hash}(B) \rangle$.
An honest validator sends a height $H$ \EOTS signature only after sending \EOTS signatures for the previous heights $1, \ldots, H-1$.
A client finalizes a block $B$ at height $H$ upon observing a quorum of $2f+1$ \emph{unique} height $H$ \EOTS signatures for block $B$, and after it has finalized blocks for all previous heights
(\Cref{alg:client-finality}).

\begin{algorithm}[t]
\caption{Key extraction by the forensic protocol.}
\label{alg:slashing-condition}
\begin{algorithmic}[1]\small
\Function{$\textsc{Key\textnormal{-}Extract}$}{$\mathrm{signatures}$}
    \Let{\mathsf{height}}{0}
    \Upon{$\langle \mathsf{Final}, H, \mathsf{hash}(B') \rangle_\validator\ \land\ \langle \mathsf{Final}, H, \mathsf{hash}(B) \rangle_\validator\ \land\ B \neq B'$}
        \State $\text{Identify } \validator \text{ as a protocol violator}$
        \Let{\sigma, \sigma'}{\langle \mathsf{Final}, H, \mathsf{hash}(B) \rangle, \langle \mathsf{Final}, H, \mathsf{hash}(B') \rangle}
        \Let{\sk_\validator}{\mathsf{EOTS\textnormal{-}Ext}(\pk_\validator, \mathsf{hash}(B), \sigma, \mathsf{hash}(B'), \sigma', H)}
    \EndUpon
    \State \Return $\sk_{\validator}$
\EndFunction
\end{algorithmic}
\end{algorithm}

\paragraph{Exposing the Adversarial Keys.}
The protocol identifies a validator as a protocol violator and returns its secret signing key upon receiving two \EOTS signatures created by the validator for the same height, but different blocks (\Cref{alg:slashing-condition}).

\paragraph{Security.}
We require Tendermint with the finality gadget to satisfy \emph{EOTS safety}, which captures the ability of the protocol to expose the secret keys of the adversarial validators after a safety violation.

\begin{definition}[EOTS safety]
\label{def:EOTS-safety}
A protocol provides EOTS safety with resilience $\fA$, if (i) when there is a safety violation, the secret signing keys of at least $\fA$ adversarial validators can be extracted by an efficient algorithm, and (ii) no efficient adversary can forge signatures on behalf of an honest validator. 
Such a protocol is said to provide \emph{$\fA$-EOTS safety}.
\end{definition}

Note that $(f+1)$-EOTS safety implies $(f+1)$-accountable safety (thus $f$-safety).

\begin{theorem}[EOTS Safety]
\label{thm:tendermint-accountability-extraction}
Tendermint with the finality gadget satisfies $(f+1)$-EOTS safety.
\end{theorem}
\Cref{thm:tendermint-accountability-extraction} follows from a quorum intersection argument and the $\mathsf{EXT\textnormal{-}SCMA}$ security of the EOTS.
Its proof is given in~\Cref{sec:appendix-proof-thm-daps-safety}.

Looking ahead, this theorem holds not only for protocols with a static validator set, but also dynamic ones, as long as the clients agree on the validator set for each height, which is ensured by the protocol in Section~\ref{sec:timestamping-protocol}.

Although the finality gadget ensures EOTS safety,
it does so by imposing a stronger \emph{finality condition}, \ie, the existence of $2f+1$ \EOTS signatures, as opposed to the original confirmation (decision) rule of Tendermint. 
We must thus ensure that the finality gadget retains the liveness of the Tendermint protocol under honest supermajority.
\begin{theorem}[Liveness]
\label{thm:tendermint-liveness}
If the number of adversarial validators is less than or equal to $f$, Tendermint with the finality gadget satisfies liveness with finite latency after GST.
\end{theorem}
The proof follows from the liveness of Tendermint and is given in~\Cref{sec:appendix-proof-thm-liveness}.

\paragraph{Performance and Discussion.}
Each validator has to send a single \EOTS signature per height.
This implies a linear communication complexity for these signatures in the number of heights and validators, a small overhead on top of the complexity of Tendermint (\cf~\Cref{sec:implementation} for concrete numbers).
Here, hierarchical deterministic wallets can be used to store a single EOTS key per validator.

Our finality gadget can be composed with any SMR consensus protocol with a fixed-sized validator set to equip the protocol with accountable and EOTS safety.
Therefore, our Bitcoin staking protocol can be instantiated with any such SMR protocol as the consumer chain.
Then, clients can choose between outputting the blocks as soon as they are confirmed, thus ensuring liveness in the absence of \EOTS signatures, or when they are attested by the \EOTS signatures, thus ensuring EOTS safety (\cf~\cite{NTT21,SWK+21,NTT22} for the nested ledger paradigm).

\section{Bond Contract}
\label{sec:bond-contract}

We now explain how the extracted keys of the adversarial validators are used to slash their stake for economic safety.
\paragraph{Bond Contract.} 
The validators stake their bitcoins in a \emph{bond contract} on Bitcoin.
These deposits remain locked for a \emph{predetermined} duration measured in the number of Bitcoin blocks (indefinitely for static stake), during which the validators must participate in the consumer chain's consensus.
Once the lock expires, each validator can retrieve its stake by sending it to an address it controls.
To enforce slashing, the bond contract must ensure that a validator's stake can only be sent to an unspendable output while it is locked.
One way to achieve this functionality is to use a timelock along with a \emph{covenant} to restrict the spending address until the lock expires.
Then, to slash a coin, it is sufficient to input a spending transaction (called the \emph{slashing transaction}) to Bitcoin, upon which the contract sends the coin to an unspendable address ($\mathsf{OP\_RETURN}$ output) specified by the covenant (see~\Cref{alg.bond.contract} in~\Cref{sec:appendix-bond-contract-covenant}).
If the validator's secret key is exposed, anyone can use the key to create a slashing transaction.
Hence, an exposed validator cannot avoid slashing, even if it colludes with some of the miners.

\paragraph{Covenant Emulation.}
As covenants are not enabled as part of the Bitcoin script, we instead emulate their functionality with a \emph{covenant committee}.
We structure the bond contract as an $(m+1)$-out-of-$(m+1)$ multi-signature, such that $m+1$ signatures by the $m$ members of the covenant committee and the staked validator are required to spend the deposit before the validator's duties end (Alg.~\ref{alg.bond.contract.musig}). 
The committee co-signs a slashing transaction at the creation time of the bond contract, so that anyone can complete and execute it if the validator's secret key is exposed. 

The covenant committee is trusted to never co-sign a different transaction collectively, as that would break the covenant. 
The committee members should ideally delete their signing keys after generating their signatures, to ensure that a future attacker cannot break the covenant, even if they compromise the committee. 
If at least one of the $m$ members is honest and manages to keep its signing key private (existential honesty), then the covenant becomes unbreakable.
The more committee members there are, the more plausible this assumption becomes. 
Natural choices of committee members are PoS chain users with high-value transactions.

The downside of covenant committees is that it is \emph{permissioned}: if only a single member is offline or refuses to participate, then the committee cannot complete its signature.
The chance of defection by a committee member increases as the committee size grows.
To recover from such attacks, the committee must be reformed to exclude the members that halt the signing process.
However, to sustain our $1$-out-of-$m$ assumption, an objective measure is required to distinguish between the case of a single malicious member halting the progress, and the case, where $m-1$ malicious members try to exclude the only honest member from the committee.
We can achieve this by requiring the committee members to publish their nonces, public keys and partial signatures on Bitcoin when the signature is not completed within some acceptable time-frame.
This allows all users to observe which committee members published a correct signature on time, and which ones refused to sign and thus must be excluded from the next signing attempt.

\paragraph{Efficiency of Emulated Covenants.} 
The committee can be represented by a space-efficient multi-signature scheme, \eg, MuSig2~\cite{nick2021musig2}. 
Then, the size of an emulated covenant on Bitcoin is around $100$ bytes, consisting of a $32$-bytes aggregate public key and a $64$-bytes signature.
Optimistically, all committee members are honest and the signature is promptly created off-chain. 
When parties must post their partial signatures to Bitcoin due to unresponsive members (worst-case), emulation requires $16$ kBytes for $m = 100$, assuming $32$-bytes keys, $64$-bytes signatures and two $32$-byte nonces per member for delinearization.
If these partial signatures are posted as $\mathsf{OP\_RETURN}$ transactions, which allow attaching $80$ arbitrary bytes to the transaction output~\cite{op_return}, this costs less than $40$ USD as of 4 Feb '26!\footnote{Size of $1$ $\mathsf{OP\_RETURN}$ transaction is $205$ bytes, it takes $\sim1.2$ satoshi per byte to have a transaction mined within six blocks with a latency of $1$ hour (on 13 Nov '24)~\cite{transaction-fee}, and the average Bitcoin price on that day was $\$71{,}327.00$ USD~\cite{btc-price}.}

\begin{algorithm}
\caption{The bond contract, emulated with a deleted-key covenant. The committee pre-signs with MuSig2.}
\label{alg.bond.contract.musig}
\begin{verbatim}
OP_IF
    <1 year>
    OP_CHECKLOCKTIMEVERIFY OP_DROP
OP_ELSE
    <committee_pubkey>
    OP_CHECKSIGVERIFY
OP_ENDIF

<validator_pubkey>
OP_CHECKSIG 
\end{verbatim}

\end{algorithm}

\paragraph{From EOTS safety to Economic Safety.}
We next prove economic safety for a static validator set.
\begin{definition}[Economic Safety]
\label{def:economic-safety}
A protocol provides economic safety with resilience $\fA$ ($\fA$-economic-safety), if (i) when there is a safety violation, Bitcoin stake of at least $\fA$ adversarial validators are slashed, and (ii) no honest validator is ever slashed (w.o.p.).
\end{definition}

\begin{theorem}[Economic Safety]
\label{thm:economic-safety}
The static-stake Bitcoin staking protocol using a covenant committee satisfies $(f+1)$-economic safety if one of the committee members is honest.
\end{theorem}
Proof is in~\Cref{sec:appendix-proof-thm-economic-safety} and follows from~\Cref{thm:tendermint-accountability-extraction}.

\begin{remark}
Even if the whole covenant committee is adversarial, it cannot slash the stake of an honest validator; since the adversarial committee members cannot forge a slashing transaction for the honest validators as argued above.
\end{remark}
By~\cite{roughgarden_talk}, no PoS blockchain secured only by its native stake can slash malicious validators after a safety violation if their fraction exceeds $2/3$.
In light of this, the Bitcoin staking protocol with a covenant committee improves on native PoS staking by ensuring that at least $1/3$ of the validators are slashed after a safety violation assuming honesty of one of the committee members.
Indeed, if the committee members are the same entities as the validators, our solution reduces the requirement of having over $1/3$ honest validators for slashing to having a \emph{single} honest validator.

\section{Timestamping Protocol}
\label{sec:timestamping-protocol}

We now extend the protocol 
to support a \emph{dynamic validator set}, which comes with two main challenges:
(i) ensuring agreement on the validator set at each height of the consumer chain as stake shifts hands, and (ii) slashing the adversarial validators' stake on Bitcoin after a safety violation on the consumer chain, but before their stake is unbonded.
Our protocol relies on the timestamps of the Bitcoin blocks \emph{within the consumer chain} to ensure agreement on the validator set; while using the timestamps of the consumer blocks \emph{within Bitcoin} to help slash the adversarial validators before they unbond.
We next state the security properties of the complete Bitcoin staking protocol and provide a high-level description highlighting its components.
The full protocol description and the algorithms are in~\Cref{sec:appendix-timestamping-protocol}.
Going forward, we will denote Bitcoin heights by $h$ and consumer chain heights by $H$ to avoid confusion.

\subsection{Economic Security}

\begin{theorem}[Economic Safety]
\label{thm:safety-smart-contracts}
The dynamic-stake Bitcoin staking protocol using a covenant committee satisfies $(f+1)$-economic safety if one of the committee members is honest.
\end{theorem}
To give intuition about the proof, in the sections below, we describe potential safety attacks exploiting the dynamic nature of the stake, and how the protocol prevents them.
We opted to highlight these specific attacks as they motivate different components of the Bitcoin staking protocol. 
These parts are in turn proven to be sufficient to ensure economic safety against \emph{all} attacks in~\Cref{sec:proof-of-thm-safety}.

\begin{theorem}[Liveness]
\label{thm:liveness-smart-contracts}
If the number of adversarial validators in any window of Bitcoin blocks is $\leq f$, the Bitcoin staking protocol satisfies liveness with finite latency under a synchronous network.
\end{theorem}
Proof of Theorem~\ref{thm:liveness-smart-contracts} is given in~\Cref{sec:proof-of-thm-liveness}.
It shows that when there are sufficiently few adversarial validators, the timestamps on Bitcoin do not affect the consumer chain output by the clients under a synchronous network.
Liveness then follows from the consumer chain's liveness.
Note that unlike the earlier results proven under partial synchrony, for liveness under a dynamic stake, we do require a synchronous network (see \Cref{sec:proof-of-thm-liveness} for an extended discussion).

\subsection{Overview of the Timestamping Protocol}

\paragraph{Bonding and Unbonding.}
To bond its stake, a validator locks it in a bond contract on Bitcoin via a \emph{bonding transaction}.
The stake remains locked for $K_a + k_u$ blocks, where $K_a$ denotes the predetermined period during which the staker must fulfill its validator duties toward the consumer chain.
The parameter $k_u = \Theta(k_c(m \cdot \Tconfirm) + k_f)$ is called the \emph{unbonding delay}, during which the stake stays locked even though the validator no longer participates in the consumer chain's consensus (see~\Cref{prop:bitcoin-secure} for the function $k_c(\Delta t)$ and the constant $k_f$).
Some minimal unbonding delay is necessary to tolerate delays in posting messages to Bitcoin, such as evidence of protocol violation. 

A proposer of consumer blocks must include the hash of the highest confirmed Bitcoin block it observes within its consumer block.
Let $b$ be the highest Bitcoin block (at some Bitcoin height $h$), whose hash is referred to by the consumer blocks of height lower than $H$.
Then, the validator set for a consumer block at height $H$ consists of the validators who have bonded their stake at Bitcoin blocks of height greater than $h-K_a$.
Thus, timestamps of the Bitcoin blocks within the consumer blocks ensure agreement on the validator set for each height of the consumer chain.

\paragraph{Generating the Context Tuples.} Prior to bonding its stake, each validator $\validator$ generates $(\sct_{H,\validator}, \pct_{H,\validator}) \gets \mathsf{EOTS\textnormal{-}ContextGen}(1^{\kappa})$ for an estimated window of consumer chain heights and publishes the public components $\pct_{H,\validator}$ on Bitcoin.
It periodically posts new $\pct_{h,\validator}$ during the $K_a$ Bitcoin blocks of its time as a validator.

Formally, let $w$ denote the height of the consumer chain right before $\validator$ inputs its bonding transaction to Bitcoin.
Then, $\validator$ initially generates $\pct_{H,\validator}$ for the heights $H \in \{w+1, \ldots, w+w_0\}$.
In the worst case, $\validator$ is selected as a validator for a consumer chain height greater than $w$, and some context components remain unused.
However, in any case, $\validator$ will be able to generate EOTS signatures that verify with respect to the $\pct_{H,\validator}$ on Bitcoin for all the consumer chain heights it serves as a validator.

\paragraph{Timestamping on Bitcoin.}
\label{sec:smart-contracts-timestamping-protocol}
Validators send periodic timestamps of the consumer chain to Bitcoin.
These timestamps consist of the hash of the timestamped consumer block, a quorum of $2f+1$ \EOTS signatures on this hash, and the block's height.
Timestamps can be sent for every consumer block, or for every $m$-th consumer block for some $m > 1$.
When there is a safety violation on the consumer chain, and a client observes conflicting consumer blocks finalized by the finality gadget, blocks with earlier timestamps take precedence over those with later timestamps on Bitcoin.

Timestamping of the consumer blocks on Bitcoin is crucial for ensuring economic security when adversarial validators create conflicting consumer chain forks after unbonding their Bitcoin stake (long-range attacks); as the later fork can now be ignored by the clients due to its later timestamp.

Note that the timestamping protocol does not affect the latency of the consumer chain: each consumer block is finalized as soon as a quorum of $2/3$ \EOTS signatures is obtained on it.
The timestamps are rather used to preserve economic safety in the case of attacks via the stopping rules below, which are triggered only in the event of successful attacks on the consumer chain by a majority of Bitcoin-staked validators.

\paragraph{Stopping Rules.}
The timestamping protocol also imposes \emph{stopping rules} for the clients.
These rules require clients to stop adding new consumer blocks to their ledgers and to send timestamps to Bitcoin for the latest finalized consumer blocks in their views:

\smallskip
\noindent
\textbf{Safe-stop rule 1:} Clients stop outputting new consumer blocks when they observe a so-called \emph{data-unavailable} timestamp on Bitcoin such that the underlying, timestamped consumer blocks are not available in the client's view.

\smallskip
\noindent
\textbf{Block output rule:} Clients do not output consumer blocks that were confirmed by an old validator sets, determined by an old Bitcoin block $b$, unless these consumer blocks were timestamped on Bitcoin within a few blocks of $b$.

\smallskip
\noindent
\textbf{Safe-stop rule 2:} Clients do not output the consumer blocks timestamped by some old Bitcoin block $b'$, if these consumer blocks conflict with other consumer blocks whose timestamp appears within a few Bitcoin blocks of $b'$.

\paragraph{Slashing.}
\label{sec:smart-contracts-slashing}

In certain cases, clients send timestamps to Bitcoin in addition to the periodic timestamps 
to warn other clients about potential safety violations.
If there were indeed a safety violation, these extra timestamps ensure the slashing of the adversarial validators.
Details about these conditions can be found in~\Cref{sec:appendix-timestamping-protocol}.

\subsection{Intuition behind the Stopping Rules}
In this section, we motivate the stopping rules above by describing how they mitigate certain types of attacks.
Although we showcase these specific attacks, these rules suffice to ensure the protocol's security against \emph{all} attacks as proven by~\Cref{thm:safety-smart-contracts,thm:liveness-smart-contracts}.
\subsubsection{Data availability attack and safe-stop rule 1}
\label{sec:stopping-rule-1}

\begin{figure*}
{
    \centering
    \includegraphics[width = 0.9\linewidth]{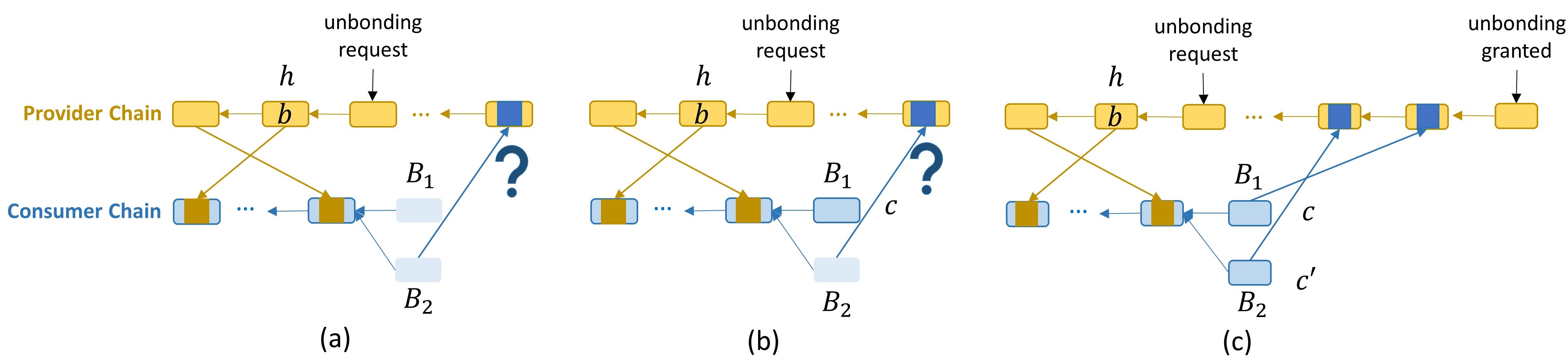}
    \caption{Illustration of the data availability attack and the safe-stop rule 1 (\cf Section~\ref{sec:stopping-rule-1}). Yellow squares within the consumer blocks represent the hashes of Bitcoin blocks. 
    Similarly, blue squares within Bitcoin blocks (called provider chain for symmetry) represent the timestamps of the consumer blocks.
    Light blue blocks denote unavailable consumer blocks.}
\label{fig:safe-stop-1}
}
\end{figure*}

If the clients observe a data-unavailable timestamp on Bitcoin such that the consumer block $B$ of the timestamp or a block in $B$'s prefix is (partially or fully) unavailable or not finalized, they stop outputting new consumer blocks to prevent non-slashable safety violations (safe-stop rule 1).
This so-called data availability attack was first discussed in~\cite{btc-pos} and is illustrated in Fig.~\ref{fig:safe-stop-1}.
In the figure, the Bitcoin block at height $h$ denotes the highest Bitcoin block referred by the earlier finalized consumer blocks.
Over $2/3$ of the validators specified by $h$ are adversarial and create two conflicting consumer blocks, $B_1$ and $B_2$.
They subsequently send a timestamp to Bitcoin for $B_2$, but initially keep both blocks private (Fig.~\ref{fig:safe-stop-1}-a).
Then, the adversary reveals $B_1$ to a client $\client$, but keeps $B_2$ hidden (Fig.~\ref{fig:safe-stop-1}-b).
At this point, if $\client$ outputs $B_1$ as part of its ledger, it would cause a safety violation.
This is because the adversary can reveal both blocks after unbonding its stake to a late-coming client $\client'$, which would output $B_2$ instead of $B_1$, as blocks with earlier timestamps take precedence over those with later ones (Fig.~\ref{fig:safe-stop-1}-c).
Moreover, the adversary cannot be slashed as it has already unbonded.
Hence, to prevent non-slashable safety violations, upon observing data-unavailable timestamps, clients stop adding new consumer blocks to their ledgers.

\subsubsection{Escaping stake attack and block output rule}
\label{sec:stopping-rule-1-5}

\begin{figure*}
{
    \centering
    \includegraphics[width = 0.9\linewidth]{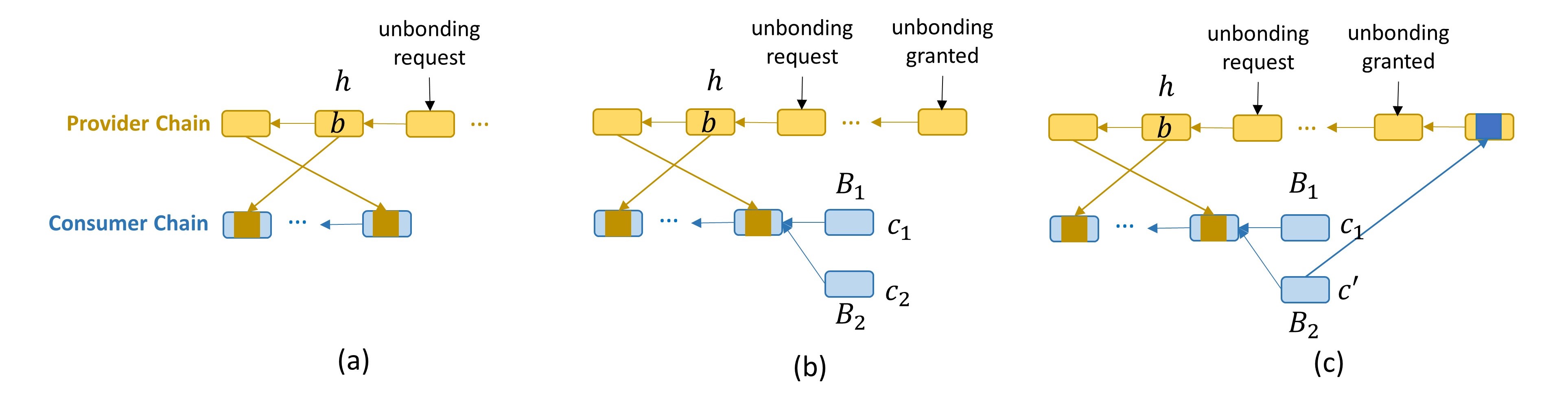}
    \caption{Illustration of the escaping stake attacks and the block output rules (\cf Section~\ref{sec:stopping-rule-1-5}).}
\label{fig:safe-stop-1-5}
}
\end{figure*}

We next describe \emph{escaping stake attacks} that exploit the fact that the stake is maintained on Bitcoin rather than the consumer chain.
Again, suppose over $2/3$ of the validators specified by $b$ are adversarial.
In the first attack, the adversarial validators send an unbonding request to Bitcoin (Fig.~\ref{fig:safe-stop-1-5}-a), and once the request is granted, create two conflicting finalized consumer blocks $B_1$ and $B_2$ (Fig.~\ref{fig:safe-stop-1-5}-b).
They show the blocks $B_1$ and $B_2$ to the clients $\client_1$ and $\client_2$ respectively, yet, keep block $B_2$ hidden from $\client_1$ and vice versa.
At this point, if the clients choose to output their respective blocks, then they risk a non-slashable safety violation, as the adversarial validators have unbonded (\ie, the stake has \emph{escaped}).
In the second attack on Fig.~\ref{fig:safe-stop-1-5}-c, a late-coming client $\client'$ outputs $B_2$ upon observing its timestamp, thus conflicting with $\client_1$ that has output $B_1$ before, after the adversarial stake has escaped.

To avoid these attacks, clients refuse to output consumer blocks finalized by old validator sets determined by an old Bitcoin block $b$, if these consumer blocks have not been timestamped on Bitcoin shortly after $b$.
Thus, they would reject blocks $B_1$ and $B_2$, as block $b$ has become too deep in Bitcoin by the time the clients observe $B_1$ and $B_2$ (Fig.~\ref{fig:safe-stop-1-5}-b).
Similarly, 
client $c'$ would reject block $B_2$, as its timestamp appears long after block $b$ (Fig.~\ref{fig:safe-stop-1-5}-c). 
Indeed, if the majority of validators were honest, a consumer block's timestamp would appear on Bitcoin, long before Bitcoin grows by more than $k_u$ blocks, the unbonding delay.
In~\Cref{sec:stopping-rule-2}, we discuss a similar type of attack called the mismatched timestamp attack and how it is prevented by the safe-stop rule 2.

\section{Implementation and Measurements}
\label{sec:implementation}

\subsection{CPU and Memory Usage}

Bitcoin staking has been implemented and running continuously in production for over a year to secure an anonymous blockchain. In this system, $60$ Bitcoin-staked validators (called finality providers) secure a Tendermint blockchain. We use a version of this open-source implementation to evaluate the operational costs of running a validator, measured in its CPU and memory usage. 
The implementation covers the entire lifecycle of a validator, including submitting a Bitcoin transaction to lock up funds, maintaining EOTS key pairs, monitoring the Tendermint (consumer chain) consensus protocol, and creating \EOTS signatures.

To simulate a production environment, we set up an end-to-end testbed with the following components: a private Bitcoin blockchain running \texttt{bitcoind}, a Tendermint blockchain implemented using the Cosmos SDK, a covenant committee, a monitoring program that slashes equivocating validators, and a validator. The components reside on the same physical server, and are isolated in separate Docker containers. We configure Tendermint to produce a block every 5 seconds. The validator communicates with a Tendermint blockchain client, a \emph{separate process}, and signs each block confirmed by it. In steady state, the validator uses $179$ MB of memory, and less than $10\%$ of a core on a Xeon E5 2698 v4 CPU on top of the resources consumed by the original Tendermint blockchain client.
This implies that the validator is lightweight, and fits in even the smallest cloud VM instances. Further optimization is possible in a less portable implementation.

Clients of the consumer chain use the chain's original confirmation rule and the quorum of the EOTS signatures together to finalize blocks.
Thus, both the PoS validators and the clients run light clients of Bitcoin to verify that the EOTS signatures correspond to the Bitcoin addresses with stake. 
This adds little overhead for the clients (and validators) as Bitcoin is light-weight (24 MB/hour) compared to most PoS protocols (for Ethereum, 300 MB/hour).

In our experiment, the covenant committee consists of $9$ members, and the covenant is emulated by a $6$-out-of-$9$ multisig.
Memory and compute requirements for covenant committee members are negligible compared to validators.

\subsection{Latency, Unbonding Delay and Cost}

The Bitcoin staking protocol increases the confirmation latency of the consumer chain by only one round of communication ($\delta$ time given an actual network delay of $\delta$).

An important contribution of the timestamping protocol is reducing the unbonding delay from weeks as in many PoS blockchains~\cite{cosmos_delay} to a matter of hours.
A careful security analysis in~\Cref{sec:appendix-timestamping-protocol,sec:proof-of-thm-safety,sec:proof-of-thm-liveness} shows that, when the timestamps of consumer blocks are sent to Bitcoin as soon as they are finalized, setting the unbonding delay to $k_u = 2 k_c(m \cdot \Tconfirm) + 4 k_f$ suffices to satisfy Theorems~\ref{thm:safety-smart-contracts} and~\ref{thm:liveness-smart-contracts}.
Since our protocol uses the confirmed Bitcoin chain, we make the confirmation delay $k$ explicit in the expression for $k_u = 2 k_c(m \cdot \Tconfirm) + 4 k_f + k$, re-defining $k_f$ as Bitcoin's transaction inclusion delay, and $k_c(m \cdot \Tconfirm)$ as the worst-case Bitcoin chain growth during one consumer chain period of $m$ heights (\Cref{prop:bitcoin-secure}).

Note that reducing $k_c(m \cdot \Tconfirm)$ by having shorter consumer chain periods (\ie, smaller $m$) implies frequent timestamps on Bitcoin (\Cref{sec:appendix-timestamping-protocol}), which in turn means a larger cost for the consumer chain in transaction fees.
Similarly, reducing $k_f$ by ensuring faster inclusion in Bitcoin implies larger transaction fees for each timestamp.

Each timestamp consists of the SHA256 hash of the timestamped consumer block ($32$ bytes), an aggregate BLS signature on this hash, and the block height ($8$ bytes).
The aggregate signature has size $48$ bytes along with a bit map of signers ($8$ bytes for $n = 60$ validators).
Then, each timestamp consists of $96$ bytes, which can be posted to Bitcoin as part of two $\mathsf{OP\_RETURN}$ transactions, each of $205$ bytes.

We demonstrate the dependency between cost and unbonding delay by Table~\ref{tab:measurements}.
Achieving $k_f = 1$\footnote{$k_f$ and $k$ were originally measured in the number of newly mined Bitcoin blocks. For simplicity, we report their wall-clock time equivalents.} (high transaction fee) and $6$ (low transaction fee) hours cost respectively $2.01$ and $0.27$ satoshis per Byte in transaction fees~\cite{transaction-fee}.
A frequency of timestamping every $1$ (high frequency) and $6$ hours (low frequency) imply $8760$ and $1460$ timestamps per year respectively (in these cases, $k_c(m \cdot \Tconfirm)$ corresponds to $1$ and $6$ hours in Bitcoin height in the $k_u$ equation).
We take the $k=20$ blocks deep prefix of the longest Bitcoin chain in our calculations, which corresponds to an extra $\sim3$ hours of delay for unbonding.
Depth $20$ was chosen to have a low probability ($10^{-7}$) of safety violation~\cite[Figure 2]{btc-latency}\footnote{For an adversary with $10\%$ hash rate and network delay $\Delta = 10$ sec.}.
\begin{table}[h]
\centering
\begin{tabular}{|c|c|c|}
\hline
  & high BTC tx fee & low BTC tx fee  \\
\hline
 \makecell{high \\ timestamp \\ frequency} & \makecell{cost $\approx 5,500$ USD \\ delay $\approx 9$ hrs} & \makecell{cost $\approx 910$ USD \\ delay $\approx 29$ hrs} \\
\hline
 \makecell{low \\ timestamp \\ frequency} & \makecell{cost  $\approx 740$ USD \\ delay $\approx 19$ hrs} & \makecell{cost $\approx 120$ USD \\ delay $\approx 39$ hrs} \\
\hline
\end{tabular}
\caption{Unbonding delays and the associated yearly cost in transaction fees given the Bitcoin price ($\$75,825$), and transaction fees on April 29th~\cite{btc-price}. Unbonding delay is $k_u = 2 k_c(m \cdot \Tconfirm) + 4 k_f + k$.}
\label{tab:measurements}
\end{table}

\section{Security Cost for Bitcoin staking}
\label{sec:cost}

\begin{figure}[h]
    \centering
    \includegraphics[width=\linewidth]{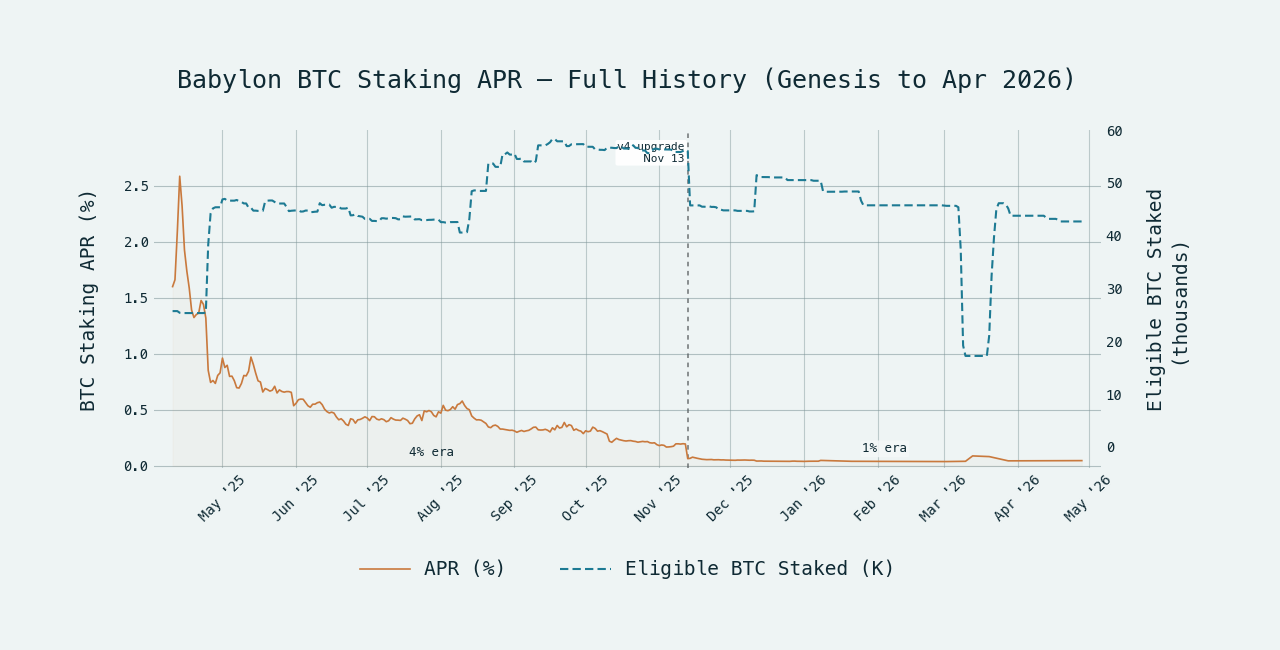}
    \caption{APRs for Bitcoin staking on Babylon PoS chain. Staking rewards for the Bitcoin stakers comes from the inflation of the BABY token, the native token of the Babylon chain. A governance change in Nov 2025 reduces the original inflation rate of $4\%$ to $1\%$.}
\label{fig:history}
\end{figure}

Security is not free. In a PoW chain, miners have costs such as electricity and hardware. In a PoS chain, stakers have capital costs. Chains pay for these costs of security by programmatically generating rewards for the miners or stakers. These rewards usually come from the inflation of the native token of the chain or the transaction fees. An interesting question is how does the security cost of Bitcoin staking compare to PoS staking secured by the native token?

Data from Babylon \cite{btc-staking}, the first PoS chain which adopts Bitcoin staking, sheds some light on this question. Figure \ref{fig:security_cost} compares current staking reward rates (APRs) of the top $20$ PoS chains. The APR for Bitcoin staking is 2 to 3 orders of magnitude smaller than all the other chains which are secured by their native token. Figure \ref{fig:history} shows the APR from the time of the launch of Babylon to April 2026. Interestingly, even though the APR decreases over time, the amount of Bitcoin stake increases steadily over time.
This data suggests that from a PoS blockchain point of view,  Bitcoin staking is a very economical way to obtain a high level of economic security compared to staking using native tokens.

\section*{Acknowledgements}
The work of ENT was conducted mostly while at Stanford University.

\bibliographystyle{plain}
\bibliography{pubs}

\appendix

\section{Ethical Considerations}

We follow the stakeholder-based ethics analysis and discuss impacts from both (i) our research procedures (design, implementation, and evaluation), and (ii) publication of the protocol and results.

\paragraph{Stakeholder identification.}
Direct stakeholders include (1) the authors of the paper, (2) implementers and deployers of Bitcoin staking protocol as outlined in the paper, and (3) any entities operating components of the deployed protocol such as covenant committees.
Indirect stakeholders include (4) Bitcoin holders who may stake, (5) end users and applications of PoS consumer chains secured by remote Bitcoin stake, (6) the broader Bitcoin ecosystem (miners and full-node operators) affected by additional transaction load, and (7) regulators and auditors responsible for oversight. Finally, (8) adversaries are also stakeholders insofar as they may attempt to exploit or misuse the protocol.

\paragraph{Research procedures and data.}
Our work involves no interaction with human subjects and no collection of non-public personal data. 
Our empirical evaluation uses a controlled testbed deployment and public blockchain data in aggregated form (e.g., staking totals and reward rates), rather than attempting to identify individuals behind addresses. 
We did not conduct experiments intended to degrade or disrupt third-party production systems.

\paragraph{Potential benefits.}
Bitcoin staking is intended to improve the security of PoS systems by enabling “security sharing” with Bitcoin while providing automatic, on-chain economic penalties for provable safety violations.

\paragraph{Potential harms and risks.}
\textit{Financial harms to stakers and operators.} The protocol is explicitly designed to slash misbehaving validators by enabling extraction of a validator secret key after certain safety violations. 
This creates the intended economic penalty, but also raises risks: (i) operational mistakes could trigger slashing; (ii) implementations may contain bugs that lead to incorrect behavior; and (iii) if a validator reuses the same Bitcoin key material outside of staking, exposure of that key could unintentionally endanger unrelated funds. 
Even when the protocol’s intended slashing sends funds to an unspendable destination, accidental key exposure can still be harmful if key reuse occurs.

\textit{Harms to consumer-chain users.} Defensive ``safe-stop" rules (used to avoid non-slashable safety violations) can reduce availability during adversarial conditions, potentially causing outages or delayed finality for applications built on the consumer chain.

\textit{Trust and centralization risks.} When covenants are emulated via a covenant committee, additional trust assumptions and operational dependencies are introduced. Committee members may face coercion or targeted attacks, and committee unavailability can complicate operations (e.g., reformation procedures). Even if the committee cannot unilaterally steal funds, reliance on committee behavior can shift power and risk toward a small set of actors.

\textit{Ecosystem externalities.} Timestamping and related transactions may increase load on Bitcoin block space and fees, creating a small negative externality for other Bitcoin users.
Additionally, strengthening the security of PoS systems is dual-use: it can support socially beneficial applications, but may also support harmful or illicit applications that prefer robust infrastructure.

\paragraph{Mitigations.}
We mitigate these risks in several ways.
First, the protocol design aims to prevent slashing of honest validators under the stated assumptions (``trustless staker safety''), and uses explicit unbonding delays and timestamping rules intended to avoid scenarios where stake can escape before accountability and slashing can be enforced.
Second, to reduce unintended losses from key exposure, we recommend that deployments enforce \emph{key separation}: validators should use dedicated keys and addresses for staking only, and avoid address/key reuse for other funds. Deployments should adopt standard validator operational safeguards (e.g., single-instance signing policies and slashing-protection tooling) to reduce accidental double-signing.
Third, where covenant committees are used, deployers should minimize concentrated trust by choosing diverse and independent committee members, using well-vetted threshold multi-signature constructions, and applying strong operational security practices (including key deletion after presigning where applicable). Operational procedures should be documented so that failures or refusal-to-sign events can be detected and addressed transparently.
Fourth, we describe assumptions, tradeoffs, and parameter dependencies to reduce the chance of unsafe deployment, and we encourage independent security review before custodying significant value.

\paragraph{Decision to conduct and publish.}
We proceeded with this research because remote staking and Bitcoin-based security sharing are already active areas of deployment, and publishing a clear protocol specification, threat model, and analysis enables broader scrutiny, auditing, and safer implementations. We judge that the benefits of transparency and improved economic safety outweigh the risks of misuse, provided deployers follow strong operational practices and comply with applicable law and regulation. Nonetheless, we recognize that misuse and unsafe deployment remain possible; we therefore emphasize the above mitigations and the importance of independent review.

\section{Open Science}

An anonymized version of the open-source implementation of the Bitcoin staking protocol, discussed in \Cref{sec:implementation}, is available at \url{https://anonymous.4open.science/r/ccs-2026-657F/}.

\section{Security Proofs for Extractable One-Time Signatures}
\label{sec:eots-proof}

We next provide the security definitions and proofs for the proposed EOTS scheme.

\paragraph{Correctness}
\begin{definition}[Correctness]\label{def:eots-correctness}
An EOTS scheme is \emph{correct} if for every $\kappa \in \mathbb{N}$
and every message $m \in \mathcal{M}$, the following probability is $1$.
\begin{align*}
\Pr\!\left[
\mathsf{EOTS\text{-}Verify}(\pk, m, \pct, \sigma) = 1
\;\middle|\;
\begin{aligned}
&\sk\sample \mathsf{EOTS\text{-}KeyGen}(1^\kappa) \\
&\pk \leftarrow \mathsf{EOTS\text{-}PK}(\sk) \\
&(\sct, \pct)\sample \mathsf{EOTS\text{-}ContextGen}(1^\kappa) \\
&\sigma\sample \mathsf{EOTS\text{-}Sign}(\sk, \sct, m)
\end{aligned}
\right]
\end{align*}
\end{definition}

\begin{theorem}[Correctness of the EOTS Construction]
\label{thm:eots-correctness}
The EOTS scheme of Section~4 satisfies Definition~\ref{def:eots-correctness}.
\end{theorem}

\begin{proof}
Let $\sk = x \in \mathbb{Z}_q^\times$ be sampled by $\mathsf{EOTS\text{-}KeyGen}(1^\kappa)$, so that $\pk = g^x$.
Let $(\sct, \pct) = (k, \hctx(g^k)) \sample \mathsf{EOTS\text{-}ContextGen}(1^\kappa)$ with $k \in \mathbb{Z}_q^\times$.
Then for any message $m \in \{0,1\}^*$, $\mathsf{EOTS\text{-}Sign}(x, k, m)$ outputs $\sigma = (r, s)$ with $r = g^k$ and $s = k + x \cdot \hsig(r \| m) \bmod q$.
We verify that $\mathsf{EOTS\text{-}Verify}(\pk, m, \pct, \sigma) = 1$.

The verification algorithm computes
\[
r_v \;=\; g^s \cdot \bigl(\pk^{\hsig(r \| m)}\bigr)^{-1}
\;=\; g^{k + x \cdot \hsig(r \| m)}
\cdot g^{-x \cdot \hsig(r \| m)}
\;=\; g^k \;=\; r,
\]
so the second check $r_v \stackrel{?}{=} r$ passes.

Moreover, $\hctx(r) = \hctx(g^k) = \pct$ by construction of $\pct$ in $\mathsf{EOTS\text{-}ContextGen}$, so the first check $\hctx(r) \stackrel{?}{=} \pct$ passes as well.  
Hence $\mathsf{EOTS\text{-}Verify}(\pk, m, \pct, \sigma) = 1$ with probability $1$.
\end{proof}

\paragraph{\textsc{PK-DER} Security}
Besides the other properties, we additionally prove \textsf{PK-DER} security, which states that the algorithm \textsf{EOTS-PK} works as expected: no adversary can come up with a key pair ($\sk$, $\pk$) such that (i) \textsf{EOTS-PK} produces $\pk$ on input $\sk$, and (ii) the adversary is able to produce verifying signatures, yet (iii) $\sk$ does not enable producing verifying signatures via the signing algorithm.
This property is necessary when the keys are not generated through a trusted setup.

\begin{algorithm}
\caption{Game for Public Key Derivation Security (\textsf{PK-DER})}
\label{daps:game:pk}
\begin{algorithmic}[1]\small
\State $(\sk, \pk, \pct^*, m_1, \sigma, m_2) \gets \mathcal{A}$
\State $(\sct, \pct) \sample \mathsf{EOTS\textnormal{-}CtxGen}(1^{\kappa})$
\State \Return $\mathsf{EOTS\textnormal{-}PK}(\sk) = \pk \wedge \mathsf{EOTS\textnormal{-}Verify}(\pk,
m_1, \pct^*, \sigma) \wedge \neg \mathsf{EOTS\textnormal{-}Verify}(\pk, m_2, \pct, \mathsf{EOTS\textnormal{-}Sign}(\sk, \sct, m_2))$
\end{algorithmic}
\end{algorithm}

\begin{definition}[\textsf{PK-DER} security]
\label{daps:sec:pk} The \textsf{PK-DER} game formalizes the public key derivation security. 
\begin{equation*}
\forall \mathcal{A} \in \mathsf{PPT},\ 
\Advantage^{\mathsf{PK\text{-}DER}}_{\mathsf{EOTS}}(\Adv) \deq \Pr[\mathsf{PK\textnormal{-}DER}^{\mathcal{A}}(1^{\kappa}) = 1]
< \negl
\end{equation*}
\end{definition}

\begin{theorem}[\textsf{PK-DER} security]
\label{thm:der}
The EOTS scheme of Section~\ref{sec:daps-constr} is \textsf{PK-DER}
secure.
\end{theorem}

\begin{proof}
Consider the output $(x, y, \hashctx_{r^*}, m_1, \sigma, m_2)$ of a \textsf{PPT} adversary $\mathcal{A}$.
We analyze the probability that $\mathcal{A}$ wins the \textsf{PK-DER} game.
Note that for any tuple $(k, \hashctx_r) \sample \mathsf{EOTS\textnormal{-}CtxGen}(1^{\kappa})$, it must be $\hashctx_r = \hctx(g^k)$.
Suppose $\mathsf{EOTS\textnormal{-}PK}(x) = y$, which implies $y = g^x$.
Then, given
\begin{equation*}
\mathsf{EOTS\textnormal{-}Sign}(x, k, m_2) = (g^k, k+x\hsig(g^k||m_2))
\end{equation*}
it holds that
\begin{align*}
&\ g^{\mathsf{EOTS\textnormal{-}Sign}(x, k, m_2)}(y^{\hsig(g^k || m_2)})^{-1} \\
=&\ g^{k+x\hsig(g^k || m_2)}(g^{x\hsig(g^k || m_2)})^{-1} = g^k
\end{align*}
and $\hctx(g^k) = \hashctx_r$, which implies
\begin{equation*}
\mathsf{EOTS\textnormal{-}Verify}(y, m_2, \hashctx_r, \mathsf{EOTS\textnormal{-}Sign}(x, k, m_2)) = 1
\end{equation*}
In other words, $\mathcal{A}$ wins the game with probability $1$.
\end{proof}

Note that the secret key in our construction can be any number in $\mathbb{Z}_q^{\times}$ and all operations are done in $\mathbb{Z}_q^{\times}$ or $G$.
Therefore it is impossible for $\mathcal{A}$ to provide an invalid key pair that passes the $\mathsf{EOTS\textnormal{-}PK}(x) = y$ assertion but fails the $\mathsf{EOTS\textnormal{-}Verify}(y, r, m_1, s)$ assertion, thus the latter assertion is redundant for the proof of Theorem~\ref{thm:der}.

\begin{algorithm}
\caption{Game for Strong Existential Unforgeability under 
Chosen Message Attacks with Distinct Contexts (\textsf{sEUF-CMADC})}
\label{daps:game:seuf-cmadr}
\begin{algorithmic}[1]\small
\State $\sk \sample \mathsf{EOTS\textnormal{-}KeyGen}(1^{\kappa})$; $\pk
\gets \mathsf{EOTS\textnormal{-}PK}(\sk)$
\For{$i = 1$ to $N$} \Comment{Here, $N$ is $\poly(\kappa)$.}
    \State $(\sct_i, \pct_i) \sample \mathsf{EOTS\textnormal{-}CtxGen}(1^{\kappa})$
\EndFor
\State $Q \gets \emptyset$
\State $\Qctx \gets \emptyset$
\State $(m^*, \pct^*, \sigma^*) \gets \mathcal{A}^{\OEOTS(\cdot),\OEOTSctx(\cdot)}(\pk, \{\pct_i\}_{i=1,\ldots,N})$ \Comment{$\mathcal{A}$ can call the oracle multiple times}
\State \Return $(m^*, \pct^*, \sigma^*) \notin Q \wedge \mathsf{EOTS\textnormal{-}Verify}(\pk, \pct^*, m^*,
\sigma^*)$
\label{daps:game:seuf-cmadr:return}
\Statex

\Function{$\OEOTSctx$}{$m, \pct$} \Comment{Signing oracle with context}
  \If{$\pct \in \{\pct_i\}_{i = 1, \ldots, N}$ and $\pct \notin \Qctx$}
    \State $\sigma \sample \mathsf{EOTS\textnormal{-}Sign}(\sk, \sct_{i'}, m)$ where $\pct = \pct_{i'}$ \Comment{$\sct_{i'}$ is the secret context component corresponding to $\pct_{i'}$.}
    \State $Q \gets Q \cup \{(m, \pct, \sigma)\}$
    \State $\Qctx \gets \Qctx \cup \{\pct\}$
    \State \Return $(\pct, \sigma)$
  \Else
    \State \Return $\bot$
  \EndIf
\EndFunction
\Statex

\Function{$\OEOTS$}{$m$} \Comment{Signing oracle without context}
    \State $(\sct, \pct) \sample \mathsf{EOTS\textnormal{-}CtxGen}(1^{\kappa})$
    \State $\sigma \sample \mathsf{EOTS\textnormal{-}Sign}(\sk, \sct, m)$
    \State $Q \gets Q \cup \{(m, \pct, \sigma)\}$
    \State \Return $(\pct, \sigma)$
\EndFunction
\end{algorithmic}
\end{algorithm}

\paragraph{\textsf{sEUF-CMADC} Security}

\begin{definition}[\textsf{sEUF-CMADC} security] The \textsf{sEUF-CMADC} game formalizes the existential unforgeability guarantee for the EOTS scheme. Namely, for all PPT $\Adv$,
\label{eots:sec:seuf-cmadc}
\begin{equation*}
\Advantage^{\mathsf{sEUF\text{-}CMADC}}_{\mathsf{EOTS}}(\Adv) \deq
\Pr[\mathsf{sEUF\textnormal{-}CMADC}^{\mathcal{A}}(1^{\kappa}) = 1]
< \negl
\end{equation*}
\end{definition}
This security property guarantees that signatures are, except with negligible probability, unforgeable when the secret key is unknown, even in the face of multiple signature queries.

\begin{algorithm}
\begin{algorithmic}[1]\small
\caption{Game for Strong Existential Unforgeability under Chosen Message Attacks (\textsf{sEUF-CMA})}
\label{daps:game:seuf-cma}
\State $(\sk, \pk) \sample \mathsf{Schnorr\textnormal{-}KeyGen}(1^{\kappa})$
\State $\mathcal{R} \gets \emptyset$
\State $(m^*, \sigma^*) \gets \mathcal{A}^{{\OSch}(\cdot)}(\pk)$
\State \Return $(m^*, \sigma^*) \notin \mathcal{R} \wedge \mathsf{Schnorr\textnormal{-}Verify}(\pk, m^*, \sigma^*)$
\label{daps:game:seuf-cma:return}
\Statex

\Function{$\OSch$}{$m$}
  \State $\sigma \sample \mathsf{Schnorr\textnormal{-}Sign}(\sk, m)$
  \State $\mathcal{R} \gets \mathcal{R} \cup \{(m, \sigma)\}$
  \State \Return $\sigma$
\EndFunction
\end{algorithmic}
\end{algorithm}

\begin{algorithm}
\caption{Schnorr key generation algorithm: $\mathsf{Schnorr\textnormal{-}KeyGen}()$}
\label{daps:schnorr:keygen}
\begin{algorithmic}[1]\small
\State $x \sample U(\mathbb{Z}^{\times}_q)$
\State \Return $(x, g^x)$
\end{algorithmic}
\end{algorithm}

\begin{algorithm}
\caption{Schnorr signature algorithm: $\mathsf{Schnorr\textnormal{-}Sign}(x, m)$}
\label{daps:schnorr:sign}
\begin{algorithmic}[1]\small
\State $k \sample U(\mathbb{Z}^{\times}_q)$
\State $r \gets g^k$
\State $e \gets \mathcal{H}(r||m)$
\State $s \gets k+xe$
\State \Return $(r, s)$
\end{algorithmic}
\end{algorithm}

\begin{algorithm}
\caption{Schnorr verification algorithm: $\mathsf{Schnorr\textnormal{-}Verify}(y, m, (r, s))$}
\label{daps:schnorr:verify}
\begin{algorithmic}[1]\small
\State $r_v \gets g^s(y^{\mathcal{H}(r || m)})^{-1}$
\State \Return $r_v = r$
\end{algorithmic}
\end{algorithm}

\medskip
\noindent\textbf{Assumption.} 
For the security proof, we assume that the Schnorr signature scheme with the hash function $\hsig$ and group $\G$ of prime order $q$ is \textsf{sEUF-CMA} secure in the random oracle model whenever the hardness of discrete log holds (the algorithms are defined in~\Cref{daps:schnorr:keygen,daps:schnorr:sign,daps:schnorr:verify}). Namely, for all PPT $\Adv$,
\[
\Advantage^{\mathsf{sEUF\text{-}CMA}}_{\mathsf{Schnorr}}(\mathcal{A}) = \Pr[\mathsf{sEUF\textnormal{-}CMA}^{\mathcal{A}}(1^{\kappa}) = 1]
\;\le\; \negl.
\]

\begin{theorem}
\label{thm:sEUF-CMADC}
Under the above assumption, the EOTS scheme of Section~4 is
sEUF-CMADC secure: for every PPT adversary $\mathcal{B}$, there exists a PPT adversary $\Adv$ such that
\[
\Advantage^{\mathsf{sEUF\text{-}CMADC}}_{\mathsf{EOTS}}(\mathcal{B})
\;\le\;
\Advantage^{\mathsf{sEUF\text{-}CMA}}_{\mathsf{Schnorr}}(\mathcal{A})
\;+\; \negl
\]
\end{theorem}

\begin{proof}[Proof of~\Cref{thm:sEUF-CMADC}]
Let $\mathcal{B}$ be a PPT adversary winning the \textsf{sEUF-CMADC} game (\Cref{daps:game:seuf-cmadr}) with probability $\varepsilon(\kappa)$.  Let
$q_H, q_{\mathsf{ctx}}, q_{\mathsf{nc}}$ be upper bounds, polynomial in $\kappa$,
on the number of $\mathcal{B}$'s queries to $\hctx$, to $\OEOTSctx$,
and to $\OEOTS$ respectively, and set
$Q_S := q_{\mathsf{ctx}} + q_{\mathsf{nc}}$ (the total number of
calls $\mathcal{A}$ will make to $\OSch$).  We construct a PPT
adversary $\mathcal{A}$ playing the sEUF-CMA game with Schnorr signatures and show that
$\Advantage^{\mathsf{sEUF\text{-}CMA}}_{\mathsf{Schnorr}}(\mathcal{A}) \;\ge\; \varepsilon(\kappa) - \negl$.

\smallskip
\noindent
\textbf{Construction of $\mathcal{A}$.}
$\mathcal{A}$ receives a Schnorr public key $\pk = g^x$ together with oracle access to the Schnorr signing oracle $\OSch$ and to the oracle $\hsig$. 
It maintains a table $\Tctx$ implementing $\hctx$ and a record $\mathcal{R}$ of every triple $(m,r,s)$ for $(r,s)$ obtained as the answer to $\OSch(m)$.
It initially sets $Q := \varnothing$, $\Qctx := \varnothing$, and $\mathcal{R} = \varnothing$.
\begin{enumerate}
    \item \textbf{Setup.} $\mathcal{A}$ samples $\pct_1, \ldots, \pct_N \sample \{0,1\}^\kappa$ \\ and returns $(\pk, \{\pct_i\}_{i=1}^N)$ to $\mathcal{B}$.
    \item \textbf{$\hctx$ oracle.} Upon receiving some query $r \in \G$, if $\Tctx[r]$ is defined, return it; otherwise sample $u \sample \{0,1\}^\kappa$, set $\Tctx[r] \deq u$, and return $u$.
    \item \textbf{$\hsig$ oracle.} Forward each query verbatim to the $\hsig$ oracle and return the answer.
    \item \textbf{$\OEOTSctx(m, \pct)$.} If $\pct \notin \{\pct_1, \ldots, \pct_N\}$ or $\pct \in \Qctx$, return $\bot$.  Otherwise:
    \begin{enumerate}[leftmargin=*]
        \item Query $(r, s) \leftarrow \OSch(m)$ and append $(m, r, s)$ to $\mathcal{R}$.
        \item If $\Tctx[r]$ is already defined and $\Tctx[r] \neq \pct$, abort.
        \item Set $\Tctx[r] \deq \pct$, append $(m, \pct, (r,s))$ to $Q$, append $\pct$ to $\Qctx$, and return $(\pct, (r,s))$.
    \end{enumerate}
    \item \textbf{$\OEOTS(m)$.} Query $(r, s) \leftarrow \OSch(m)$ and append $(m, r, s)$ to $\mathcal{R}$.  If $\Tctx[r]$ is undefined, sample $u \sample \{0,1\}^\kappa$ and set $\Tctx[r] \deq u$. Set $\pct \deq \Tctx[r]$, append $(m, \pct, (r,s))$ to $Q$, and return $(\pct, (r,s))$.
    \item \textbf{Output.} When $\mathcal{B}$ returns $(m^*, \pct^*, \sigma^*)$ with $\sigma^* = (r^*, s^*)$, $\mathcal{A}$ outputs $(m^*, \sigma^*)$ as its Schnorr forgery candidate.
\end{enumerate}

\smallskip
\noindent
\textbf{Bad Events.} Let $\Bad := \mathsf{B1} \cup \mathsf{B2}$, where:
    \begin{description}[leftmargin=*]
        \item[$\mathsf{B1}$:] In the real-world experiment, two of the setup-time secret nonces $k_1, \ldots, k_N$ collide (so the corresponding $g^{k_i}$ are not pairwise distinct).
        \item[$\mathsf{B2}$:]  At some $\OEOTSctx(m, \pct)$ call that passes the check, the nonce $r$ returned by $\OSch(m)$ satisfies $\Tctx[r] = u$ for some $u \neq \pct$ defined prior to the call (Step~4b above aborts).
\end{description}

The Schnorr signing oracle samples its nonce $k$ uniformly at random from $\Z_q^\times$, so $r = g^k$ is uniform on $\G \setminus \{1\}$ and independent of $\mathcal{B}$'s prior view.
Then, since $q$ is exponential in $\kappa$, 
$$\Pr[\mathsf{B1}] = {N \choose 2} \frac{1}{q-1} = \negl$$
Moreover, at any $\OEOTSctx$ call, $\Tctx$ is defined on at most $q_H + Q_S - 1$ elements of $\G$, which implies $\Pr[\mathsf{B2}] \;\le\; q_{\mathsf{ctx}} \cdot (q_H + Q_S - 1) \,
(q-1)^{-1} = \negl$.
Therefore, $\Pr[\Bad] = \negl$.

\smallskip
\noindent
\textbf{View Equivalence.} We claim that, conditioned on $\neg \Bad$, $\mathcal{B}$'s view in the simulation is identically distributed to its view in the real sEUF-CMADC game.

\textit{Public values at setup.}
In the real game, the challenger samples $k_1, \ldots, k_N \sample \Z_q^\times$ independently and publishes $\pct_i := \hctx(g^{k_i})$ for each $i$.
Conditioned on $\neg \mathsf{B1}$, the values $g^{k_1}, \ldots, g^{k_N}$ are pairwise distinct, so $\hctx(g^{k_1}), \ldots, \hctx(g^{k_N})$ are $N$ i.i.d uniform $\kappa$-bit strings.
In the simulation, the $\pct_i$ are sampled directly as $N$ i.i.d.\ uniform $\kappa$-bit strings.
Hence the marginal distribution of the public setup is identical in both worlds.

\textit{Random oracle queries.}
$\hsig$ is forwarded directly, so its distribution is unchanged.
For $\hctx$, sampling realizes a uniformly random function on points where no programming has occurred.
Programming events occur in two places: (i) Step~4c of $\OEOTSctx$, where $\Tctx[r] \deq \pct$, and (ii) $\OEOTS$, where $\Tctx[r] \deq u$ for a uniform $u$.
Each programmed value is therefore a uniform $\kappa$-bit string, independent of all other defined entries, exactly as in a true random oracle, implying that the output of $\hctx$ queries is identically distributed in both worlds.

\textit{Signing oracle queries.}
Conditioned on $\neg \mathsf{B2}$, every programming step in $\OEOTSctx$ targets a previously undefined target entry; so $\OEOTSctx(m, \pct)$ queries do not abort.
Consider any such query that does not return $\bot$ and does not abort.  
In the real game, the answer is $(\pct, (r_i, s_i))$ where $r_i = g^{k_i}$ and $\pct_i = \hctx(r_i)$, and $s_i = k_i + x \cdot \hsig(r_i \| m) \bmod q$.
In the simulation, the answer is $(\pct, (r, s))$ where $(r, s) = \OSch(m)$ for some $r = g^k$, $k \sample \Z_q^\times$ and $s = k + x \cdot \hsig(r \| m) \bmod q$.

Note that in both worlds, the public nonce $r$ is uniform in $\G \setminus \{1\}$, and the $s$-component is deterministic once the $r$ component is given.
The constraint $\hctx(r) = \pct$ is enforced in both cases, either by the setup-time sampling in the real game, or by setting $\Tctx[r] \deq \pct$ in the simulation.
Therefore, the output of $\OEOTSctx(m, \pct)$ (and by a similar reasoning, that of $\OEOTS(m)$) queries is identically distributed in both worlds.

\smallskip
\noindent
\textbf{Freshness of the Forgery.}
Suppose $\mathcal{B}$ wins the simulated game and $\neg \Bad$ holds.
Write $\sigma^* = (r^*, s^*)$.  By the EOTS verification rule,
\begin{equation*}
\hctx(r^*) = \pct^* \quad\text{and}\quad \mathsf{Schnorr\text{-}Verify}\bigl(\pk, m^*, (r^*, s^*)\bigr) = 1.
\end{equation*}
We show that $(m^*, \sigma^*)$ is a fresh Schnorr forgery, \ie, the exact pair $(m^*, r^*, s^*)$ is not in $\mathcal{R}$.

Suppose, for contradiction, that $(m^*, r^*, s^*) \in \mathcal{R}$.
Then, $\mathcal{A}$ obtained $(r^*, s^*)$ from $\OSch(m^*)$ during the simulation, while answering one of $\mathcal{B}$'s signing queries on the message $m^*$.

\emph{Case~1: $\OEOTSctx(m^*, \pct')$ query.}
In this case, $\mathcal{A}$ programmed $\Tctx[r^*] := \pct'$ in Step~4c and stored $(m^*, \pct', (r^*, s^*))$ in $Q$.

\emph{Case~2: $\OEOTS(m^*)$ query.}
Then, $\mathcal{A}$ either reused a previously defined value $\Tctx[r^*] = \pct'$, or sampled a fresh $\pct' \sample \{0,1\}^\kappa$ and set $\Tctx[r^*] := \pct'$.
In either case, $(m^*, \pct', (r^*, s^*)) \in Q$.
After the call, $\Tctx[r^*] = \pct'$ in both cases.
This entry is never overwritten afterwards: given a $\hctx$ query $r^*$, the simulator returns the same $\pct'$.
Therefore, $\pct^* = \pct'$, implying that $(m^*, \pct', (r^*, s^*)) \in Q$, which is a contradiction.
Hence, it must be that $(m^*, r^*, s^*) \notin \mathcal{R}$ and $(m^*, \sigma^*)$ is a fresh Schnorr forgery.

Given the freshness of the forgery, we conclude that $\Adv$ wins the Schnorr \textsf{sEUF-CMA} game whenever $\mathcal{B}$ wins the EOTS \textsf{sEUF-CMADC} game and $\neg\mathsf{Bad}$ holds.
Therefore,
\[
\Advantage^{\mathsf{sEUF\text{-}CMADC}}_{\mathsf{EOTS}}(\mathcal{B})
\;\le\;
\Advantage^{\mathsf{sEUF\text{-}CMA}}_{\mathsf{Schnorr}}(\mathcal{A})
\;+\; \Pr[\Bad],
\]
where $\Pr[\Bad] = \negl$.
\end{proof}

\paragraph{\textsf{EXT-SCMA} Security}

\begin{algorithm}
\caption{Game for Extractability under Single Chosen Message
Attacks (\textsf{EXT-SCMA})}
\label{daps:game:ext-scma}
\begin{algorithmic}[1]\small
\State $(\pk, \pct, m_1, \sigma_1, m_2, \sigma_2) \gets \mathcal{A}$
\State \Return 
\begin{align*}
& m_1 \neq m_2\ \wedge \\ 
& \mathsf{EOTS\textnormal{-}Verify}(\pk, m_1, \pct, \sigma_1)\ \wedge \\
& \mathsf{EOTS\textnormal{-}Verify}(\pk, m_2, \pct, \sigma_2)\ \wedge \\
& \mathsf{EOTS\textnormal{-}PK}(\mathsf{EOTS\textnormal{-}Extract}(\pk, \pct, m_1, \sigma_1, m_2, \sigma_2))
\neq \pk
\end{align*}
\end{algorithmic}
\end{algorithm}

\begin{definition}[\textsf{EXT-SCMA} security]
\label{EOTS:sec:ext-scma} The \textsf{EXT-SCMA} game formalizes the extractability guarantee for the EOTS scheme.
\begin{equation*}
\forall \mathcal{A} \in \mathsf{PPT},
\Advantage^{\mathsf{EXT\text{-}SCMA}}_{\mathsf{EOTS}}(\Adv) \deq \Pr[\mathsf{EXT\textnormal{-}SCMA}^{\mathcal{A}}(1^{\kappa}) = 1]
< \negl
\end{equation*}
\end{definition}
Lastly, \textsf{EXT-SCMA} security guarantees that two valid signatures on distinct messages with the same public key and public context component can be used to extract the expected secret key, except with negligible probability.

\begin{theorem}[\textsf{EXT-SCMA}]
\label{thm:cma}
The EOTS scheme of Section~\ref{sec:daps-constr} is \textsf{EXT-SCMA} secure whenever $\hctx$ and $\hsig$ are modelled as random oracles.
\end{theorem}

\begin{proof}

Consider the output $(y, \hashctx_r, m_1, \sigma_1, m_2, \sigma_2)$ of a \textsf{PPT} adversary $\mathcal{A}$.
We analyze the probability that $\mathcal{A}$ wins the \textsf{EXT-SCMA} game.
Suppose $m_1 \neq m_2$, and both $\mathsf{EOTS\textnormal{-}Verify}(y, m_1, \hashctx_r, \sigma_1) = 1$ and $\mathsf{EOTS\textnormal{-}Verify}(y, m_2, \hashctx_r, \sigma_2) = 1$.
Next, given for $\sigma_1 = (r_1, s_1)$ and $\sigma_2 = (r_2, s_2)$, we consider the following two disjoint events:

\textbf{Event 1:} $r \deq r_1 = r_2$ and $\hsig(m_1 || r) \neq \hsig(m_2 || r)$:
Since $g$ is a generator of group $G$ with size $q$, for each $b
\in G$ there is a unique $a \in \mathbb{Z}^{\times}_q$ such that
$g^a = b$.
Let $x, k$ be the unique elements of
$\mathbb{Z}^{\times}_q$ such that $g^x = y$ and $g^k = r$.
Then,
\begin{gather*}
\mathsf{EOTS\textnormal{-}Verify}(y, m_1, r, s_1) = 1 \Rightarrow
r = g^{s_1}(y^{\hsig(r || m_1)})^{-1} \Rightarrow \\
g^k = g^{s_1}(g^{x\hsig(r ||
m_1)})^{-1} \Rightarrow \hsig(r || m_1) = \frac{s_1 - k}{x} \enspace.
\end{gather*}
Likewise
\begin{equation*}
\mathsf{EOTS\textnormal{-}Verify}(y, m_2, \hashctx_r, \sigma_2) = 1 \Rightarrow \hsig(r ||
m_2) = \frac{s_2 - k}{x} \enspace.
\end{equation*}
Then, since $\hsig(m_1 || r) \neq \hsig(m_2 || r)$,
\begin{align*}
&\ \mathsf{EOTS\textnormal{-}PK}(\mathsf{EOTS\textnormal{-}Extract}(y, r, m_1, s_1, m_2, s_2)) \\
=&\ \mathsf{EOTS\textnormal{-}PK}\left(\frac{s_1-s_2}{\hsig(r||m_1)-\hsig(r||m_2)}\right) \\
=&\ \mathsf{EOTS\textnormal{-}PK}\left(\frac{s_1-s_2}{\frac{s_1-k}{x} - \frac{s_2-k}{x}}\right) \\
=&\ \mathsf{EOTS\textnormal{-}PK}(x) = g^x = y \enspace,
\end{align*}
which implies that $\mathcal{A}$ loses the \textsf{EXT-SCMA} game.

\textbf{Event 2:} $r_1 \neq r_2$, or $r_1 = r_2 = r$ and $\hsig(m_1 || r) = \hsig(m_2 || r)$:
In this case, $\hctx(r_1) = \hctx(r_2)$ for $r_1 \neq r_2$, and $\hsig(m_1 || r) = \hsig(m_2 || r)$ for $(m_1 || r) \neq (m_2 || r)$.
However, since $\hctx$ and $\hsig$ are random oracles, the probability that the adversary finds such $m_1, m_2, r_1, r_2$ is $\negl$, concluding the proof.
\end{proof}
In fact, for proving $\textsf{EXT-SCMA}$-security, it is sufficient for $\hctx$ and $\hsig$ to be collision-resistant. 

\section{Bond Contract Implementation using Covenants}
\label{sec:appendix-bond-contract-covenant}

\begin{algorithm}
\caption{A simple bond contract implemented in Bitcoin Script using OP\_CHECKTEMPLATEVERIFY. In a year, the validator can take its deposit back. Until then, if its key is leaked, anyone can execute the slashing transaction.}
\label{alg.bond.contract}
\begin{verbatim}
OP_IF
    <1 year>
    OP_CHECKLOCKTIMEVERIFY OP_DROP
OP_ELSE
    <hash_of_slashing_transaction>
    OP_CHECKTEMPLATEVERIFY
OP_ENDIF

<validator_pubkey>
OP_CHECKSIG 
\end{verbatim}

\end{algorithm}

\section{Proof of~\Cref{thm:tendermint-accountability-extraction}}
\label{sec:appendix-proof-thm-daps-safety}

\begin{proof}[Proof of~\Cref{thm:tendermint-accountability-extraction}]
Suppose the clients $\client_1$ and $\client_2$ finalized two conflicting chains.
Then, there must be an earliest height $H$, at which two conflicting blocks, $B_1$ and $B_2$ are finalized in their views.
Hence, $\client_1$ and $\client_2$ must have respectively observed two quorums of $2f+1$ height $H$ \EOTS signatures $\langle \mathsf{Final}, H, \mathsf{hash}(B_1) \rangle$ and $\langle \mathsf{Final}, H, \mathsf{hash}(B_2) \rangle$ for $B_1$ and $B_2$.
Upon obtaining the two quorums from the clients $\client_1$ and $\client_2$, an efficient extraction algorithm identifies the $f+1$ validators at the intersection of the two quorums as protocol violators, since they have satisfied the condition in~\Cref{alg:slashing-condition}.
By \textsf{EXT-SCMA} security of EOTS (\Cref{EOTS:sec:ext-scma}), the algorithm can then extract their secret signing keys (w.o.p.).
Moreover, since honest validators create at most one \EOTS signature per height, 
by \textsf{sEUF-CMADC} security of EOTS (\Cref{eots:sec:seuf-cmadc}), no PPT adversary can forge signatures on behalf of an honest validator.
Thus, Tendermint with the finality gadget satisfies $(f+1)$-EOTS safety.
\end{proof}

\section{Proof of~\Cref{thm:tendermint-liveness}}
\label{sec:appendix-proof-thm-liveness}
\begin{proposition}
\label{prop:context}
Suppose the duration of a Tendermint round is lower bounded by some $\epsilon > 0$.
Then, there exists a fixed $w_0 \in \mathbb{Z}^+$ such that if the public context components $\pct_{H,\validator}$ for Tendermint heights $H \in \{w+1, \ldots, w+w_0\}$ is input to Bitcoin before the Tendermint chain reaches height $w-w_0+1$, then these components appear in the confirmed Bitcoin chain of all validators and clients by the time the Tendermint chain reaches height $w$.
\end{proposition}
\Cref{prop:context} follows directly from~\Cref{prop:bitcoin-secure}: $w_0$ is chosen so that the wall-clock duration of $w_0$ Tendermint rounds (at least $w_0 \cdot \epsilon$) exceeds the maximum Bitcoin tx-inclusion delay $k_f$ in wall-clock time, guaranteeing that any $\pct$ broadcast at Tendermint height $w-w_0+1$ has been confirmed on Bitcoin by Tendermint height $w$.

\begin{proof}[Proof of~\Cref{thm:tendermint-liveness}]
Let $\chain^\validator_{t}$ denote the sequence of Tendermint blocks confirmed (decided) by an honest validator $\validator$ following the original confirmation rule of Tendermint~\cite[Algorithm 1, line 49]{2018tendermint}.
Note that \Cref{alg:validator-finality} only invokes EOTS signing after a Tendermint decision is reached---the gadget adds an output hook without modifying the prevote/precommit logic---so the Tendermint protocol code executed by honest validators is unaffected by the finality gadget.
Thus, when the number of adversarial validators is less than or equal to $f$, Tendermint satisfies agreement, validity, and after GST, termination by \cite[Lemmas 3, 4, 7]{2018tendermint}.
Then, for any honest validators $\validator$ and $\validator'$ and times $t$ and $t'$, (i) $\chain^\validator_{t} \preceq \chain^{\validator'}_{t'}$ or vice versa, (ii) if a block $B$ appears in $\chain^\validator_{t}$ at height $H$ at some time $t$, then $B$ appears within $\chain^{\validator'}_{t'}$ at the same height by time $t' = \max(t,\text{GST})+\Delta$, (iii) these chains satisfy liveness per Definition~\ref{def:security}.
By property (i), for all times $t$ and honest validators $\validator$, $\chain^\validator_t \preceq \chain_t := \cup_{\text{honest } \validator'} \chain^{\validator'}_t$, and for all times $t$ and $t' > t$, $\chain_t \preceq \chain_{t'}$.
Thus, if a block at height $H$ conflicts with $\chain_t$ at some time $t$, it eventually conflicts with the height $H$ blocks within any honest validator $\validator$'s chain $\chain^{\validator}$, and vice versa.
Therefore, by Alg.~\ref{alg:validator-finality} and (ii), an honest validator sends an \EOTS signature for each block in $\chain_{t}$ by time $t' = \max(t,\text{GST})+\Delta$, and only for the blocks within $\chain^\validator_{t}$ by time $t'$ (By~\Cref{prop:context}, for any honest validator and Tendermint height at which a block is confirmed, there is a public context component on Bitcoin).
Then, by the bound on the number of adversarial validators, each block $B \in \chain_t$ receives $2f+1$ \EOTS signatures
by round $\max(t,\text{GST})+\Delta$, which are observed by all clients at all times $t' \geq \max(t,\text{GST})+2\Delta$, \ie, all clients 
finalize the blocks in $\chain_t$ by $\max(t,\text{GST})+2\Delta$.

Finally, by (ii) and (iii), any transaction $\tx$ input to an honest validator at some time $t$ appears in $\chain_{t'}$ for all rounds $t' \geq \max(t,\text{GST}) + \Tconfirm + \Delta$.
Hence, $\tx$ appears in all finalized chains $\chain^{\client}_{t'}$ for all clients $\client$ and times $t' \geq \max(t,\text{GST}) + \Tconfirm + 2\Delta$, concluding liveness.
\end{proof}

\section{Proof of~\Cref{thm:economic-safety}}
\label{sec:appendix-proof-thm-economic-safety}

\paragraph{Watchtower assumption.}
The proof below assumes the existence of an honest \emph{forensic watchtower}: an online client that observes the EOTS-Final signatures broadcast by the validators and invokes the forensic protocol of \Cref{alg:slashing-condition} once two conflicting quorums are received.
Under partial synchrony, this assumption is discharged (post-GST) by any honest client connected to the validators' gossip channel.
The dynamic-stake protocol of \Cref{sec:timestamping-protocol} discharges it unconditionally by routing EOTS-Final signatures through Bitcoin's synchronous broadcast (see~\Cref{prop:bitcoin-secure}).

\begin{proof}[Proof of~\Cref{thm:economic-safety}]
By Theorem~\ref{thm:tendermint-accountability-extraction}, Tendermint with the finality gadget satisfies $(f+1)$-EOTS safety: when there is a safety violation, two conflicting EOTS-Final quorums exist at some height.
The watchtower observes both quorums and invokes the forensic protocol of~\Cref{alg:slashing-condition}, which (by \textsf{EXT-SCMA} security of EOTS, \Cref{EOTS:sec:ext-scma}) extracts the secret signing keys of at least $f+1$ adversarial validators (w.o.p.).
Using these extracted keys, the watchtower constructs and broadcasts slashing transactions to the bond contracts of the identified $f+1$ validators.

We claim that no race-spend by the adversary can divert the stake away from the burn address.
Before the timelock expires, every valid spend of the bond contract is constrained to send the deposit to an unspendable address: in the CTV variant of \Cref{alg.bond.contract}, this is enforced on-chain by \texttt{OP\_CHECKTEMPLATEVERIFY} against the pre-committed slashing template (whose output is $\mathsf{OP\_RETURN}$); in the deleted-key covenant-committee variant of \Cref{alg.bond.contract.musig}, the spend additionally requires a committee co-signature, and under the existential-honesty assumption (at least one committee member keeps its signing key private and refuses to co-sign anything other than the slashing transaction) no committee co-signature on a non-burn spend exists.
Hence, even if the adversary obtains the leaked secret key first and races to broadcast its own spend, that spend---if valid---must burn the stake.
Therefore, in the event of a safety violation, the bitcoin of $f+1$ adversarial validators is sent to the burn address (either by the honest client's slashing transaction or by an adversarial race-spend that nonetheless lands on the burn address), and is thus slashed.

On the other hand, by the existential unforgeability of the signatures of the honest validators (Def.~\ref{eots:sec:seuf-cmadc}), no slashing transaction can be created on behalf of an honest validator, and its stake cannot be slashed (w.o.p.).
Therefore, the protocol satisfies $(f+1)$-economic safety.
\end{proof}

\section{The Timestamping Protocol}
\label{sec:appendix-timestamping-protocol}

Each client and validator downloads the consumer blocks and track the timestamps and the bond contract on the provider chain.
In the rest of this and the following sections, for generality, we use the term provider chain for Bitcoin and assume that the consumer chain is a consensus protocol with the finality gadget.
Before describing the details of the timestamping protocol, we recall the parameters $k_c(\Delta t)$ and $k_f$ defined by Proposition~\ref{prop:bitcoin-secure}.
Intuitively, $k_c(\Delta t)$ denotes the number of provider blocks added to the provider chain during a time window of length $\Delta t$. Throughout the analysis, the relevant window is one consumer-chain period of $m$ heights, with duration $m \cdot \Tconfirm$, where $\Tconfirm$ is the consumer chain's confirmation latency (\Cref{def:security}); thus we use $k_c(m \cdot \Tconfirm)$ for the per-period bound.
Similarly, $k_f$ denotes the time, measured in the number of provider blocks, it takes for a message posted to the provider chain to appear in a confirmed provider block.

\setlength{\textfloatsep}{0.3cm}
\setlength{\floatsep}{0.3cm}
\begin{algorithm}
    \caption{The algorithm used by a client $\client$ to determine if a consumer block $B$ is available and \valid.
    It takes the consumer chain $\chain$ ending at $B$ and the confirmed provider chain $\powchain$ in $\client$'s view as input and outputs true if $B$ is available and \valid.
    The function $\textsc{GetValidators}$ %
    outputs the validator set determined for the next consumer chain height by an available and \valid consumer chain $\chain'$ and a confirmed provider chain $\powchain$ taken as input.
    It outputs $\bot$ if any provider block among those referred by the consumer blocks within $\chain'$ is not in $\powchain$.
    The function $\textsc{Signed}$ checks if there are $2f+1$ EOTS signatures on a given consumer block by the specified validator set.
    }
    \label{alg.bitcoin.valid}
    \begin{algorithmic}[1]\small
    \Function{\sc IsValid}{$\chain, \powchain$}
        \If{$\chain = B_0$}
            \Comment{If $C$ is the genesis consumer block}
            \State\Return $\mathrm{True}$
        \ElsIf{$\chain[0] \neq B_0$}
            \State\Return $\mathrm{False}$
        \EndIf
        \Let{B_0,\ldots,B_r}{\chain}
        \For{$i = 1$ to $r$}
            \Let{\rm vals}{\textsc{GetValidators}(\chain[:i-1], \powchain)}
            \If{$\neg\textsc{Signed}(\chain[i], \mathrm{vals}) \lor \mathrm{vals} = \bot$}
                \State\Return $\mathrm{False}$
            \EndIf
        \EndFor
        \State\Return $\mathrm{True}$
    \EndFunction
    \end{algorithmic}
\end{algorithm}

\setlength{\textfloatsep}{0.3cm}
\setlength{\floatsep}{0.3cm}
\begin{algorithm}
    \caption{The algorithm used by a client $\client$ to find the canonical consumer chain $\chain^\client_t$ at time $t$. 
    It takes as input a tree $\TT$ of available and \valid consumer blocks and the confirmed provider chain $\powchain$ in $\client$'s view at time $t$. 
    The function $\textsc{GetCkpts}$ outputs the ordered sequence of timestamps on the given provider chain $\powchain$. 
    The function $\textsc{IsCor}$ checks if a given timestamp is correct based on the height $H$ of the canonical consumer chain output so far, the tracked current period $\mathrm{per}$ and the validator set identified by this canonical consumer chain.
    The function $\textsc{ProH}$ returns the height of the provider block containing a given timestamp or referred by a given consumer block depending on its input.
    The function $\textsc{ConH}$ returns the height of the specified consumer chain or the height included within the specified timestamp.
    The function $\textsc{GetCh}$ returns the chain of available and \valid consumer blocks behind a given timestamp using $\TT$.
    It returns $\bot$ if there is an unavailable or \invalid block in the consumer chain defined by the block at the preimage of the given timestamp.
    The function $\textsc{GetProH}$ takes a consumer chain $\chain$ and a provider chain $\powchain$ as input and returns the height of the highest provider block in $\powchain$ among those referred by the consumer blocks within $\chain$ (if this highest provider block is not in $\powchain$, it returns $\bot$).
    The function $\textsc{IsLast}$ returns true if a given consumer chain ends at a block that is the last block of its period.
    The function $\textsc{GetChildren}$ returns the children of a given block. 
    }
    \label{alg.bitcoin.forkchoice}
    \begin{algorithmic}[1]\small
    \Function{\sc OutputConsumerCh}{$\TT, \powchain$}
        \Let{\mathrm{per}, \chain, H}{1, B_0, 0} \Comment{$H$: timestamped consumer height}
        \Let{\ckpt_1,\ldots, \ckpt_r}{\textsc{GetCkpts}(\powchain)}
        \Let{\mathrm{vals}}{\textsc{GetValidators}(B_0, \powchain)}
        \Let{h}{\textsc{ProH}(B_0)}
        \For{$i=1$ to $r$} \Comment{Get timestamped consumer chain}
            \If{$\textsc{IsCor}(\ckpt_i, \mathrm{vals}, H, \mathrm{per}) \land \textsc{ProH}(\ckpt_i) < h+k_d$}\label{line:isvalid}
                \Let{\chain_i}{\textsc{GetCh}(\TT, \ckpt_i)}
                \If{$\chain_i = \bot$}
                    \label{line:btc2}
                    \State\Return $\chain$~\Comment{Safe-stop rule 1}
                \ElsIf{$\chain \preceq \chain_i \land \textsc{ConH}(\chain_i)=\textsc{ConH}(\ckpt_i)$}
                    \label{line:btc1}
                    \If{$|\powchain| \geq h+k_d+k_f \land \exists \ckpt \colon (\textsc{ProH}(\ckpt) < h+k_d+k_f \land \ckpt \text{ conflicts with } \chain_i)$}
                        \label{line:btc4}
                        \State\Return $\chain$~\Comment{Safe-stop rule 2}
                    \Else
                        \Let{\chain}{\chain_i}
                        \label{line:update}
                        \Comment{Update $\chain$.}
                        \Let{H}{|\chain_i|}
                        \If{$\textsc{IsLast}(\chain_i)$}
                            \Let{h}{\textsc{GetProH}(\chain_i, \powchain)}
                            \Let{\rm per}{per+1}
                            \Let{\rm vals}{\textsc{GetValidators}(\chain_i, \powchain)}
                        \EndIf
                    \EndIf
                \EndIf
            \EndIf
        \EndFor
        \Let{\mathrm{chs}}{\textsc{GetChildren}(\TT,\chain[-1])}
        \Let{\mathrm{chs}}{\{B \colon B \in \mathrm{chs} \land  |\powchain| < \textsc{GetProH}(B.\chain, \powchain) + k_d\}}
        \label{line:btc3}
        \While{$|\mathrm{chs}| = 1$}
            \Let{\chain}{\chain \concat \mathrm{chs}}~\Comment{Add the new child to $\chain$}
            \Let{\rm chs}{\textsc{GetChildren}(\T,\mathrm{chs})}
            \Let{\mathrm{chs}}{\{B \in \mathrm{chs} \colon |\powchain| < \textsc{GetProH}(B.\chain, \powchain) + k_d\}}~\label{line:heightcheck}
        \EndWhile
        \State\Return $\chain$ 
    \EndFunction
    \end{algorithmic}
\end{algorithm}

\subsection{Determining the Validator Set}
\label{sec:pos-protocol-valset}

An honest validator includes the hash of the highest confirmed provider block in its view within the proposed consumer block.
When an honest validator first receives a consumer block $B$ proposed at a certain height, it checks if the provider block $b$ referred by the hash in $B$'s parent is confirmed in its view.
If $b$ becomes confirmed before the validator moves to the voting step of the protocol ($\mathsf{Prevote}$ in Tendermint), it continues to execute the consumer chain protocol as specified.
Otherwise, it ignores block $B$.

When a client $\client$ first downloads a consumer block $B$, it checks if the block is \emph{\valid} in its view.
The genesis consumer block, $B_0$, is assumed to be \valid and specifies the initial validator set by referring to a provider block containing the bonding transactions for this initial set.
Validity of any other consumer block $B$ at height $|B|$ is determined by $\client$ according to the following rules: (i) 
there are $2f+1$ \EOTS signatures for $B$ with context $|B|$ from the correct validator set for height $|B|$, (ii) the provider block $b$ referred by $B$ is confirmed in $\client$'s view, and (iii) $B$'s parent is \valid (Alg.~\ref{alg.bitcoin.valid}).
If so, $\client$ accepts $B$ as a \emph{\valid} consumer block.

The validator set stays fixed during periods of $m$ consecutive consumer chain heights, where $m$ can be as little as $1$.
Suppose a client $\client$ wants to determine the validator set for period $e$ after observing \valid blocks for the periods $1, \ldots, e-1$.
Let $b$ be the highest confirmed provider block referred by the consumer blocks from periods $1, \ldots, e-1$ in $\client$'s view at some time $t$.
Then, $\client$ determines the validator set for period $e$ as those who have bonded their stake (on the provider chain) by the provider block $b$, and whose validator duties have not ended by $b$. %

\begin{lemma}[Validator-set agreement]
\label{lem:valset-agreement}
Suppose Bitcoin is secure (\Cref{prop:bitcoin-secure}).
For any consumer-chain period $e$ and any two honest clients $\client_1, \client_2$ that have observed \valid consumer blocks for periods $1, \ldots, e-1$, there exists a finite time at which $\client_1$ and $\client_2$ agree on the validator set assigned to period $e$.
\end{lemma}
\begin{proof}
By~\Cref{alg.bitcoin.valid}, validity of a consumer block requires the referenced provider block to be confirmed in the client's view.
Since both $\client_1$ and $\client_2$ have observed \valid blocks for periods $1, \ldots, e-1$, both eventually observe the same confirmed provider chain prefix containing those references (by Bitcoin safety, \Cref{prop:bitcoin-secure}).
The validator set for period $e$ is a deterministic function of the bonding transactions referenced in this prefix, computed identically by all honest clients; hence $\client_1$ and $\client_2$ agree on the period-$e$ validator set as soon as both have processed the provider blocks referenced by the period $e-1$ consumer blocks.
\end{proof}

For ease of description, in the rest of this section, we assume that there are $n = 3f+1$ bonded validators, each with equal stake, at each period $e \in \mathbb{Z}^+$.
We note that our protocol accommodates different numbers of validators at different periods with inhomogeneous stake amounts.
In the latter case, the voting power of the validators must be scaled in proportion to the fraction of their stake within their period's total stake.
Then, we can guarantee that if a safety violation is committed across blocks within a period $e$ with a total stake of $p$, at least $p/3$ tokens belonging to the adversarial validators can be slashed.

\subsection{Bonding and Unbonding}
To join the validator set, a validator locks its stake in the provider chain's bond contract via a bonding transaction.
Upon bonding its stake at some provider block $b$, it can act as a validator for the consumer chain heights described above while it continues its validator duties.
The validator duties end at the provider block $b'$ that extends $b$ by some fixed amount $K_a$ determined by the protocol.
Afterwards, the validator can retrieve its stake at or after the provider block extending $b'$ by $k_u$ blocks (\ie, extending $b$ by $K_a+k_u$ blocks on Bitcoin), where $k_u = 2 k_c(m \cdot \Tconfirm) + 4 k_f$ and it is called the \emph{unbonding delay}.
To prevent early unbonding, the bond contract enforces a timelock on the bonded stake until the $(K_a+k_u)$-th block extending $b$ (\cf Algs.~\ref{alg.bond.contract} and~\ref{alg.bond.contract.musig}).

\subsection{Generating the Context Tuples} 
Each validator $\validator$ periodically generates context tuples 
$$(\sct_{H,\validator}, \pct_{H,\validator}) \gets \mathsf{EOTS\textnormal{-}ContextGen}(1^{\kappa})$$
for a window of $w_0$ consumer chain heights $H \in \{w+1, \ldots, w+w_0\}$ and publishes the public components $\pct_{H,\validator}$ on Bitcoin.
Prior to bonding its stake, $\validator$ posts $\pct_{H,\validator}$ for an estimated initial window of consumer chain heights for which it might be determined as a validator per the rule above.
It periodically posts new $\pct_{H,\validator}$ during the $K_a$ Bitcoin blocks of its time as a validator.

More specifically, let $w$ denote the height of the consumer chain right before $\validator$ inputs its bonding transaction to Bitcoin.
Then, $\validator$ initially generates $\pct_{H,\validator}$ for a window of heights $H \in \{w+1, \ldots, w+w_0\}$.
In the worst case, $\validator$ might be selected as a validator for a consumer chain height greater than $w$, and some context components remain unused.
However, in any case, $\validator$ will be able to generate EOTS signatures that verify with respect to the $\pct_{H,\validator}$ on Bitcoin for all the consumer chain heights at which it is selected as a validator.

As before, the validators can minimize the Bitcoin footprint by posting a Merkle root; but for simplicity, we assume that the components $\pct_{H,\validator}$ are published in the clear.

\subsection{Timestamping on the Provider Chain}
Each honest validator periodically sends the \emph{timestamp} of the last block of each period to the provider chain.
To avoid duplicate timestamps, a single client or validator called the \emph{watchtower} can be tasked with timestamping new blocks.
The timestamp of a consumer block consists of the hash of the block, its height and the quorum of $2f+1$ \EOTS signatures on the block (with context equal to its height) by the corresponding validator set.
Note that the period $e$ of a consumer block can be found by dividing its height $H$ with $m$ (\ie, $e = \lfloor H/m \rfloor +1$).

Two timestamps are said to conflict if they both include (i) the same consumer block height $H$, (ii) different consumer block hashes, and (iii) $2f+1$ \EOTS signatures with height $H$ as context on their respective consumer block hashes.

\subsection{Block Output Rules (Alg.~\ref{alg.bitcoin.forkchoice})}
When there is a posterior corruption attack, a client $\client$ might observe conflicting \valid consumer blocks finalized by the same validator set.
In this case, $\client$ wants to identify and output only the \emph{canonical} consumer chain consisting of blocks signed earlier in time.
Towards this goal, it first downloads the blocktree of all \valid consumer blocks.
Let $\ckpt_i$, $i \in [r]$, denote the sequence of timestamps on the provider chain, listed from the genesis to the tip of the chain (denoted by $\powchain^{\client}_t$) in $\client$'s view at time $t$.
Starting at the genesis consumer block, $\client$ constructs the canonical consumer chain (denoted by $\chain^{\client}_t$) one block at a time, by sequentially processing these timestamps.
For $i=1,\ldots, r$, let $\chain_i$ denote the chain ending at the consumer block (denoted by $B_i$), which is the preimage of the hash within $\ckpt_i$, if $B_i$ and its prefix are available and \valid in $\client$'s view at time $t$.
Suppose $\client$ has processed the timestamp sequence until some timestamp $\ckpt_j$ and constructed so far as its canonical consumer chain, the chain $\chain$ of available and \valid consumer blocks ending at some block $B$ with consumer chain height $H$ and period $e$.
Define $\tilde{e}=e+1$ if $B$ is the last block of its period; and $\tilde{e}=e$ otherwise.
Let $h_{\tilde{e}-1}$ denote the height of the highest confirmed provider block referred by the consumer blocks within the periods $1, \ldots, \tilde{e}-1$.
We call the next timestamp $\ckpt_{j+1}$ \emph{correct}, if (i) the height $H_{j+1}$ included in $\ckpt_{j+1}$ is larger than $H$ and $\lfloor H_{j+1}/m \rfloor+1 = \tilde{e}$, (ii) $\ckpt_{j+1}$ includes over $2f+1$ \EOTS signatures on its consumer block hash with context equal to height $H_{j+1}$ by the validator set of period $\tilde{e}$, and (iii) $\ckpt_{j+1}$ appears at a provider chain height less than $h_{\tilde{e}-1} + k_d$, where $k_d = 2 k_c(m \cdot \Tconfirm) + k_f$ is called the \emph{timestamp delay} (Line~\ref{line:isvalid}, Alg.~\ref{alg.bitcoin.forkchoice}).
The items (i) and (ii) above are checked for a timestamp by the function $\textsc{IsCor}(.)$ in Alg.~\ref{alg.bitcoin.forkchoice}.
Then;
\begin{enumerate}
    \item \textbf{Safe-stop Rule 1:} (Line~\ref{line:btc2}, Alg.~\ref{alg.bitcoin.forkchoice}) If (i) the timestamp $\ckpt_{j+1}$ is correct, and (ii) a block in $\chain_{j+1}$ is either unavailable or \invalid in $\client$'s view, then $\client$ stops going through the sequence $\ckpt_i$, $i \in [r]$, and outputs $\chain$ as its canonical consumer chain\footnote{Client $\client$ knows the correct validator set for all periods $e \leq \tilde{e}$. This is because every block in its current canonical consumer chain $\chain$ is available and \valid in its view. In particular, if $B$ is the last block of period $e$, $\client$ can infer the validator set of period $e+1$ from $\chain$.}.
    \item (Line~\ref{line:btc1}, Alg.~\ref{alg.bitcoin.forkchoice}) If (i) the timestamp $\ckpt_{j+1}$ is correct, (ii) every block in $\chain_{j+1}$ is available and \valid in $\client$'s view, (iii) the chain $\chain_{j+1}$ is of the height specified by $\ckpt_{j+1}$, and (iv) $\chain \preceq \chain_{j+1}$ (\ie, $\chain_{j+1}$ is consistent with the consumer chain output so far);
    \begin{itemize}
        \item \textbf{Safe-stop Rule 2:} (Line~\ref{line:btc4}, Alg.~\ref{alg.bitcoin.forkchoice}, Fig.~\ref{fig:safe-stop-1}) If $|\powchain^{\client}_t| \geq h_{\tilde{e}-1}+k_d+k_f$ and there is a correct timestamp at a provider chain height less than $h_{\tilde{e}-1}+k_d+k_f$ conflicting with $\chain_{j+1}$, then $\client$ stops going through the sequence $\ckpt_i$, $i \in [r]$, and outputs $\chain$ as its canonical consumer chain.
        \item \textbf{Update:} (Line~\ref{line:update}, Alg.~\ref{alg.bitcoin.forkchoice}) If $|\powchain^{\client}_t| < h_{\tilde{e}-1}+k_d+k_f$, or if $|\powchain^{\client}_t| \geq h_{\tilde{e}-1}+k_d+k_f$ and there is no correct timestamp at provider chain heights less than $h_{\tilde{e}-1}+k_d+k_f$ conflicting with $\chain_{j+1}$, then $\client$ sets $\chain_{j+1}$ as the new canonical chain ($\chain \leftarrow \chain_{j+1}$) and moves to $\ckpt_{j+2}$ as the next timestamp.
    \end{itemize}
    \item \textbf{Ignore:} If none of the cases above are satisfied, $\client$ ignores $\ckpt_{j+1}$ and moves to $\ckpt_{j+2}$ as the next timestamp.
\end{enumerate}

Unless one of the safe-stop rules is triggered, $\client$ processes all timestamps on its provider chain and identifies a timestamped canonical consumer chain ending at some block $B_\ell$ from period $e_\ell$.
Let $h_{\ell-1}$ denote the height of the highest confirmed provider block referred by the consumer blocks by the end of period $e_\ell-1$.
Then, starting at $B_\ell$, $\client$ complements the timestamped canonical chain by outputting a chain of available and \valid consumer blocks uniquely extending $B_\ell$, as long as the height of $\client$'s provider chain is less than $h_{\ell-1}+k_d$ (Line~\ref{line:heightcheck}, Alg.~\ref{alg.bitcoin.forkchoice}).
This ensures that the clients output consumer chain blocks as soon as they are finalized.

\subsection{Enforcing Slashing on the Provider Chain}

In certain cases, clients 
send timestamps to the provider chain in addition to the periodic timestamps for the last consumer block of every period, to warn other clients about a potential safety violation.
If there were indeed a safety violation, these extra timestamps ensure the extraction of the adversarial validators' secret keys, and thus slashing of their stake.
Let $\chain$ denote the canonical consumer chain in $\client$'s view and $e$ denote the period of the last block in $\chain$. 
Let $h$ denote the height of the highest provider block in $\client$'s provider chain, among those referred by the consumer blocks (all of which are \valid by definition) in $\chain$ from periods $1, \ldots, e-1$.
Then, if any of the following happens, $\client$ sends a timestamp to the provider chain for \emph{all} of the consumer blocks within $\chain$ that follow the last consumer block in $\chain$ with a correct timestamp on the provider chain \emph{before or at block} $h' - k_d - k_f$, where $h'$ denotes the height of $\client$'s provider chain: 
\begin{enumerate}
    \item Client $\client$ decides to go offline. (In this case, it must notify other clients about its view of the finalized consumer blocks; otherwise economic security is impossible.)
    \item The safe-stop rule 1 is triggered for client $\client$.
    \item The safe-stop rule 2 is triggered for client $\client$.
    \item Client $\client$ does not observe a correct timestamp on its provider chain for the last consumer block $B$ from period $e$ before the provider chain reaches height $h+k_d$. 
\end{enumerate}

Upon obtaining two quorums of $2f+1$ \EOTS signatures for the same height but two different block hashes, any client can extract the secret keys of $f+1$ validators using Alg.~\ref{alg:slashing-condition} and send this evidence to the respective bond contracts to slash their stake.

\subsection{Mismatched Timestamp Attack and Safe-stop Rule 2}
\label{sec:stopping-rule-2}

\begin{figure*}
{
    \centering
    \includegraphics[width = 0.9\linewidth]{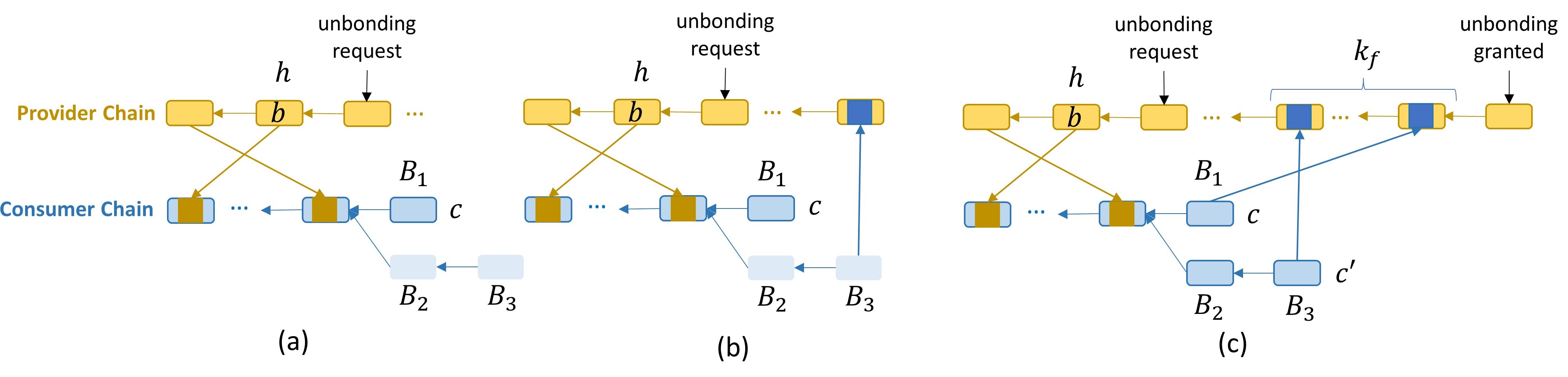}
    \caption{Illustration of the mismatched timestamp attack and the safe-stop rule 2 (\cf Section~\ref{sec:stopping-rule-2}). 
    }
\label{fig:safe-stop-2}
}
\end{figure*}

In a mismatched timestamp attack, the adversary exploits the fact that we cannot always extract the adversarial validators' keys by only inspecting the timestamps on Bitcoin (as opposed to observing the conflicting consumer blocks).
Suppose over $2/3$ of the validators specified by $b$ are adversarial and create three conflicting consumer blocks, $B_1$, $B_2$ and $B_3$.
Here, the adversary reveals $B_1$ to a client $\client$, but keeps $B_2$ and $B_3$ private (Fig.~\ref{fig:safe-stop-2}-a).
However, $B_3$ is timestamped before $B_1$ on Bitcoin (Fig.~\ref{fig:safe-stop-2}-b).
Upon seeing $B_3$'s timestamp, whose block $B_3$ is either unavailable, or conflicting with $B_1$, $\client$ sends a timestamp of $B_1$ to Bitcoin to notify the future clients about a potential safety violation.
$B_1$'s timestamp appears on Bitcoin within a few blocks of the timestamp for $B_3$.
Now, depending on the design of the protocol, there are two possibilities:

\noindent
\textbf{1) Frequent timestamps:} If every consumer block is to be timestamped in order, the adversary must first send a timestamp of $B_2$ to Bitcoin before $B_3$'s timestamp is accepted by the clients.
Then, by sending a timestamp of $B_1$, $\client$ would have notified the clients about the adversarial validators that have finalized conflicting blocks with their signatures.
Therefore, these validators can be slashed before they unbond.
However, this solution requires frequent timestamping, which might not be suitable for Bitcoin.
    
\noindent
\textbf{2) Rare timestamps:} 
In this case, the adversary can directly timestamp $B_3$ without sending a timestamp of $B_2$ or any of the blocks in $B_3$'s prefix, which are kept hidden from $\client$.
Then, the clients cannot necessarily expose the adversary's secret keys without observing the \EOTS signatures on $B_2$.
Thus, the adversarial validators will be allowed to unbond.
Finally, suppose a late-coming client $\client'$ observes the system after the adversary unbonded its stake (Fig.~\ref{fig:safe-stop-2}-c).
At this point, the adversary shows the previously unavailable blocks $B_2$ and $B_3$ to $\client'$.
If $\client'$ decides to output $B_2$ and $B_3$ instead of $B_1$ as chains with earlier timestamps take precedence, it would cause a non-slashable safety violation by conflicting with $\client$ that has output $B_1$.
Therefore, to prevent such safety violations, $\client'$ does not output consumer blocks whose timestamp conflicts with another timestamp within vicinity (\eg, $k_f$ blocks) of the original timestamp on Bitcoin (safe-stop rule 2).

\section{Proof of Theorem~\ref{thm:safety-smart-contracts}}
\label{sec:proof-of-thm-safety}

\begin{proof}[Proof of Theorem~\ref{thm:safety-smart-contracts}]
Suppose there are two clients $\client_1$ and $\client_2$ such that the canonical consumer chains $\chain^{\client_1}_{t_1}$ and $\chain^{\client_2}_{t_2}$ are not consistent.
Let $B_1$ denote the consumer block with the smallest height in $\chain^{\client_1}_{t_1}$ among those conflicting with $\chain^{\client_2}_{t_2}$.
Similarly, let $B_2$ denote the block with the smallest height in $\chain^{\client_2}_{t_2}$ among those conflicting with $\chain^{\client_1}_{t_1}$.
Let $B_0$ denote the common parent of $B_1$ and $B_2$.

Suppose $\client_1$ first outputted $B_1$ at some time $t_a$, and $\client_2$ first outputted $B_2$ at some time $t_b$ as part of its canonical consumer chain.
The validator set for the height of $B_1$ and $B_2$ is determined by the unique highest provider block $b \in \powchain^{\client_1}_{t_a}, \powchain^{\client_2}_{t_b}$, with height $h$, among those referred by the consumer blocks ending at the largest completed period at or before $B_0$ (this block is unique by the security of the provider chain).
Next, we consider the following cases:

\noindent
\textbf{Case A:} $|\powchain^{\client_1}_{t_a}| \geq h + k_d + k_f$ and $|\powchain^{\client_2}_{t_b}| \geq h + k_d + k_f$. 
Then, both $\client_1$ and $\client_2$ must have respectively output $B_1$ and $B_2$ at Line~\ref{line:update}, Alg.~\ref{alg.bitcoin.forkchoice} upon observing the correct timestamps $\ckpt_1, \ckpt_2 \in \powchain^{\client_1}_{t_a}, \powchain^{\client_2}_{t_b}$ at heights less than $h+k_d$.
Without loss of generality, suppose $\ckpt_1$ appears in the prefix of $\ckpt_2$.
In this case, if every block in the consumer chain determined by $\ckpt_1$ is available and \valid in $\client_2$'s view at time $t_b$, then $\client_2$ would also output $B_1$ upon observing $\ckpt_1$.
Thus, there must be a block within the consumer chain determined by $\ckpt_1$ that is unavailable or \invalid in $\client_2$'s view at time $t_b$.
However, in this case, the safe-stop rule 1 is triggered for $\client_2$ upon observing $\ckpt_1$, and it does not output $B_2$ (Line~\ref{line:btc2}, Alg.~\ref{alg.bitcoin.forkchoice}).
Hence, case A cannot happen.

\noindent
\textbf{Case B:} $|\powchain^{\client_1}_{t_a}| < h + k_d + k_f$ and $|\powchain^{\client_2}_{t_b}| < h + k_d + k_f$.
Then, one of the following cases must have happened:
\begin{itemize}
    \item \textbf{Case 1:} Safe-stop rule 1 is triggered for $\client_1$ before its provider chain reaches height $h + k_d$. 
    \item \textbf{Case 2:} Client $\client_1$ decides to go offline before its provider chain reaches height $h + k_d$. 
    \item \textbf{Case 3:} Neither of the cases 1 and 2 happen until $\client_1$'s provider chain reaches height $h + k_d$. However, $\client_1$ does not observe any correct timestamp for $B_1$ or its descendants by the time its provider chain reaches height $h + k_d$. 
    \item \textbf{Case 4:} Neither of the cases 1 and 2 happen until $\client_1$'s provider chain reaches height $h + k_d$. Client $\client_1$ observes a correct timestamp $\ckpt_1$ for $B_1$ or its descendants on its provider chain at a height less than $h + k_d$, and every block timestamped by $\ckpt_1$ is available and \valid in $\client_1$'s view.
\end{itemize}

\begin{itemize}
    \item \textbf{Case I:} Safe-stop rule 1 is triggered for $\client_2$ before its provider chain reaches height $h + k_d$. 
    \item \textbf{Case II:} Client $\client_2$ decides to go offline before its provider chain reaches height $h + k_d$. 
    \item \textbf{Case III:} Neither of the cases I and II happen until $\client_2$'s provider chain reaches height $h + k_d$. However, $\client_2$ does not observe any correct timestamp for $B_2$ or its descendants by the time its provider chain reaches height $h + k_d$. 
    \item \textbf{Case IV:} Neither of the cases I and II happen until $\client_2$'s provider chain reaches height $h + k_d$. Client $\client_2$ observes a correct timestamp $\ckpt_2$ for $B_2$ or its descendants on its provider chain at a height less than $h + k_d$, and every block timestamped by $\ckpt_2$ is available and \valid in $\client_2$'s view.
\end{itemize}

If cases 1, 2 or 3 (I, II or III) happen, $\client_1$ ($\client_2$) sends timestamps to the provider chain for \emph{all} of the blocks within its canonical consumer chain that follow the last consumer block with a correct timestamp on the provider chain at least before $h$, \ie, at least all blocks following $B_0$. 

Then, we can deduce the following:
\begin{itemize}    
\item \textbf{(1 and I), (1 and II), (1 and III), (2 and I), (2 and II), (2 and III), (3 and I), (3 and II), (3 and III):} In these cases, all online clients learn about the conflicting blocks $B_1$ and $B_2$, before the provider chain in its view reaches height $h + k_d + 2k_f$.
\item \textbf{(4 and I), (4 and II), (4 and III):} Either $\client_1$ learns about the conflicting blocks $B_1$ and $B_2$, or goes offline, before the provider chain reaches height $h + k_d + k_f$ in its view. 
In the latter case, $\client_1$ sends timestamps to the provider chain for all of its consumer blocks following $B_0$, and an online client learns about the conflicting blocks $B_1$ and $B_2$ before the provider chain reaches height $h + k_d + 2k_f$ in its view.
\item \textbf{(1 and IV), (2 and IV), (3 and IV):} This is the same as the cases above, with the roles of $\client_1$ and $\client_2$ reversed.
\item \textbf{(4 and IV):} 
Without loss of generality, assume that $\ckpt_2$ appears earlier than $\ckpt_1$ on the provider chain.
Suppose at time $t_a$, $\client_1$'s provider chain has height $\leq h+k_d$.
Then, $\client_1$ observes $\ckpt_2$ on its provider chain by the time it reaches height $h+k_d$.
In this case, either a block in the consumer chain determined by $\ckpt_2$
is unavailable or invalid in $\client_1$'s view, in which case safe-stop rule 1 is triggered for $\client_1$.
Or, $\client_1$ observes a correct timestamp ($\ckpt_2$) on its provider chain before height $h+k_d$, and the consumer chain at its preimage conflicts with $B_1$.
Thus, $\client_1$ either learns about the conflicting blocks, or sends timestamps for all of its consumer blocks following $B_0$ to the provider chain.

Now, suppose at time $t_a$, $\client_1$'s provider chain has height $>h+k_d$.
Now, if every block in the consumer chain determined by $\ckpt_2$ is available and \valid in $\client_1$'s view at that time, then $\client_1$ would also output $B_2$ upon observing $\ckpt_2$ by time $t_a$.
Thus, there must be a block within the consumer chain determined by $\ckpt_2$ that is unavailable or \invalid in $\client_1$'s view at that time.
However, in this case, the safe-stop rule 1 is triggered for $\client_1$ upon observing $\ckpt_2$, and it does not output $B_1$ (Line~\ref{line:btc2}, Alg.~\ref{alg.bitcoin.forkchoice}).
Hence, this case cannot happen.

Finally, by a symmetric argument, if at time $t_b$, $\client_2$'s provider chain has height $\leq h+k_d$, then $\client_2$ either learns about the conflicting blocks, or sends timestamps for all of its consumer blocks following $B_0$ to the provider chain by the time its provider chain reaches height $h+k_d$.
In this case, all online clients learn about the conflicting blocks $B_1$
and $B_2$, before the provider chain
reaches height $h + k_d + k_f$ in any view.
If at time $t_b$, $\client_2$'s provider chain has height $>h+k_d$, then either the safe-stop rule 1 is triggered for $\client_2$ upon observing $\ckpt_1$, or it learns about the conflicting blocks $B_1$
and $B_2$ at time $t_b$.
In either case, an online client learns about the conflicting blocks $B_1$ and $B_2$, before the provider chain reaches height $h + k_d + 2k_f$ in any view.
\end{itemize}

\noindent
\textbf{Case C:} $|\powchain^{\client_1}_{t_a}| < h + k_d + k_f$ and $|\powchain^{\client_2}_{t_b}| \geq h + k_d + k_f$.
In this case, $\client_2$ must have output $B_2$ at Line~\ref{line:update}, Alg.~\ref{alg.bitcoin.forkchoice} upon observing a timestamp $\ckpt_2 \in \powchain^{\client_2}_{t_b}$ at a height less than $h+k_d$.
Then, depending on which of the four cases 1-2-3-4 is true for $\client_1$, we investigate the following events:

\begin{itemize}    
\item \textbf{1-2 and $|\powchain^{\client_2}_{t_b}| \geq h + k_d + k_f$:} In this case, $\client_2$ observes a timestamp on its provider chain, before height $h+k_d+k_f$, that conflicts with $\ckpt_2$.
Then, $\client_2$ does not output $B_2$ upon observing the timestamp $\ckpt_2 \in \powchain^{\client_2}_{t_b}$ due to the safe-stop rule 2 (Line~\ref{line:btc4}, Alg.~\ref{alg.bitcoin.forkchoice}).
Hence, this case cannot happen.
\item \textbf{3-4 and $|\powchain^{\client_2}_{t_b}| \geq h + k_d + k_f$:}
Suppose at time $t_a$, $\client_1$'s provider chain has height less than $h+k_d$.
Then, $\client_1$ observes $\ckpt_2$ on its provider chain by the time it reaches height $h+k_d$.
In this case, either a block in the consumer chain determined by $\ckpt_2$
is unavailable or invalid in $\client_1$'s view, in which case safe-stop rule 1 is triggered for $\client_1$.
Or, $\client_1$ observes a correct timestamp ($\ckpt_2$) on its provider chain before height $h+k_d$, and the consumer chain at its preimage conflicts with $B_1$.
Therefore, $\client_1$ either learns about the conflicting blocks, or sends timestamps for all of its consumer blocks following $B_0$ to the provider chain before the provider chain reaches height $h + k_d$ in its view.
In the latter case, $\client_2$ does not output $B_2$ upon observing the timestamp $\ckpt_2 \in \powchain^{\client_2}_{t_b}$ due to the safe-stop rule 2 (Line~\ref{line:btc4}, Alg.~\ref{alg.bitcoin.forkchoice}).
Hence, this case cannot happen.

Now, suppose at time $t_a$, $\client_1$'s provider chain has height $h+k_d$ or more.
Now, if every block in the consumer chain determined by $\ckpt_2$ is available and \valid in $\client_1$'s view at that time, then $\client_1$ learns about the conflicting blocks at time $t_a$.
On the other hand, if there is a block within the consumer chain determined by $\ckpt_2$ that is unavailable or \invalid in $\client_1$'s view at time $t_a$, then the safe-stop rule 1 is triggered for $\client_1$ upon observing $\ckpt_2$, and it sends timestamps for all of its consumer blocks following $B_0$ to the provider chain before the provider chain reaches height $h + k_d + k_f$ in its view.
In either case, an online client learns about the conflicting blocks $B_1$ and $B_2$, before the provider chain reaches height $h + k_d + 2k_f$ in any view.
\end{itemize}

\noindent
\textbf{Case D:} $|\powchain^{\client_1}_{t_a}| \geq h + k_d + k_f$ and $|\powchain^{\client_2}_{t_b}| < h + k_d + k_f$. 
This is the same as case C, with the roles of $\client_1$ and $\client_2$ reversed.

Finally, we observe that in all possible cases, the watchtower $\client$ learns about the conflicting blocks $B_1$ and $B_2$ at the same height $h'$, along with
two quorums of $2f+1$ height $h'$ \EOTS signatures $\langle \mathsf{Final}, h', id(B_1) \rangle$ and $\langle \mathsf{Final}, h', id(B_2) \rangle$ for these blocks, both before the provider chain reaches height $h+k_d+2k_f$.
Upon obtaining the two quorums or the evidence, the watchtower invokes the forensic protocol, which identifies $f+1$ adversarial validators as protocol violators either by the accountable safety of the consumer chain, or by intersecting the two \EOTS signature quorums as they have satisfied the condition in Alg.~\ref{alg:slashing-condition}.
In the latter case, by \textsf{EXT-SCMA} security of EOTS (Def.~\ref{EOTS:sec:ext-scma}), the forensic protocol can extract their secret signing keys (w.o.p.), before the provider chain reaches height $h+k_d+2k_f$ in $\client$'s view.
Then, in either case, the watchtower $\client$ uses the extracted secret keys to construct and broadcast a slashing transaction to the bond contract, which is confirmed in the provider chain before it reaches height $h+k_d+3k_f = h+k_u$ in the view of any client.
Before this height, the bond contract restricts every valid spend to a burn-address output---enforced on-chain by \texttt{OP\_CHECKTEMPLATEVERIFY} in the CTV variant of \Cref{alg.bond.contract}, or by the deleted-key covenant committee under existential honesty in the variant of \Cref{alg.bond.contract.musig}.
Hence, even an adversarial race-spend with the leaked key can only burn the stake.
Since none of the $f+1$ identified validators can recover their stake to a self-controlled address before height $h+k_u$ (due to the timelock), and the slashing transaction (or an equivalent race-spend) confirms before that height, $f+1$ adversarial validators get slashed. 

Since honest validators send at most one \EOTS signature per height, for any honest validator, given the set $Q$ of message, height, signature tuples returned by the validator, $\forall (h, B, B')$ such that $ (id(B), h, .) \in Q \land (id(B'), h, .) \in Q$, it holds that $id(B) = id(B')$.
Thus, by Defs.~\ref{def:accountable-safety} and~\ref{eots:sec:seuf-cmadc}, no honest validator's stake can be slashed.
Therefore, the Bitcoin staking protocol satisfies $(f+1)$-economic safety.
\end{proof}

\section{Proof of Theorem~\ref{thm:liveness-smart-contracts}}
\label{sec:proof-of-thm-liveness}

For the liveness result, we assume a synchronous or a partially synchronous network, where GST is sufficiently bounded.
This is because an arbitrarily large GST prevents the liveness of the underlying consumer chain for long intervals, during which the validators assigned to the pending period of $m$ consumer blocks can unbond on the provider chain before their period is completed (which can only happen after GST).
We also assume the consumer-chain period length $m$ satisfies $m \geq f+1$, ensuring that each period contains at least one honest proposer under Tendermint's round-robin proposer selection (since at most $f$ adversaries can consecutively hold proposer slots).

\begin{proof}[Proof of Theorem~\ref{thm:liveness-smart-contracts}]

We prove the theorem by induction on the periods of consumer blocks.
Let $h$ denote the height of the provider block $b_0$ referred by the genesis consumer block $B_0$.
Over $2f+1$ validators within the initial validator set $S_0$ are honest.

\textbf{Induction Hypothesis:} Only a single \valid consumer block can gather $2f+1$ \EOTS signatures at any height of period $e$.
Safe-stop rules cannot be triggered for any client by the timestamps from periods $1, \ldots, e$.
By the time all relevant honest validators have entered period $e$, the highest provider block $b_{e-1}$ (at height $h_{e-1}$) referred by the blocks within the past periods $1, \ldots, e-1$ is at most $k_c(m \cdot \Tconfirm)$ deep in the provider chain of any client.
Correct timestamps of the available and \valid consumer blocks from period $e$ appear on the provider chain by height $h_{e-1} + k_d$.

\textbf{Base step:} Only a single \valid consumer block can gather $2f+1$ \EOTS signatures at any consumer chain height of the first period.
Therefore, all consumer blocks of period $e=1$ become confirmed and gather $2f+1$ \EOTS signatures within $\Theta(m \cdot \Tconfirm)$ time by the security of Tendermint (\cite[Lemmas 3, 4, 7]{2018tendermint}), during which the provider chain advances by less than $k_c(m \cdot \Tconfirm)$ blocks in the view of any client.
Thus, by the time all relevant honest validators have entered period $2$, the highest provider block $b_1$ referred by the consumer blocks of the first period is at most $k_c(m \cdot \Tconfirm)$ deep in the provider chain of any client (and no validator of period $1$ could have unbonded during period 1).
Furthermore, a timestamp of the blocks in the first period appears on all provider chains before height $h + k_c(m \cdot \Tconfirm) + k_f < h + k_d$, and all blocks attested by the timestamps of the first period are available, \valid and consistent, implying that the clients keep outputting confirmed consumer blocks and that the safe-stop rules 1 and 2 cannot be triggered for any client.

\textbf{Inductive step:} Suppose the induction hypothesis holds for all periods $1, \ldots, e-1$.
At least $2f+1$ of the validators assigned to period $e$ are honest by assumption.
By the induction hypothesis and the honesty assumption, only a single \valid consumer block can become confirmed and gather $2f+1$ \EOTS signatures at any height of period $e$.
For the same reason, all blocks attested
by the periods $e$ timestamps are available, valid and consistent, implying that the clients keep outputting
confirmed consumer blocks and that the safe-stop rules 1 and 2 cannot be triggered for any client by a period $e$ timestamp.

By the induction hypothesis, by the time all relevant honest validators have entered period $e$, the highest provider block $b_{e-1}$ (at height $h_{e-1}$) referred by the blocks within the past periods $1, \ldots, e-1$ is at most $k_c(m \cdot \Tconfirm)$ deep in the provider chain of any client.
All consumer blocks of period $e$ become confirmed (and gather $2f+1$ \EOTS signatures) within $\Theta(m \cdot \Tconfirm)$ time by the security of Tendermint (\cite[Lemmas 3, 4, 7]{2018tendermint}), during which the provider chain advances less than $k_c(m \cdot \Tconfirm)$ blocks in the view of any client.
Therefore, as $2 k_c(m \cdot \Tconfirm) < k_u$, no period $e$ validator could have unbonded during this time, and
by the time all relevant honest validators have entered period $e+1$, the highest provider block $b_e$ referred by the consumer blocks of periods $1, \ldots, e$ is at most $k_c(m \cdot \Tconfirm)$ deep in the provider chain of any client.
Furthermore, a timestamp of the blocks in period $e$ appears on the provider chain of all clients before height $h_{e-1} + k_f + 2 k_c(m \cdot \Tconfirm) \leq h + k_d$.

Finally, since there is an honest block among the $m$ finalized blocks of any periods w.o.p. and by the induction hypothesis,
liveness is satisfied w.o.p.
\end{proof}

Note that when the number of adversarial validators in any window of provider blocks is $\leq f$ and the clients remain online, validators do not send extra timestamps to the provider chain as the cases 1-2-3-4 used to enforce slashing on the provider chain are never triggered.

\section{Lack of Accountable Safety in Tendermint}
\label{sec:appendix-lack-of-acc-safety}
In this section, we show that Tendermint as it stands actually lacks accountable safety.
However, with our finality gadget, it can be made EOTS, and by implication, accountably-safe, thus, can be used as part of the remote staking protocol.
\subsection{Locking in Tendermint}
Each honest validator maintains four variables throughout the protocol execution: $\mathsf{lockedValue}$, $\mathsf{lockedRound}$, $\mathsf{validValue}$ and $\mathsf{validRound}$.
The $\mathsf{lockedValue}$ denotes the most recent block, \ie the one with the largest round, for which the validator sent a precommit, and $\mathsf{lockedRound}$ denotes the round of this precommit.
Similarly, $\mathsf{validValue}$ denotes the most recent block for which the validator has observed $2f+1$ prevotes by distinct validators, and $\mathsf{validRound}$ denotes this round.
If a validator has received a proposal $\langle \mathsf{Proposal}, H, r, v, vr \rangle$ from the round leader before entering the $\mathsf{Prevote}$ step and its $\mathsf{lockedRound} = -1$, \ie it is not locked on any block, it sends a prevote for the proposed block.
Otherwise, if its $\mathsf{lockedRound} > -1$, \ie it is locked on a block $\mathsf{lockedValue}$, the validator checks if either $v$ is the same as its $\mathsf{lockedValue}$ \textbf{(voting rule 1)} or if it has observed $2f+1$ round $vr$ prevotes $\langle \mathsf{Prevote}, H, vr, id(v) \rangle$ for $v$, such that $vr > \mathsf{lockedRound}$ \textbf{(voting rule 2)}.
If either of the voting rules is satisfied, it sends a prevote for the proposed block.
Otherwise, it sends a prevote with the \emph{nil} value.

\subsection{Accountable Safety for Tendermint}
\label{sec:acc-safety-rules}

If two clients finalize conflicting blocks $B$ and $B'$ at the same round, then they can identify the adversarial validators that sent precommits for both blocks by inspecting the $2f+1$ precommits for $B$ and $B'$.
However, when the conflicting blocks are finalized at different rounds $r$ and $r'>r$, they cannot use the quorum intersection argument directly on the two precommit quorums.
To understand this, consider an honest validator that sent a precommit for $B$ at round $r$.
Even though the validator locked on $B$ at round $r$ and set its $\mathsf{lockedValue} = B$ and $\mathsf{lockedRound} = r$, it might have observed a quorum of $2f+1$ prevotes for block $B'$ at a later round $r^*>r$.
In this case, upon observing the proposal $\langle \mathsf{Proposal}, H, r', B', r^* \rangle$, the honest validator would send a prevote for block $B'$ by \textbf{voting rule 2}, after which it could send a precommit.
Then, the naive intersection argument between the precommit quorums would identify this honest validator as adversarial, which violates accountable safety.

To find the validators culpable for the safety violation in the example above, we consider the first round $r^*$ such that a collection of $2f+1$ prevotes from round $r^*$, \ie, $\langle \mathsf{Prevote}, H, r^*, id(B') \rangle$, is formed for block $B'$.
The set of validators that sent these prevotes constitute the potential set of adversarial validators.
Suppose these validators broadcast prevotes for some proposal $\langle \mathsf{Proposal}, H, r^*, B', v, vr \rangle$.
Now, since $r^*$ is the first round greater than $r$, where a quorum of $2f+1$ prevotes is formed for $B'$, no validator could have observed a quorum for $B'$ at any round $vr \in (r, r^*)$.
Thus, the validators that were locked on $B$ at round $r$ should not have sent prevotes for $B'$ as none of the voting rules could have been satisfied in their views.
Sending a prevote in such circumstances is called the \emph{amnesia} attack since the adversarial validators \emph{forget} that they had an earlier lock on $B$ (\cf~\cite{tendermintacc}).
Consequently, to determine the set of adversarial validators, clients must find the intersection of the validator sets that have sent the $2f+1$ precommits for block $B$ at round $r$ and the $2f+1$ prevotes for $B'$ at round $r^*$.

\subsection{Lack of Accountable Safety under Partial Synchrony}
Unfortunately, the current version of Tendermint \cite{2018tendermint} does not allow clients to generate a proof of protocol violation in the case of an amnesia attack.
This is due to the indistinguishability of two worlds with different sets of adversarial validators under partial synchrony.

Consider a client that aims to identify the culpable validators in the attack above (by calling the forensic protocol), after collecting transcripts and observing the quorum of $2f+1$ round $r^*>r$ prevotes for $B'$.
For this purpose, the protocol must ascertain that $r^*$ is the earliest round, where a quorum of $2f+1$ prevotes was formed for block $B'$.
In world 1, this is indeed the case.
Then, the protocol can identify the validators that sent both a round $r^*$ prevote and a round $r$ precommit for $B$ as adversarial,
since there is no set of $2f+1$ prevotes for $B'$ from any round $r'' < r^*$ that could have prompted these validators to release their locks on $B$.
However, in world 2, there is a round $r''<r^*$, in which the adversarial validators sent a quorum of $2f+1$ round $r''$ prevotes for $B'$ to the honest validators.
No client (other than the honest validators) receives these prevotes for block $B'$ due to partial synchrony.
Thus, for the clients and the forensic protocol invoked by them, the two worlds are \emph{indistinguishable}.
Then, the adversary can convince the clients that an honest validator is a protocol violator by giving them the same proof output by the forensic protocol in world 1, which contradicts accountable safety.
\begin{theorem}
\label{thm:tendermint-not-accountable-psync}
Tendermint protocol does not provide accountable safety with resilience greater than one validator under a partially synchronous network.
\end{theorem}

\begin{proof}
Towards contradiction, suppose Tendermint provides accountable safety with resilience of greater than one validator.
Consider rounds $r=0,1,2$ and $3$ of some height $H$ before GST.
There are $3f+1$ validators.
Let $P$, $Q$ and $R$ denote disjoint sets of $f$ validators each.
Let $x$ denote the remaining validator.
We next consider the following two worlds.

\noindent
\textbf{World 1:}
Validators in $R$ and $x$ are adversarial and the rest are honest.

\emph{Round 0:}
At round $0$, the adversary delivers only the messages among the validators in $P \cup R \cup x$.
A new block $B$ is proposed at round $0$, and gathers $2f+1$ prevotes and precommits from the validators in $P \cup R \cup x$.
However, the honest validators in $P$ do not observe the precommits by those in $R$.
Thus, even though they lock on $B$, they do not decide $B$.

\emph{Round 1:}
At round $1$, the adversary delivers only the messages among the validators in $Q \cup R \cup x$.
A new block $B'$ is proposed by an honest validator in $Q$ and gathers $f$ round $1$ prevotes from the validators in $Q$. 
The adversarial validators in $R \cup x$ do not send round $1$ prevotes for $B'$.
Hence, the honest validators in $Q$ send precommits with the \emph{nil} value at round $1$, and $B'$ cannot be decided by the round $1$ prevotes and precommits.

\emph{Round 2:}
At round $2$, the adversary again delivers only the messages among the validators in $Q \cup R \cup x$.
The adversarial leader $x$ sends the proposal $\langle \mathsf{PROPOSAL}, H, r=2, B', vr=-1 \rangle$.
The block $B'$ gathers $2f+1$ round $2$ prevotes $\langle \mathsf{PREVOTE}, H, r=2, id(B') \rangle$ from the validators in $Q \cup R \cup x$; however, the adversarial validators in $R \cup x$ do not show their prevotes to the honest validators in $Q$.
Hence, the honest validators in $Q$ send precommits with the \emph{nil} value at round $2$, and $B'$ cannot be decided by the round $2$ prevotes and precommits.

\emph{Round 3:}
Finally, at round $3$, the adversary delivers only the messages among the validators in $P \cup R \cup x$.
An adversarial validator sends the proposal $\langle \mathsf{PROPOSAL}, H, r=3, B', vr=2 \rangle$, and the adversary delivers the $2f+1$ round $2$ prevotes for $B'$ to the honest validators in $P$.
Hence, the validators in $P$ unlock from $B$ and, along with the adversarial validators in $R \cup x$, send prevotes and precommits for $B'$.

\noindent
\textbf{Clients in World 1:}
A client $\client_1$ decides $B$ at the end of round $0$ upon observing the round $0$ prevotes and precommits for $B$ by the validators in $P \cup R \cup x$.
A different client $\client_2$ decides $B'$ at the end of round $3$ upon observing the messages sent by the validators in rounds $1$, $2$ and $3$.
Since Tendermint is accountably-safe with a resilience of greater than one validator, upon collecting the messages received by the clients, the forensic protocol outputs at least one validator from the set $R$ (otherwise, it must have identified an honest validator which would imply a contradiction).

\noindent
\textbf{World 2:}
Validators in $P$ and $x$ are adversarial and the rest are honest.

\emph{Round 0:}
At round $0$, the adversary delivers only the messages among the validators in $P \cup R \cup x$.
A new block $B$ is proposed at round $0$, and gathers $2f+1$ prevotes and precommits from the validators in $P \cup R \cup x$.
However, the honest validators in $R$ do not observe the precommits by those in $P$.
Thus, even though they lock on $B$, they do not decide $B$.

\emph{Round 1:}
At round $1$, the adversary delivers only the messages among the validators in $P \cup Q \cup x$.
A new block $B'$ is proposed by an honest validator in $Q$.
The block $B'$ gathers $2f+1$ round $1$ prevotes from the validators in $P \cup Q \cup x$; however, the adversarial validators in $P \cup x$ do not show their prevotes to the honest validators in $Q$.
Hence, the honest validators in $Q$ send precommits with the \emph{nil} value at round $1$, and $B'$ cannot be decided by the round $1$ prevotes and precommits.

\emph{Round 2:}
At round $2$, the adversary delivers only the messages among the validators in $Q \cup R \cup x$.
The adversarial leader $x$ sends two proposals: $\langle \mathsf{PROPOSAL}, H, r=2, B', vr=-1 \rangle$ to the validators in $Q$ and $\langle \mathsf{PROPOSAL}, H, r=2, B', vr=1 \rangle$ to the validators in $R$.
It also shows the $2f+1$ round $1$ prevotes for $B'$ to the validators in $R$.
Consequently, the block $B'$ gathers $2f+1$ round $2$ prevotes $\langle \mathsf{PREVOTE}, H, r=2, id(B') \rangle$ from the validators in $Q \cup R \cup x$; however, the adversarial validator $x$ does not show its prevote to the honest validators in $Q \cup R$.
Hence, the honest validators in $Q \cup R$ send precommits with the \emph{nil} value at round $2$, and $B'$ cannot be decided by the round $2$ prevotes and precommits.

\emph{Round 3:}
Finally, at round $3$, the adversary delivers only the messages among the validators in $P \cup R \cup x$.
An adversarial validator sends the proposal $\langle \mathsf{PROPOSAL}, H, r=3, B', vr=2 \rangle$, and the adversary delivers the $2f+1$ round $2$ prevotes for $B'$ to the honest validators in $R$.
Hence, all validators in $P \cup R \cup x$, send prevotes and precommits for $B'$.

\noindent
\textbf{Clients in World 2:}
A client $\client_1$ decides $B$ at the end of round $0$ upon observing the round $0$ prevotes and precommits for $B$ by the validators in $P \cup R \cup x$.
A different client $\client_2$ decides $B'$ at the end of round $3$ upon observing all round $1$, $2$ and $3$ messages, except the round $1$ prevotes by the validators in $P \cup x$ and the round $2$ proposal $\langle \mathsf{PROPOSAL}, H, r=2, B', vr=1 \rangle$ by $x$.
The adversarial validators in $P \cup x$ send the same messages  to the forensic protocol as they do in world 1.
Hence, the forensic protocol receives the same set of messages as in world 1 and identifies $x$ and the same subset of the validators in $R$ as in world 1 as protocol violators with overwhelming probability.
Since the validators in $R$ are honest in world 2, this is a contradiction with the definition of accountable safety.
\end{proof}

In Tendermint, proposals do not include the $2f+1$ prevotes that justify the leader's $\mathsf{validValue}$.
The protocol instead expects the validators to receive these prevotes from the network, which happens in a timely manner under synchrony.
However, Theorem~\ref{thm:tendermint-not-accountable-psync} holds even if these prevotes are broadcast alongside the proposals (as in HotStuff); since its proof already assumes that the clients expect to see the round $vr$ prevotes that justify a proposal $\langle \mathsf{Proposal}, H, r=2, B', vr \rangle$ before considering the proposal itself.

\subsection{Lack of Accountable Safety under Synchrony}

If the network is known to become synchronous when the forensic protocol is invoked, then the protocol can distinguish the two worlds above with different sets of honest validators by querying the honest validators and learning about the $2f+1$ prevotes from round $r''$ in world 2.
However, this is not sufficient to provide accountable safety, which requires the forensic protocol to generate a transferable proof of protocol violation.
As it is not possible to create a \emph{proof of absence}, each client must check for themselves which world they are in, \ie, they must verify the absence or presence of the $2f+1$ prevotes from some round $r'' < r^*$ by communicating with the honest validators.
This observation is formalized by the following theorem:
\begin{theorem}
\label{thm:tendermint-not-accountable-sync}
Tendermint protocol does not provide accountable safety with resilience greater than one validator, even if the network is known to become synchronous when the forensic protocol is invoked.
\end{theorem}

If the network were known to become synchronous when the forensic protocol is invoked, the forensic protocol would receive the $2f+1$ round $1$ prevotes by the validators in $P \cup Q \cup x$ from the honest validators in $R$ (who observed these round $1$ prevotes in round $2$) and identify those in $P$ as protocol violators in world 2.

\begin{proof}
Towards contradiction, suppose Tendermint provides accountable safety with resilience greater than one validator.
At the invocation of the forensic protocol, the network has become synchronous. 
We next construct the following two worlds inspired by the proof of Theorem~\ref{thm:tendermint-not-accountable-psync}:

\noindent
\textbf{World 1:} This is the same as world 1 described by the proof of Theorem~\ref{thm:tendermint-not-accountable-psync}.
The forensic protocol does not receive any round $1$ messages from the validators in $P$ and generates a proof that irrefutably identifies a validator in $R$ as a protocol violator.

\noindent
\textbf{World 2:} This is the same as world 2 described by the proof of Theorem~\ref{thm:tendermint-not-accountable-psync}, except that since the network has become synchronous, the forensic protocol has also received the round $1$ prevotes by the validators in $P \cup x$ and the round $2$ proposal $\langle \mathsf{PROPOSAL}, H, r=2, B', vr=1 \rangle$.
Thus, the set of messages received by the forensic protocol in world 2 is a superset of the messages received in world 2 of the proof of Theorem~\ref{thm:tendermint-not-accountable-psync}, which is the same as the messages received in world 1.
This implies that given these messages, an adversarial client can generate the same proof as the one generated in world 1, which irrefutably identifies a validator in $R$ as a protocol violator.
However, since the validators in $R$ are honest in world 2, this is a contradiction with the definition of accountable safety.
\end{proof}

\subsection{Tendermint Made Accountably-safe}
\label{sec:tendermint-made-acc-safe}
Inspired by the HotStuff-view protocol in~\cite{DBLP:conf/ccs/ShengWNKV21}, we can change Tendermint so that each $\mathsf{Prevote}$ message includes the $\mathsf{validRound}$ number $vr$ within the proposal it supports. For instance, if a validator sends a prevote for the proposal $\langle \mathsf{Proposal}, H, r^*, B', vr \rangle$, then it includes $vr$ to its prevote as shown: $\langle \mathsf{Prevote}, H, 2, id(B'), vr \rangle$.
This small change suffices to make Tendermint accountably-safe and the proof of accountable safety proceeds similar to \cite[Theorem 5.1]{DBLP:conf/ccs/ShengWNKV21}.

\subsection{EOTS Finality Gadget Bypasses the Amnesia Obstruction}
\label{sec:eots-bypass-amnesia}

The negative results above (\Cref{thm:tendermint-not-accountable-psync,thm:tendermint-not-accountable-sync}) show that quorum-intersection arguments on the prevote/precommit transcript fail under amnesia attacks: an adversarial validator can equivocate across rounds at the Tendermint layer without producing two messages that a forensic protocol can use to identify it.
The EOTS finality gadget (\Cref{alg:validator-finality,alg:slashing-condition}) circumvents this obstruction by introducing an \emph{additional} signing round on a fresh per-height message $\langle \mathsf{Final}, H, \mathsf{hash}(B) \rangle$, signed with a separate EOTS key (independent of the Tendermint consensus key) under a per-height context $\pct_{H,\validator}$.
A block is considered finalized only if it gathers $2f+1$ EOTS-Final signatures at height $H$.
Hence, two conflicting finalizations at the same height require at least $f+1$ validators in the intersection of two such quorums to produce two EOTS signatures with the same $\pct$ on different messages---regardless of how those validators behaved at the prevote/precommit layer.
This is exactly the slashing condition of \Cref{alg:slashing-condition}, and \textsf{EXT-SCMA} security of EOTS (\Cref{EOTS:sec:ext-scma}) then extracts their secret keys.
\Cref{thm:tendermint-accountability-extraction} formalizes the resulting $(f+1)$-EOTS-safety; its proof appears in~\Cref{sec:appendix-proof-thm-daps-safety}.

\end{document}